\documentclass[11pt,onecolumn]{IEEEtran}

\IEEEoverridecommandlockouts

\makeatletter

\def\ps@headings{%
\def\@oddhead{\mbox{}\scriptsize\rightmark \hfil \thepage}%
\def\@evenhead{\scriptsize\thepage \hfil \leftmark\mbox{}}%
\def\@oddfoot{}%
\def\@evenfoot{}}

\makeatother

\usepackage{amsfonts,amsmath,amssymb}
\usepackage{color,subfigure,url, dsfont,epsf,
               graphics,graphicx,latexsym}

\newcommand{\PNN}{{\mathbb N}}  
\newcommand{\PRR}{{\mathbb R}} 

\newcommand{\DQ}[2]{D_{\ell,#1}(#2)}

\newcommand{\dd}{{\rm d}}
\newcommand{\oo}{{\rm o}}

\newcommand{\calC}{{\mathcal C}}
\newcommand{\EE}{{\mathbb E}}
\newcommand{\FF}{{\mathcal F}}
\newcommand{\calM}{{\mathcal M}}
\newcommand{\PP}{{\mathbb P}}
\newcommand{\WW}{{\mathcal W}}

\newcommand{\Load}{L}
\newcommand{\Qorig}{X}
\newcommand{\Lonenorm}[1]{\lvert \lvert #1 \rvert \rvert}  

\newcommand{\sumi}{\sum_{i = 1}^{N}}
\newcommand{\sumk}{\sum_{k = 1}^{K}}

\newcommand{\limR}{\lim_{R \to \infty}}

\newcommand{\be}{\begin{eqnarray}}
\newcommand{\ee}{\end{eqnarray}}
\newcommand{\ben}{\begin{eqnarray*}}
\newcommand{\een}{\end{eqnarray*}}

\newcommand{\expect}[1]{\EE\left[{\displaystyle #1}\right]}
\newcommand{\expectx}[1]{\EE_{\hat{x}}\left[{\displaystyle #1}\right]}
\newcommand{\indi}[1]{{\rm I}_{\left\{{\displaystyle #1}\right\}}}
\newcommand{\pr}[1]{\PP\left\{{\displaystyle #1}\right\}}

\newtheorem{corollary}{Corollary}
\newtheorem{definition}{Definition}
\newtheorem{lemma}{Lemma}
\newtheorem{proposition}{Proposition}
\newtheorem{remark}{Remark}
\newtheorem{theorem}{Theorem}

\title{Queue-Based Random-Access Algorithms: \\
Fluid Limits and Stability Issues}

\author{\IEEEauthorblockN{Javad Ghaderi$^\dagger$, Sem Borst$^\ddagger$,
Phil Whiting$^\ddagger$\\}
\IEEEauthorblockA{
$^\dagger$Department of ECE and Coordinated Science Lab, \\
University of Illinois at Urbana-Champaign,\thanks{The work of the
first author was carried out in the course of a summer internship
project at Alcatel-Lucent Bell Labs.} \\
{\normalsize Email: jghaderi@illinois.edu} \\ \ \\
$^\ddagger$Department of Mathematics of Networks and Communications, \\
Alcatel-Lucent Bell Labs, Murray Hill, NJ 07974 \\
{\normalsize Email: \{sem,pwhiting\}@research.bell-labs.com}}}

\begin{document}

\maketitle

\begin{abstract}
We use fluid limits to explore the (in)stability properties of
wireless networks with queue-based random-access algorithms.
Queue-based random-access schemes are simple and inherently distributed in nature,
yet provide the capability to match the optimal throughput performance
of centralized scheduling mechanisms in a wide range of scenarios.
Unfortunately, the type of activation rules for which throughput
optimality has been established, may result in excessive queue
lengths and delays.
The use of more aggressive/persistent access schemes can improve the
delay performance, but does not offer any universal maximum-stability
guarantees.

In order to gain qualitative insight and investigate the (in)stability
properties of more aggressive/persistent activation rules,
we examine fluid limits where the dynamics are scaled in space and time.
In some situations, the fluid limits have smooth deterministic
features and maximum stability is maintained, while in other scenarios
they exhibit random oscillatory characteristics, giving rise to major
technical challenges.
In the latter regime, more aggressive access schemes continue to provide
maximum stability in some networks, but may cause instability in others.
Simulation experiments are conducted to illustrate and validate
the analytical results.
\end{abstract}

\section{Introduction}

Emerging wireless mesh networks typically lack any centralized
access control entity, and instead vitally rely on the individual
nodes to operate autonomously and to efficiently share the medium
in a distributed fashion.
This requires the nodes to schedule their individual transmissions
and decide on the use of a shared medium based on knowledge that is
locally available or only involves limited exchange of information.
A popular mechanism for distributed medium access control is provided
by the so-called Carrier-Sense Multiple-Access (CSMA) protocol.
In the CSMA protocol each node attempts to access the medium after
a certain back-off time, but nodes that sense activity of interfering
nodes freeze their back-off timer until the medium is sensed idle.
While the CSMA protocol is fairly easy to understand at a local level,
the interaction among interfering nodes gives rise to quite intricate
behavior and complex throughput characteristics on a macroscopic scale.
In recent years relatively parsimonious models have emerged that
provide a useful tool in evaluating the throughput characteristics
of CSMA-like networks, see for instance \cite{BKMS87,DDT07,DT06,WK05}.
Experimental results in Liew {\em et al.}~\cite{LKLW08} demonstrate
that these models, while idealized, provide throughput estimates that
match remarkably well with measurements in actual systems. 

Despite their asynchronous and distributed nature, CSMA-like
algorithms have been shown to offer the remarkable capability
of achieving the full capacity region and thus match the optimal
throughput performance of centralized scheduling mechanisms operating
in slotted time \cite{JW08,JW10,LYPCP10}.
More specifically, any throughput vector in the interior of the
convex hull associated with the independent sets in the underlying
interference graph can be achieved through suitable back-off rates
and/or transmission lengths.
Based on this observation, various ingenious algorithms have been
developed for finding the back-off rates that yield a particular target
throughput vector or that optimize a certain concave throughput utility
function in scenarios with saturated buffers \cite{JW08,JW10,ME08}.
In the same spirit, several effective approaches have been devised
for adapting the transmission lengths based on queue length information,
and been shown to guarantee maximum stability
\cite{JSSW10,RSS09,SS12,SST11}.

Roughly speaking, the maximum-stability guarantees were established
under the condition that the activity factors of the various nodes
behave as logarithmic functions of the queue lengths.
Unfortunately, such activity factors can induce excessive queue lengths
and delays, which has triggered a strong interest in developing approaches
for improving the delay performance \cite{GS12,LEYY12,LM10,NTS10,SS10}.
Motivated by this issue, Ghaderi \& Srikant~\cite{GS10} recently
showed that it is in fact sufficient for the \emph{logarithms} of the
activity factors to behave as logarithmic functions of the queue lengths,
divided by an arbitrarily slowly increasing, unbounded function.
These results indicate that the maximum-stability guarantees are
preserved for activity functions that are essentially linear for all
practical values of the queue lengths, although asymptotically the
activity rate must grow slower than any positive power of the queue length.
A careful inspection reveals that the proof arguments leave little
room to weaken the stated growth condition.
Since the growth condition is only a sufficient one, however, it is
not clear to what extent it is actually a strict requirement for
maximum stability to be maintained.

In the present paper we explore the scope for using more aggressive
activity functions in order to improve the delay performance while
preserving the maximum-stability guarantees.
Since the proof methods of \cite{GS10,JSSW10,RSS09,SS12,SST11} do not
easily extend to more aggressive activity functions, we will instead
adopt fluid limits where the dynamics of the system are scaled in both
space and time.
Fluid limits may be interpreted as first-order approximations of
the original stochastic process, and provide valuable qualitative
insight and a powerful approach for establishing (in)stability
properties \cite{Dai95,Dai96,DM95,Meyn95}.

As observed in~\cite{BBLP11}, qualitatively different types of fluid
limits can arise, depending on the structure of the interference graph,
in conjunction with the functional shape of the activity factors.
For sufficiently \emph{tame} activity functions as in
\cite{GS10,RSS09,SS12,SST11}, `fast mixing' is guaranteed,
where the activity process evolves on a much faster time scale than
the scaled queue lengths.
Qualitatively similar fluid limits can arise for more \emph{aggressive}
activity functions as well, provided the topology is benign in
a certain sense, which implies that the maximum-stability guarantees
are preserved in those cases.
In different regimes, however, aggressive activity functions can cause
`sluggish mixing', where the activity process evolves on a much slower
time scale than the scaled queue lengths, yielding oscillatory fluid
limits that follow random trajectories.
It is highly unusual for such random dynamics to occur,
as in queueing networks typically the random characteristics vanish
and deterministic limits emerge on the fluid scale.
A few exceptions are known for various polling-type models as
considered in \cite{FK99,KTF05,FFZ13}.

The random nature of the fluid limits gives rise to several complications
in the convergence proofs that are not commonly encountered.
Since the random-access networks that we consider are fundamentally
different from the polling type-models in the above-mentioned references,
the fluid limits are qualitatively different as well,
and require a substantially different approach to establish convergence.
Specifically, we develop an approach based on stopping time sequences
to deal with the switching probabilities governing the sample paths
of the fluid limit process.
While these proof arguments are developed in the context of
random-access networks, several key components extend far beyond the
scope of the present problem.
Hence, we believe that the proof constructs are of broader methodological
value in handling random fluid limits and of potential use in establishing
both stability and instability results for a wider range of models.
For example, the methodology that we develop could be easily applied
to prove the stability results for the random capture scheme
as conjectured in work of Feuillet {\em et al.}~\cite{FPR10}.

The possible oscillatory behavior of the fluid limit itself does not
necessarily imply that the system is unstable,
and in some situations maximum stability is in fact maintained.
In other scenarios, however, the fluid limit reflects that more
aggressive activity functions may force the system into inefficient
states for extended periods of time and produce instability.
We will demonstrate instability for super-linear activity functions,
but our proof arguments suggest that it can potentially occur for any
activity factor that grows as a positive power of the queue lengths
in networks with sufficiently many nodes.
In other words, the growth conditions for maximum stability depend on
the number of nodes, which seems loosely related to results in
\cite{JLNSW11,STT11,SA11} characterizing how (upper bounds for) the mean
queue length and delay scale as a function of the size of the network.

The remainder of the paper is organized as follows.
In Section~\ref{mode}, we present a detailed model description.
We introduce fluid limits and discuss the various qualitative
regimes in Section~\ref{qualdiamond}.
We then use the fluid limits to demonstrate the potential instability
of aggressive activity functions in Sections~\ref{brokendiamond}
and~\ref{instab}.
Simulation experiments are conducted in Section~\ref{simu} to
support the analytical results.
In Section~\ref{conc}, we make some concluding remarks and identify
topics for further research. Appendices at the end of the paper contain proofs of our results.
\section{Model description}
\label{mode}

\textit{Network, interference graph, and traffic model}

We consider a network of several nodes sharing a wireless medium
according to a random-access mechanism.
The network is represented by an undirected graph $G = (V , E)$
where the set of vertices $V = \{1,\ldots, N\}$ correspond to the
various nodes and the set of edges $E \subseteq V \times V$ indicate
which pairs of nodes interfere.
Nodes that are neighbors in the interference graph are prevented
from simultaneous activity, and thus the independent sets
correspond to the feasible joint activity states of the network.
A node is said to be blocked whenever the node itself or any of
its neighbors is active, and unblocked otherwise.
Define $S \subseteq \{0, 1\}^N$ as the set of incidence vectors of
all the independent sets of the interference graph,
and denote by $\calC = \mbox{conv}(S)$ the capacity region,
with $\mbox{conv}(\cdot)$ indicating the convex hull operator.

Packets arrive at node~$i$ as a Poisson process of rate~$\lambda_i$.
The packet transmission times at node~$i$ are independent
and exponentially distributed with mean $1 / \mu_i$.
Denote by $\rho_i = \lambda_i/\mu_i$ the traffic intensity of node~$i$.

Let $U(t) \in S$ represent the joint activity state of the network
at time~$t$, with $U_i(t)$ indicating whether node~$i$ is active at
time~$t$ or not.
Denote by $\Qorig_i(t)$ the queue length at node~$i$ at time~$t$,
i.e., the number of packets waiting for transmission or in the process
of being transmitted. \\

\textit{Queue-based random-access mechanism}

As mentioned above, the various nodes share the medium in
accordance with a random-access mechanism.
When a node ends an activity period (consisting of possibly several
back-to-back packet transmissions), it starts a back-off period.
The back-off times of node~$i$ are independent and exponentially
distributed with mean $1 / \nu_i$.
The back-off period of a node is suspended whenever it becomes
blocked by activity of any of its neighbors, and only resumed once
the node becomes unblocked again.
Thus the back-off period of a node can only end when none of its
neighbors are active.
Now suppose a back-off period of node~$i$ ends at time~$t$.
Then the node starts a transmission with probability $\phi_i(\Qorig_i(t))$,
with $\phi_i(0) = 0$, and begins a next back-off period otherwise.
When a transmission of node~$i$ ends at time~$t$, it releases the medium
and begins a back-off period with probability $\psi_i(\Qorig_i(t^-))$,
or starts the next transmission otherwise, with $\psi_i(1) = 1$.
Equivalently, node~$i$ may be thought of as activating at an exponential
rate $f_i(\Qorig_i(t))$, with $f_i(\cdot) = \nu_i \phi_i(\cdot)$,
whenever it is unblocked at time~$t$, and de-activating at rate
$g_i(\Qorig_i(t))$, with $g_i(\cdot) = \mu_i \psi_i(\cdot)$,
whenever it is active at time~$t$.
For conciseness, the functions $f_i(\cdot)$ and $g_i(\cdot)$ will be
referred to as \textit{activation} and \textit{de-activation functions}, respectively.

There are two special cases worth mentioning that (loosely) correspond
to random-access schemes considered in the literature before.
First of all, in case $\phi_i(\Qorig_i) = 1$ and $\psi_i(\Qorig_i) = 0$
for all $\Qorig_i \geq 1$, node~$i$ starts a transmission each time
a back-off period ends, and does not release the medium,
i.e., continues transmitting until its entire queue has been cleared.
This corresponds to the random-capture scheme considered in~\cite{FPR10}.
In case $\mu_i = 1$, $\nu_i = 1$, $\phi_i(\Qorig_i) = 1-\psi_i(\Qorig_i)$,
and $\psi_i(\Qorig_i) = 1 / (1 + r_i(\Qorig_i))$, node~$i$ may be
thought of as becoming (or continuing to be) active with probability
$r_i(\Qorig_i(t)) / (1 + r_i(\Qorig_i(t)))$ each time a unit-rate
Poisson clock ticks.
This roughly corresponds to the scheme considered in
\cite{GS10,JSSW10,RSS09,SS12,SST11} based on Glauber dynamics with
a `weight' function $w_i(\Qorig_i) = \log(r_i(\Qorig_i))$,
except that the latter scheme operates with a random round-robin clock,
and uses $\tilde{w}_i(\Qorig_i) =
\max\{w_i(\Qorig_i), \frac{\epsilon}{2 N} w_i(\Qorig_{\max})\}$,
with $\Qorig_{\max} = \max_{j = 1, \dots, N} \Qorig_j$. \\

\textit{Network dynamics}

Under the above-described queue-based schemes, the process
$\{(U(t), \Qorig(t))\}_{t \geq 0}$ evolves as a continuous-time Markov
process with state space $S \times {\mathbb N}_0^N$.
Transitions (due to arrivals) from a state $(U, \Qorig)$ to
$(U, \Qorig + e_i)$ occur at rate~$\lambda_i$, transitions (due to
activations) from a state $(U, \Qorig)$ with $\Qorig_i \geq 1$, $U_i = 0$,
and $U_j = 0$ for all neighbors of node~$i$, to $(U + e_i, \Qorig)$
occur at rate $\nu_i f_i(\Qorig_i)$, transitions (due to transmission
completions followed back-to-back by a subsequent transmission) from
a state $(U, \Qorig)$ with $U_i = 1$ (and thus $\Qorig_i \geq 1$) to
$(U, \Qorig - e_i)$ occur at rate $\mu_i (1 - g_i(\Qorig_i))$,
transitions (due to transmission completions followed by a back-off period)
from a state $(U, \Qorig)$ with $U_i = 1$ (and thus $\Qorig_i \geq 1$)
to $(U - e_i, \Qorig - e_i)$ occur at rate $\mu_i g_i(\Qorig_i)$.

We are interested to determine under what conditions the system is stable,
i.e., the process $\{(U(t), \Qorig(t))\}_{t \geq 0}$ is positive-recurrent.
It is easily seen that $(\rho_1, \dots, \rho_N) < \sigma \in \calC$
is a necessary condition for that to be the case.
In~\cite{GS10}, it is shown that this condition is in fact also
sufficient for weight functions of the form
$w_i(\Qorig_i) = \log(1+\Qorig_i)/y_i(\Qorig_i)$, where $y_i(\Qorig_i)$
is allowed to increase to infinity at an arbitrarily slow rate.
For practical purposes, this means that the function $r_i(\Qorig_i)$ is
essentially allowed to be linear, except that it must eventually
grow to infinity slower than any positive power of~$\Qorig_i$.
Results in~\cite{BBLP11} suggest that more aggressive choices of
the functions $f_i(\cdot)$ and $g_i(\cdot)$, which translate into
functions $r_i(\cdot)$ that grow faster to infinity, can improve the
delay performance.
In view of these results, we will be particularly interested in such
functions $r_i(\cdot)$, where the stability results of~\cite{GS10}
do not apply.
In order to examine under what conditions the system will remain
stable then, we will examine fluid limits for the process
$\{(U(t), \Qorig(t))\}_{t \geq 0}$ as introduced in the next section.

\section{Qualitative discussion of fluid limits}
\label{qualdiamond}

Fluid limits may be interpreted as first-order approximations of
the original stochastic process, and provide valuable qualitative
insight and a powerful approach for establishing (in)stability
properties \cite{Dai95,Dai96,DM95,Meyn95}.
In this section we discuss fluid limits for the process
$\{(U(t), \Qorig(t))\}_{t \geq 0}$ from a broad perspective, with the
aim to informally exhibit their qualitative features in various regimes,
and we deliberately eschew rigorous claims or proofs.

\subsection{Fluid-scaled process}

In order to obtain fluid limits, the original stochastic process
is scaled in both space and time.
More specifically, we consider a sequence of processes
$\{(U^{(R)}(t), \Qorig^{(R)}(t))\}_{t \geq 0}$ indexed by a sequence
of positive integers~$R$, each governed by similar statistical laws
as the original process, where the initial states satisfy
$\sumi \Qorig_i^{(R)}(0) = R$ and $\Qorig_i^{(R)}(0) / R \to Q_i$
as $R \to \infty$.
The process $\{(U^{(R)}(R t), \frac{1}{R} \Qorig^{(R)}(R t))\}_{t \geq 0}$
is referred to as the fluid-scaled version of the process
$\{(U^{(R)}(t), \Qorig^{(R)}(t)\}_{t \geq 0}$.
Note that the activity process is scaled in time as well but not in space.
For compactness, denote $Q^R(t) = \frac{1}{R} \Qorig^{(R)}(R t)$.
Any (possibly random) weak limit $\{Q(t)\}_{t \geq 0}$ of the sequence
$\{Q^R(t)\}_{t \geq 0}$, as $R \to \infty$, is called a fluid limit.

It is worth mentioning that the above notion of fluid limit based on
the continuous-time Markov process is only introduced for the
convenience of the qualitative discussion below.
For all the proofs of fluid limit properties and instability results
we will rely on a rescaled linear interpolation of the uniformized
jump chain (as will be defined in Appendix~A.I), with a time-integral
version of the $U(\cdot)$ component.
This construction yields convenient properties of the fluid limit paths
and allows us to extend the framework of Meyn~\cite{Meyn95} for
establishing instability results for discrete-time Markov chains.
(The original continuous-time Markov process has in fact the same fluid
limit properties, but this is not directly relevant in any of the proofs.)

The process $\{(U^{(R)}(R t), \frac{1}{R} \Qorig^{(R)}(R t))\}_{t \geq 0}$
comprises two interacting components.
On the one hand, the evolution of the (scaled) queue length process
$\frac{1}{R} \Qorig^{(R)}(R t)$ depends on the activity process
$U^{(R)}(R t)$.
On the other hand, the evolution of the activity process $U^{(R)}(R t)$
depends on the queue length process $\Qorig^{(R)}(R t)$ through the
activation and de-activation functions $f_i(\cdot)$ and $g_i(\cdot)$.
In many cases, a separation of time scales arises as $R \to \infty$,
where the transitions in $U^{(R)}(R t)$ occur on a much faster time
scale than the variations in $Q^R(t) = \frac{1}{R} \Qorig^{(R)}(t)$.
Loosely phrased, the evolution of $Q^R(t)$ is then governed by the
time-average characteristics of $U^{(R)}(\cdot)$ in a scenario
where $Q^R(t)$ is fixed at its instantaneous value.

In other cases, however, the transitions in $U^{(R)}(R t)$ may in fact
occur on a much slower time scale than the variations in $Q^R(t)$,
or there may not be a separation of time scales at all.
As a result, qualitatively different types of fluid limits can arise,
as observed in~\cite{BBLP11}, depending on the mixing properties
of the activity process.
These mixing properties, in turn, depend on the functional shape of the
activation and de-activation functions $f_i(\cdot)$ and $g_i(\cdot)$,
in conjunction with the structure of the interference graph~$G$.

\subsection{Fast mixing: smooth deterministic fluid limits}

We first consider the case of fast mixing.
In this case, the transitions in $U^{(R)}(R t)$ occur on a much faster
time scale than the variations in $Q^R(t)$, and completely average out
on the fluid scale as $R \to \infty$.
Informally speaking, this entails that the mixing time of the
activity process in a scenario with fixed activation rates
$f_i(R q_i)$ and de-activation rates $g_i(R q_i)$ grows slower
than~$R$ as $R \to \infty$.
In order to obtain a rough bound for the mixing time, assume that
$f_i(\cdot) \equiv f(\cdot)$, $g_i(\cdot) \equiv g(\cdot)$,
and denote $h(x) = f(x) / g(x)$.
Further suppose that $h(R) \to \infty$ as $R \to \infty$,
and $h(a R) / h(R) \to \hat{h}(a)$ as $R \to \infty$,
with $\hat{h}(a) > 0$ for any $a > 0$.
The latter assumptions are satisfied, for example, when
$h(x) = x^\gamma$, $\gamma > 0$, with $\hat{h}(a) = a^\gamma$,
or when $h(x) = \log(x)$ with $\hat{h}(a) \equiv 1$.
Without proof, we claim that the mixing time then grows at most at
rate $f(R)^{m^* - 1} g(R)^{- m^*}$ as $R \to \infty$, with $m^*$
the cardinality of a maximum-size independent set.
Thus, fast mixing behavior is guaranteed when $f(\cdot)$ does not
grow too fast, $g(\cdot)$ does not decay too fast, or $m^*$ is
sufficiently small, e.g.,
\begin{itemize}
\item[(i)] $g(x) = g$ and $m^* = 1$;
\item[(ii)] $f(x) = x^{1 / (m^* - 1) - \delta}$, $g(x) = g$,
and $m^* \geq 2$;
\item[(iii)] $f(x) = f$ and $g(x) \geq x^{- 1 / m^* + \delta}$;
\item[(iv)] $f(x) = f$, $g(x) = 1 / \log(1+x)$;
\item[(v)] $f(x) = \log(1+x)$ and $g(x) = g$.
\end{itemize}

As mentioned above, the fluid limit then follows an entirely
deterministic trajectory, which is described by a differential
equation of the form
\[
\frac{\dd}{\dd t}Q_i(t) = \lambda_i - \mu_i u_i(Q(t)),
\]
as long as $Q(t) > 0$ (component-wise), with the function $u_i(\cdot)$
representing the fraction of time that node~$i$ is active.
We may write
\[
u_i(q) = \sum_{s \in S} s_i \pi(s; q),
\]
with $\pi(s; q)$ denoting the fraction of time that the activity process
resides in state $s \in S$ in a scenario with fixed activation rates
$f_j(R q_j)$ and de-activation rates $g_j(R q_j)$ as $R \to \infty$.
Let $S^* = \{s \in S: \sumi s_i = m^*\}$ correspond to the
collection of all maximum-size independent sets.
Under the above-mentioned assumptions,
\be
&\pi(s; q) =
\limR \frac{\prod\limits_{i = 1}^{N} h(R q_i)^{s_i}}
{\sum_{u \in S^*} \prod\limits_{i = 1}^{N} h(R q_i)^{u_i}} \nonumber \\
&= \frac{\prod\limits_{i = 1}^{N} \hat{h}(q_i)^{s_i}}
{\sum_{u \in S^*} \prod\limits_{i = 1}^{N} \hat{h}(q_i)^{u_i}} =
\frac{\exp(\sumi s_i \log(\hat{h}(q_i)))}
{\sum_{u \in S^*} \exp(\sumi u_i \log(\hat{h}(q_i)))} \nonumber ,
\ee
for $s \in S^*$, while $\pi(s; q) = 0$ for $s \not\in S^*$.
In particular, if $h(x) = x^\gamma$, $\gamma > 0$, then
$$
\pi(s; q) = \frac{\prod\limits_{i = 1}^{N} q_i^{\gamma s_i}}
{\sum_{u \in S^*} \prod\limits_{i = 1}^{N} q_i^{\gamma u_i}} =
\frac{\exp(\gamma \sumi s_i \log(q_i))}
{\sum_{u \in S^*} \exp(\gamma \sumi u_i \log(q_i))},
$$
for $s \in S^*$.
Also, if $h(x) = \log(1+x)$, then $\pi(s; q) = 1/|S^*|$ for $s \in S^*$.

When some of the components of~$q$ are zero, i.e., some of the queue
lengths are zero at the fluid scale, it is considerably harder to
characterize $u_i(q)$, since the competition for medium access from
the queues that are zero at the fluid scale still has an impact.
It may be shown though that
$$\sumi \rho_i \indi{q_i > 0} \leq
(1 - \epsilon) \sumi u_i(q) \indi{q_i > 0}$$ for some $\epsilon > 0$,
assuming that $(\rho_1, \dots, \rho_N) < \sigma \in \calC$.
The latter inequality also holds when $q > 0$, noting that then
$\sumi u_i(q) = m^*$, while $\sumi \rho_i \leq (1 - \epsilon) m^*$
for some $\epsilon > 0$.

We conclude that almost everywhere
\be
\sumi \frac{1}{\mu_i}\frac{\dd Q_i(t)}{\dd t}
&\leq& \sumi (\rho_i - u_i(Q(t))) \indi{Q_i(t) > 0} \nonumber \\
&\leq& - \epsilon \sumi \rho_i \indi{Q_i(t) > 0} \nonumber,
\ee
as long as $Q(t) \neq 0$.
This means that $Q(t) = 0$ for all $t \geq T$ for some finite
$T < \infty$, which implies that the original Markov process is
positive-recurrent \cite{Dai95,DM95}.
This agrees with the stability results in
\cite{GS10,JSSW10,RSS09,SST11,SS12} for the case
$f(\Qorig_i) = 1 - g(\Qorig_i)$
and $g(\Qorig_i) = 1 / (1 + \exp(\tilde{w}(\Qorig_i)))$,
$\tilde{w}(\Qorig_i) =
\max\{w(\Qorig_i), \frac{\epsilon}{2 N} w(\Qorig_{\max})\}$
(with the minor differences noted in the previous section),
and suggests that these results in fact hold without the need to
know the maximum queue size $\Qorig_{\max}$.

Of course, in order to convert the above arguments into an actual
stability proof, the informal characterization of the fluid limit
needs to be rigorously justified.
This is a major challenge, and not the real goal of the present
paper, since we aim to demonstrate the opposite, namely that more
aggressive activity or de-activation functions can cause instability.
Strong evidence of the technical complications in establishing the fluid
limits is provided by recent work of Robert \& V\'eber~\cite{RV12}.
Their work focuses on the simpler case of a single work-conserving resource
(which corresponds to a full interference graph in the present setting)
without any back-off mechanism, where the service rates of the various
nodes are determined by a logarithmic function of their queue lengths.

\subsection{Sluggish mixing: erratic random fluid limits}

With the above aim in mind, we now turn to the case of sluggish mixing.
In this case, the transitions in $U^{(R)}(R t)$ occur on a much slower
time scale than the variations in $Q^R(t)$, and vanish on the fluid scale
as $R \to \infty$, except at time points where some of the queues hit zero.
The detailed behavior of the fluid limit in this case depends
delicately on the specific structure of the interference graph~$G$
and the shape of the functions $f_i(\cdot)$ and $g_i(\cdot)$.
This prevents a characterization in any degree of generality,
and hence we focus attention on some particular scenarios.

\begin{figure}
\centering
\includegraphics[width=2.5 in]{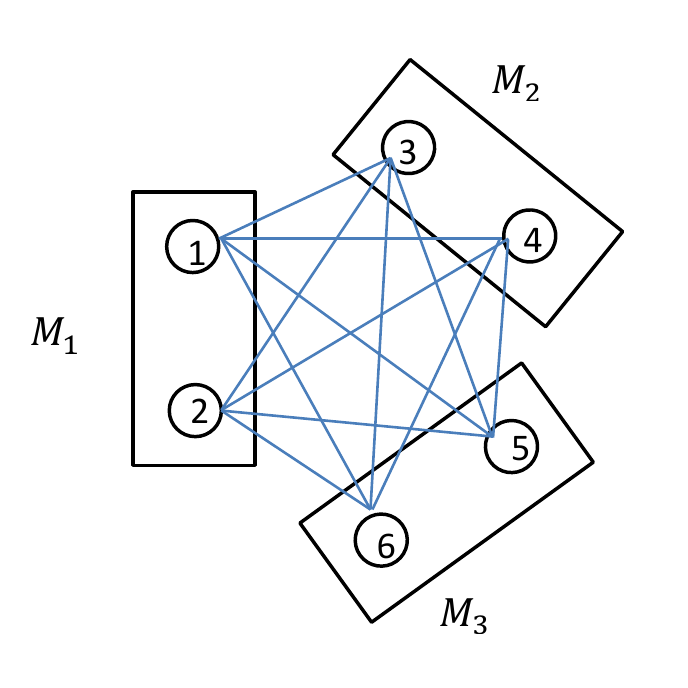}
\caption{The \emph{diamond network}: A complete partite graph with
$K = 3$ components, each containing 2~nodes.}
\label{network1}
\end{figure}

In order to show that sluggish mixing behavior itself need not imply
instability, we first examine a complete $K$-partite graph as
considered in~\cite{FPR10}, where the nodes can be partitioned into
$K \geq 2$ components.
All nodes are connected except those belonging to the same component.
Figure~\ref{network1} depicts an example of a complete partite graph
with $K = 3$ components, each containing 2~nodes.
We will refer to this network as the \emph{diamond network}, since
the edges correspond to those of an eight-faced diamond structure,
with the node pairs constituting the three components positioned at
the opposite ends of three orthogonal axes.

Denote by $M_k \subseteq \{1, \dots, N\}$ the subset of nodes
belonging the $k$-th component.
Once one of the nodes in component $M_k$ is active, other nodes
within $M_k$ can become active as well, but none of the nodes in
the other components $M_l$, $l \neq k$, can be active.
The necessary stability condition then takes the form
$\rho = \sumk \hat{\rho_k} < 1$,
with $\hat{\rho}_k = \max_{i \in M_k} \rho_i$ denoting the maximum
traffic intensity of any of the nodes in the $k$-th component.

Now consider the case that each node operates with an activation
function $f(x)$ with $\lim_{x \to \infty} f(x) > 0$ and a de-activation
function $g(x) = \oo(x^{-\gamma})$, with $\gamma > 1$,
which subsumes the random-capture scheme with $g(x) \equiv 0$ for
all $x \geq 1$ in~\cite{FPR10}.
Since the de-activation rate decays so sharply, the probability of
a node releasing the medium once it has started transmitting
with an initial queue length of order~$R$, is vanishingly small,
until the queue length falls below order~$R$ or the total number of
transmissions exceeds order~$R$ (but the latter implies the former).
Hence, in the fluid limit, a node must completely empty almost
surely before it releases the medium.
Because of the interference constraints, it further follows that once
the activity process enters one of the components, it remains there
until all the queues in that component have entirely drained (on the
fluid scale), and then randomly switches to one of the other components.
For conciseness, the fluid limit process is said to be in
an $M_k$-period during time intervals when at least one of the nodes
in component~$M_k$ is served at full rate (on the fluid scale).

Based on the above informal observations, we now proceed with a more
detailed description of the dynamics of the fluid limit process.
We do not aim to provide a proof of the stated properties,
since the main goal of the present paper is to demonstrate the
potential for instability rather than establish stability.
However, the proof arguments that we will develop for a similar but
more complicated interference graph in the remainder of the paper,
could easily be applied to provide a rigorous justification of the
fluid limit and establish the claimed stability results.

Assume that the system enters an $M_k$-period at time~$t$, then
\begin{enumerate}
\item[(a)] It spends a time period
$T_k(t) = \max_{i \in M_k} \frac{Q_i(t)}{\mu_i - \lambda_i}$ in $M_k$.
\item[(b)] During this period, the queues of the nodes in $M_k$
drain at a linear rate (or remain zero)
\[
Q_i(t + u) = \max\{Q_i(t) + (\lambda_i - \mu_i) u, 0\},~ \forall i \in M_k,
\]
while the queues of the other nodes fill at a linear rate
\[
Q_i(t + u) = Q_i(t) + \lambda_i u,~ \forall i \not\in M_k,
\]
for all $u \in [0, T_k(t)]$.
\item[(c)] At time $t + T_k(t)$, the system switches to an $M_l$-period,
$l \neq k$, with probability
\[
p_{kl}(t + T_k(t)) = \limR \frac{\sum_{i \in M_l} f(R Q_i(t + T_k(t)))}
{\sum_{l' \neq k, l} \sum_{i \in M_{l'}} f(R Q_i(t + T_k(t)))}.
\]
\end{enumerate}

Thus the fluid limit follows a piece-wise linear sample path,
with switches between different periods governed by the transition
probabilities specified above.
Figure~\ref{path1} depicts an example of the fluid limit sample path
for the network of Figure~\ref{network1} with $f(x) = 1$, $x \geq 1$.

Now define the Lyapunov function $\Load(t) := \sumk \hat{Q}_k(t)$,
with $\hat{Q}_k(t) = \max_{i \in M_k} Q_i(t) / \mu_i$.
Then, $\frac{\dd}{\dd t} \Load(t) \leq \sumk \hat\rho_k - 1 = \rho - 1 < 0$
almost everywhere when $\rho <1$, as long as $\Load(t) > 0$.
Therefore, $\Load(t) = 0$, and hence $Q(t) = 0$, for all $t \geq T$,
with $T = \frac{\Load(0)}{1 - \rho} < \infty$,
implying stability~\cite{Dai95,DM95}, even though the fluid limit
behavior is not smooth at all.

\begin{figure}
\centering
\includegraphics[width=4in]{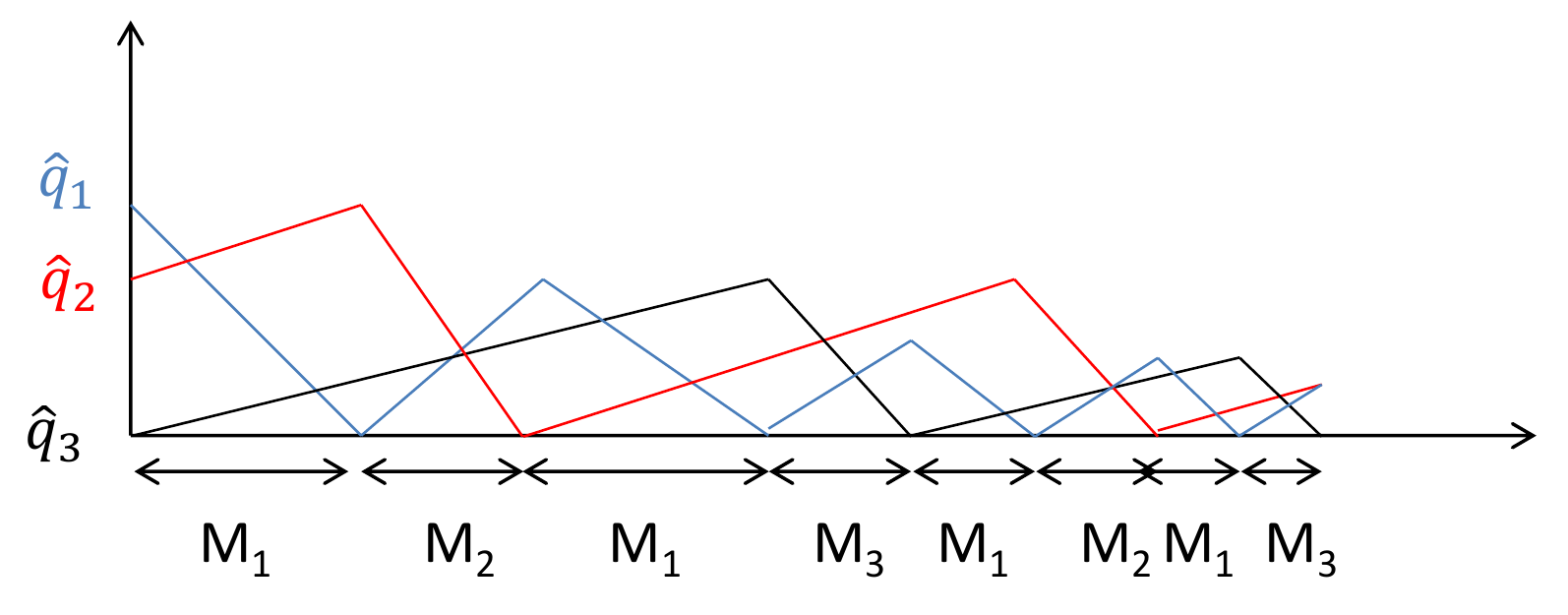}
\caption{A fluid limit sample path for the \emph{diamond network}
of Figure~\ref{network1}.}
\label{path1}
\end{figure}

\section{Fluid limits for broken-diamond network}
\label{brokendiamond}

In the previous section we discussed qualitative features of fluid
limits in various scenarios, and in particular for so-called complete
partite graphs.
We now proceed to consider a `nearly' complete partite graph,
and will demonstrate that if some of the edges between two
components~$M_k$ and~$M_l$ are removed (thus reducing interference),
the network might become unstable for `aggressive' activation
and/or deactivation functions!
Specifically, we will consider the \emph{diamond network}
of Figure~\ref{network1}, and remove the edge between nodes~4 and~5
to obtain a \emph{broken-diamond network} with an additional
component/maximal schedule $M_4$, as depicted in Figure~\ref{network2}.

\begin{figure}
\centering
\includegraphics[width=2.5in]{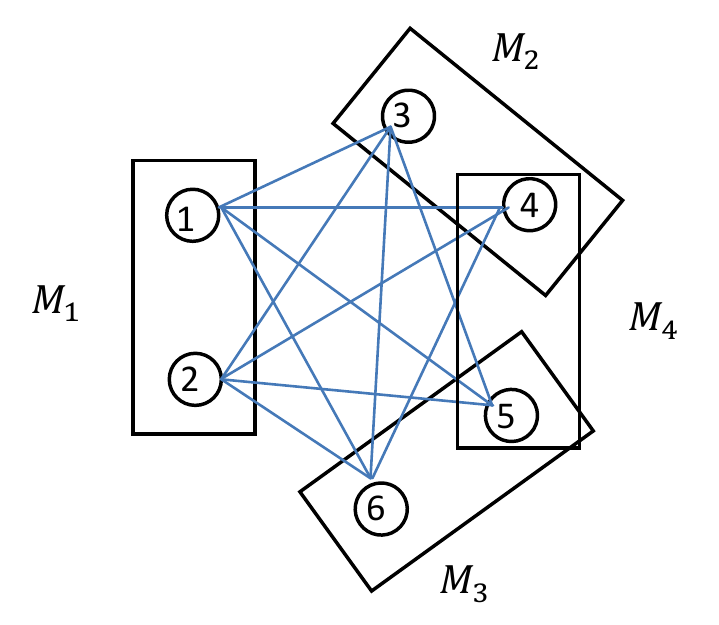}
\caption{The \emph{broken-diamond network}, obtained by removing
1~edge from the \emph{diamond network} of Figure~\ref{network1},
yielding an additional schedule $M_4$.}
\label{network2}
\end{figure}

The intuitive explanation for the potential instability may be
described as follows.
Denote $\rho_0 = \max\{\rho_1, \rho_2\}$, and assume
$\rho_3 \geq \rho_4$ and $\rho_6 \geq \rho_5$.
It is easily seen that the fraction of time that at least one of
the nodes 1, 2, 3 and 6 is served, must be no less than
$\rho = \rho_0 + \rho_3 + \rho_6$ in order for these nodes to be stable.
During some of these periods nodes~4 or~5 may also be served,
but not simultaneously, i.e., schedule $M_4$ cannot be used.
In other words, the system cannot be stable if schedule $M_4$ is
used for a fraction of the time larger than $1 - \rho$.
As it turns out, however, when the de-activation function is
sufficiently aggressive, e.g., $g(x) = \oo(x^{- \gamma})$,
with $\gamma > 1$, schedule $M_4$ is in fact persistently used for
a fraction of the time that does not tend to~0 as $\rho$ approaches~1,
which forces the system to be unstable.

Although the above arguments indicate that invoking schedule $M_4$ is
a recipe for trouble, the reason may not be directly evident from the
system dynamics, since no obvious inefficiency occurs as long as the
queues of nodes~4 and~5 are non-empty.
However, the fact that the Lyapunov function
$\Load(t) = \sum_{k = 1}^{3} \max_{i \in M_k} Q_i(t)$ may increase while
serving nodes~4 and~5, when $Q_3(t) \geq Q_4(t)$ and $Q_6(t) \geq Q_5(t)$,
is already highly suggestive. (Such an increase is depicted in Figure \ref{path2}
during the $M_4$ period of the switching sequence
$M_1 \rightarrow M_2 \rightarrow M_1 \rightarrow M_4 \rightarrow M_3 \rightarrow M_1$.)
Indeed, serving nodes~4 and~5 may make their queues smaller than those
of nodes~3 and~6, leaving these queues to be served by themselves at
a later stage, at which point inefficiency inevitably occurs.

In the sequel, the fluid limit process is said to be in a natural state
when $Q_3(t) \geq Q_4(t)$ and $Q_6(t) \geq Q_5(t)$,
with equality only when both sides are zero.
We will assume $\lambda_3 > \lambda_4$ and $\lambda_6 > \lambda_5$,
and will show that the process must always reside in a natural state
after some finite amount of time.
As described above, instability is bound to occur when schedule $M_4$
is used repeatedly for substantial periods of time while the fluid
limit process is in a natural state.
Since the process is always in a natural state after some finite
amount of time, it is intuitively plausible that such events occur
repeatedly with positive probability, but a rigorous proof that this
leads to instability is far from simple.
Such a proof requires detailed analysis of the underlying stochastic
process (in our case via fluid limits), and its conclusion crucially
depends on the de-activation function.
Indeed, the stability results in \cite{GS10,JSSW10,RSS09,SS12,SST11}
indirectly indicate that the \emph{broken-diamond network} is \emph{not}
rendered unstable for sufficiently cautious de-activation functions.

Just like for the complete partite graphs, the fluid limit process is
said to be in an $M_1$-period when node~1 or node~2 (or both) is
served at full rate.
The process is in an $M_2$- or $M_3$-period when node~3 or~6 is served
at full rate, respectively.
The process is in an $M_4$-period when nodes~4 and~5 are both served
at full rate simultaneously.

In Subsection~\ref{desc} we will provide a detailed description of the
dynamics of the fluid limit process once it has reached a natural
state and entered an $M_1$-, $M_2$-, $M_3$ or $M_4$-period.
The justification for the description follows from a collection of
lemmas and propositions which are stated and proved in Appendices A--D,
with a high-level outline provided in Subsection~\ref{over}.
In Section~\ref{instab} we will exploit the properties of the fluid
limit process in order to prove that the harmful behavior described
above indeed occurs for sufficiently aggressive de-activation functions,
implying instability of the fluid limit process as well as the original
stochastic process.

\subsection{Description of the fluid limit process}
\label{desc}

We now provide a detailed description of the dynamics of the fluid
limit process once it has reached a natural state and entered
an $M_1$-, $M_2$-, $M_3$ or $M_4$-period.
For sufficiently high load, i.e., $\rho$ sufficiently close to~1,
a natural state and such a period occur in uniformly bounded time
almost surely for any initial state.
As will be seen, for de-activation functions $g_i(x) = \oo(x^{- \gamma})$,
with $\gamma>1$, the fluid limit process then follows similar
piece-wise linear trajectories, with random switches,
as described in the previous section for complete partite graphs
and further illustrated in Figure~\ref{path2}.
For notational convenience, we henceforth assume $\mu_i \equiv 1$,
so that $\rho_i \equiv \lambda_i$, for all $i = 1, \dots, N$,
and additionally assume activation functions $f_i(x) \equiv 1$,
$x \geq 1$, for all $i = 1, \dots, N$.

\begin{figure}
\centering
\includegraphics[width=4in]{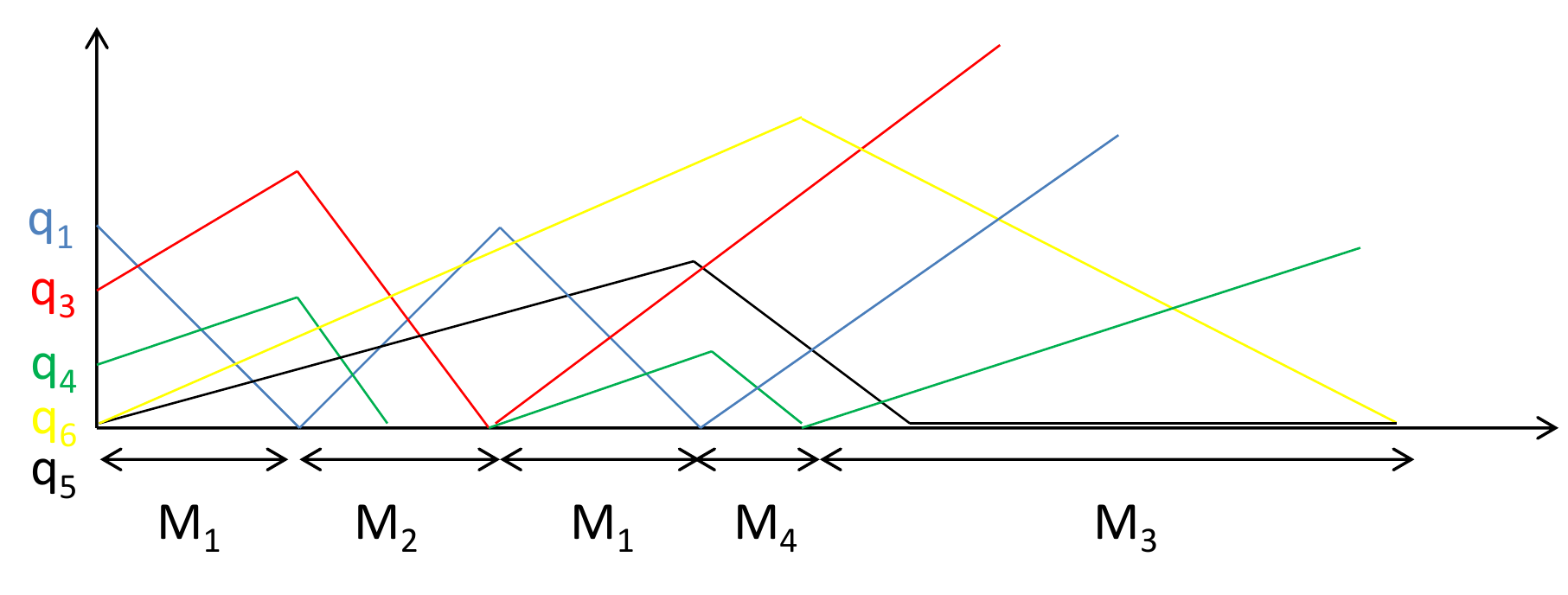}
\caption{A fluid limit sample path for the \emph{broken-diamond network}
of Figure~\ref{network2}, corresponding to the switching sequence
$M_1 \to M_2 \to M_1 \to M_4 \to M_3$.}
\label{path2}
\end{figure}


\subsubsection{$M_1$-period}

Assume the system enters an $M_1$-period at time~$t$, then
\begin{enumerate}
\item[(a)] It spends a time period
$T_1(t) = \max\left\{\frac{Q_1(t)}{1 - \rho_1},
\frac{Q_2(t)}{1 - \rho_2}\right\}$ in $M_1$.
\item[(b)] During this period, the queues of nodes~1 and~2 drain at
a linear rate (or remain zero)
\[
Q_i(t + u) = \max\{Q_i(t) - (1 - \rho_i) u, 0\},
\mbox{ for } i = 1, 2,
\]
while the queues of nodes 3, 4, 5, 6 fill at a linear rate
\[
Q_i(t + u) = Q_i(t) + \rho_i u, \mbox{ for } i = 3, 4, 5, 6,
\]
for all $u \in [0, T_1(t)]$.
In particular, $Q_1(t + T_1(t)) = Q_2(t + T_1(t)) = 0$.
\item [(c)] At time $t + T_1(t)$, the system switches to an $M_2$-, $M_3$-
or $M_4$-period with transition probabilities $p_{12} = \frac{3}{8}$,
$p_{13} = \frac{3}{8}$, and $p_{14} = \frac{1}{4}$, respectively.
\end{enumerate}

\subsubsection{$M_2$-period}

Assume that the system enters an $M_2$-period at time~$t$, then
\begin{enumerate}
\item[(a)] The system spends a time period
$T_2(t) = \frac{Q_3(t)}{1 - \rho_3}$ in $M_2$.
\item[(b)] During this period, the queues of nodes~3 and~4 drain
(or remain zero)
\[
Q_i(t + u) = \max\{Q_i(t) - (1 - \rho_i) u, 0\},
\mbox{ for } i = 3, 4,
\]
while the queues of nodes 1, 2, 5, 6 fill at a linear rate
\[
Q_i(t + u) = Q_i(t) + \rho_i u, \mbox{ for } i = 1, 2, 5, 6,
\]
for all $u \in [0, T_2(t)]$.
In particular, $Q_3(t + T_2(t)) = 0$.
\item [(c)] At time $t + T_2(t)$, the system switches to an $M_1$-
or $M_3$-period.
Note that $\frac{Q_3(t)}{1 - \rho_3} > \frac{Q_4(t)}{1 - \rho_4}$
by the assumption that $\lambda_3 > \lambda_4$ and that the process
has reached a natural state, so that $Q_3(t) > Q_4(t)$ (since
$Q_3(t) = Q_4(t) = 0$ cannot occur at the start of an $M_2$-period).
Thus node~4 has emptied before time $t + T_2(t)$, and remained empty
(on the fluid scale) since then, precluding a switch to an $M_4$-period
except for a negligible duration on the fluid scale), only allowing
the system to switch to either an $M_1$- or $M_3$-period.
The corresponding transition probabilities can be formally expressed
in terms of certain stationary distributions, but are difficult to
obtain in explicit form.
Note that in order for any of the nodes 1, 2, 5 or 6 to activate,
node~3 must be inactive.
In order for nodes 1, 2 or 6 to activate, node~4 must be inactive as
well, but the latter is not necessary in order for node~5 to activate.
Since node~4 may be active even when it is empty on the fluid scale,
it follows that node~5 enjoys an advantage in competing for access
to the medium over nodes 1, 2 and 6.
While it may be argued that node~4 is active with probability~$\rho_4$
by the time node~3 becomes inactive for the first time,
the resulting probabilities for the various nodes to gain access
to the medium first do not seem to allow a simple expression.

\begin{remark}
If the process had not yet reached a natural state, the case
$\frac{Q_3(t)}{1 - \rho_3} \leq \frac{Q_4(t)}{1 - \rho_4}$ could also
arise.
In case that inequality is strict,
i.e., $\frac{Q_3(t)}{1 - \rho_3} < \frac{Q_4(t)}{1 - \rho_4}$,
the queue of node~4 is still non-empty by time $t + T_2(t)$,
simply forcing a switch to an $M_4$-period with probability~1.

In case of equality, i.e.,
$\frac{Q_3(t)}{1 - \rho_3} = \frac{Q_4(t)}{1 - \rho_4}$, however,
the situation would be much more complicated, which serves as the
illustration for the significance of the notion of a natural state.
In order to describe these difficulties, note that the queues of nodes~3
and~4 both empty at time~$t+T_2(t)$, barring a switch to an $M_4$-period,
and permitting only a switch to either an $M_1$- or $M_3$-period.
Just like before, node~5 is the only one able to activate during
periods where node~3 is inactive while node~4 is active,
and hence enjoys an advantage in competing for access to the medium.
In fact, node~5 will gain access to the medium first almost surely
if node~3 is the first one to become inactive (in the pre-limit).
The probability of that event, and hence the transition probabilities
to an $M_1$- or $M_3$-period, depends on queue length differences
between nodes~3 and~4 at time~$t$ that can be affected by the history
of the process and are not visible on the fluid scale.
\end{remark}
\end{enumerate}

\subsubsection{$M_3$-period}

The dynamics for an $M_3$-period are entirely symmetric to those for
an $M_2$-period, but will be replicated below for completeness.

Assume that the system enters an $M_3$-period at time~$t$, then
\begin{enumerate}
\item[(a)] The system spends a time period
$T_3(t) = \frac{Q_6(t)}{1 - \rho_6}$ in $M_3$.
\item[(b)] During this period, the queues of nodes~5 and~6 drain
(or remain zero)
\[
Q_i(t + u) = \max\{Q_i(t) - (1 - \rho_i) u, 0\},
\mbox{ for } i = 5, 6,
\]
while the queues of nodes 1, 2, 3, 4 fill at a linear rate
\[
Q_i(t + u) = Q_i(t) + \rho_i u, \mbox{ for } i = 1, 2, 3, 4,
\]
for all $u \in [0, T_3(t)]$.
In particular, $Q_6(t + T_3(t)) = 0$.
\item [(c)] At time $t + T_3(t)$, the system switches to an $M_1$-
or $M_2$-period.
Note that $\frac{Q_5(t)}{1 - \rho_5} < \frac{Q_6(t)}{1 - \rho_6}$
by the assumption that $\lambda_5 > \lambda_6$ and that the process
has reached a natural state, so that $Q_5(t) < Q_6(t)$ (since
$Q_5(t) = Q_6(t) = 0$ cannot occur at the start of an $M_3$-period).

Thus node~5 has emptied before time $t + T_3(t)$, and remained empty
(on the fluid scale) since then, precluding a switch to an $M_4$-period
(except for a negligible period on the fluid scale), only allowing
the system to switch to either an $M_1$- or $M_2$-period.
The corresponding transition probabilities are difficult to obtain
in explicit form for similar reasons as mentioned in case~2(c).

\begin{remark}
If the process had not yet reached a natural state, the case
$\frac{Q_5(t)}{1 - \rho_5} \geq \frac{Q_6(t)}{1 - \rho_6}$ could also
arise.
In case that inequality is strict,
i.e., $\frac{Q_5(t)}{1 - \rho_5} < \frac{Q_6(t)}{1 - \rho_6}$,
the queue of node~5 is still non-empty by time $t + T_3(t)$,
forcing a switch to an $M_4$-period with probability~1.

In case of equality, i.e.,
$\frac{Q_5(t)}{1 - \rho_5} = \frac{Q_6(t)}{1 - \rho_6}$, the queues
of nodes~5 and~6 both empty at time~$T_3(t)$, barring a switch
to an $M_4$-period, and permitting only a switch to either an $M_1$-
or $M_2$-period.
For similar reasons as mentioned in case~2(c), the corresponding transition
probabilities depend on queue length differences that are affected by
the history of the process and are not visible on the fluid scale.
\end{remark}
\end{enumerate}

\subsubsection{$M_4$-period}

Assume that the system enters an $M_4$-period at time~$t$, then
\begin{enumerate}
\item[(a)] It spends a time period
$T_4(t) = \min\left\{\frac{Q_4(t)}{1 - \rho_4},
\frac{Q_5(t)}{1 - \rho_5}\right\}$ in $M_4$.
\item[(b)] During this period, the queues of nodes~4 and~5 drain at
a linear rate
\[
Q_i(t + u) = Q_i(t) - (1 - \rho_i) u, \mbox{ for } i = 4, 5,
\]
while the queues of nodes 1, 2, 3, 6 fill at a linear rate
\[
Q_i(t + u) = Q_i(t) + \rho_i u, \mbox{ for } i = 1, 2, 3, 6,
\]
$u \in [0, T_4(t)]$.
In particular, $\min\{Q_4(t + T_4(t)), Q_5(t + T_4(t))\} = 0$.

\item [(c)] At time $t + T_4(t)$, the system switches to either
an $M_2$- or $M_3$-period.
In order to determine which of these events can occur, we need to
distinguish between three cases, depending on whether
$\frac{Q_4(t)}{1 - \rho_4}$ is (i) larger than, (ii) equal to,
or (iii) smaller than $\frac{Q_5(t)}{1 - \rho_5}$.

In case~(i), i.e.,
$\frac{Q_4(t)}{1 - \rho_4} > \frac{Q_5(t)}{1 - \rho_5}$, we have
$Q_4(t + T_4(t)) > 0$, i.e., the queue of node~4 is still non-empty
by time $t + T_4(t)$, causing a switch to an $M_2$-period with
probability~1.

In case~(ii), i.e.,
$\frac{Q_4(t)}{1 - \rho_4} = \frac{Q_5(t)}{1 - \rho_5}$, we have
$Q_4(t + T_4(t)) = Q_5(t + T_4(t)) = 0$, i.e., the queues of nodes~4
and~5 both empty at time $t + T_4(t)$.
Even though both queues empty at the same time on the fluid scale,
there will with overwhelming probability be a long period in the
pre-limit where one of the nodes has become inactive for the first
time while the other one has yet to do so.
Since both nodes~4 and~5 must be inactive in order for nodes~1 and~2
to activate, these nodes have no chance to activate during that period,
but either node~3 or node~6 does, depending on whether node~5
or node~4 is the first one to become inactive.
As a result, the system cannot switch to an $M_1$-period, but only
to an $M_2$- or $M_3$-period.
In fact, a switch to $M_2$ will occur almost surely if node~5 is the
first one to become inactive, while a switch to $M_3$ will occur
almost surely if node~4 is the first one to become inactive.
The probabilities of these two scenarios, and hence the transition
probabilities to $M_2$ and $M_3$, depend on queue length differences
between nodes~4 and~5 at time~$t$ that are affected by the history
of the process and are not visible on the fluid scale.

In case~(iii), i.e.,
$\frac{Q_4(t)}{1 - \rho_4} < \frac{Q_5(t)}{1 - \rho_5}$, we have
$Q_5(t + T_4(t)) > 0$, i.e., the queue of node~5 is still non-empty
by time $t + T_4(t)$, forcing a switch to an $M_3$-period with
probability~1.
\end{enumerate}

\begin{remark}
As noted in the above description of the fluid limit process,
in cases 2(c), 3(c), and 4(c)(ii) the transition probabilities from
an $M_2$-period to an $M_1$- or $M_3$-period, from an $M_3$-period to
an $M_1$- or $M_2$-period, and from an $M_4$- to an $M_2$- or $M_3$-period,
depend on queue length differences that are affected by the history
of the process and are not visible on the fluid scale.
Depending on whether or not the initial state and parameter values
allow for these cases to arise, it may thus be impossible to provide
a probabilistic description the evolution of the resulting fluid limit
process, even it terms of its entire own history.
\end{remark}

\subsection{Overview of fluid limit proofs}
\label{over}

In the previous subsection we provided a description of the dynamics
of the fluid limit process once it has reached a natural state
and entered an $M_1$, $M_2$-, $M_3$ or $M_4$-period.
As was further stated, for $\rho$ sufficiently close to~1,
a natural state and such a period occur in uniformly bounded time
almost surely for any initial state.
The justification for all these properties follows from a series of
lemmas and propositions stated and proved in Appendices A--D.
In this subsection we present a high-level outline of the fluid limit
statements and proofs.

First of all, recall that the description of the fluid limit process
referred to the continuous-time Markov process representing the
system dynamics as introduced in Section~\ref{mode}.
For all the proofs of fluid limit properties and instability results
however we consider a rescaled linear interpolation of the uniformized
jump chain (as defined in Appendix A.I).
This construction yields convenient properties of the fluid limit paths
and allows us to extend the framework of Meyn~\cite{Meyn95} for
establishing instability results for discrete-time Markov chains.
(The original continuous-time Markov process has in fact the same fluid
limit properties, but this is not directly relevant in any of the proofs.)

The proofs of the fluid limit properties consist of four main parts.
Part~A identifies several basic properties of the fluid limit paths,
and in particular establishes that the queue length trajectory of each
of the individual nodes exhibits `sawtooth' behavior.
This fundamental property in fact holds in arbitrary interference graphs,
and only requires an exponent $\gamma > 1$ in the backoff probability.
Part~B of the proof shows a certain dominance property, saying that
if all the interferers of a particular node also interfere with some
other node that is currently being served at full rate, then the former
node must be empty or served at full rate (on the fluid scale) as well.
Under the assumption $\lambda_3 > \lambda_4$, $\lambda_5 < \lambda_6$,
the dominance property implies that after a finite amount of time the
fluid limit process for the broken-diamond network must always reside
in a natural state as defined in the previous subsection.
Part~C of the proof centers on the $M_1$-, $M_2$-, $M_3$-
and $M_4$-periods, and establishes that at the end of any such period,
the process immediately switches to one of the other types of periods
with the probabilities indicated in the previous subsection.
In particular, it is deduced that an $M_4$-period cannot be entered
from an $M_2$- or $M_3$-period, and must always be preceded
by an $M_1$-period once the process has reached a natural state.
The combination of the sawtooth queue length trajectories and the
switching probabilities provides a probabilistic description of the
dynamics of the fluid limit once the process has reached a natural
state and entered an $M_1$-, $M_2$-, $M_3$- or $M_4$-period.
Part~B already established that the process must always reside in
a natural state after a finite amount of time, but it remains to be
shown that the process will inevitably enter an $M_1$-, $M_2$-, $M_3$-
or $M_4$- period, which constitutes the final Part~D of the proof.
The core argument is that interfering empty and nonempty queues can
not coexist, since the empty nodes will frequently enter back-off periods,
offering the nonempty nodes abundant opportunities to gain access,
drain their queues, and cause the empty nodes to build queues in turn.

Part~A of the proof starts with the simple observation that, by the ``
skip-free'' property of the original pre-limit process, the sample
paths of the interpolated version of the uniformized jump chain are
Lipschitz continuous, and hence so are the sample paths of the
fluid-scaled process.
The fluid limit paths inherit the Lipschitz continuity,
and are thus differentiable almost everywhere with probability one.

Then fluid limit paths are determined by a countable set of `entrance'
times and `exit' times of $(0, \infty)$ with probability one.
The proof then proceeds to show that if a nonempty node (on the
fluid scale) receives any amount of service during some time interval,
then it must in fact be served at the full rate until it has completely
emptied (on the fluid scale), assuming $\gamma > 1$.
This implies that when node~$i$ is nonempty (on the fluid scale),
its queue must either increase at rate $\lambda_i$ or decrease at rate
$1 - \lambda_i$ until it has entirely drained.
In other words, the queue length trajectory of each of the individual
nodes exhibits sawtooth behavior (Theorem~\ref{thm_linearpaths}).

Part~B of the proof pertains to the joint behavior of the fluid limit
trajectories of the various queue lengths.
First of all, the natural property is proved that whenever
a particular node is served, none of its interferers can receive any
service (Lemma~\ref{lemma_conflict}).
Second, it is established that whenever a particular node is served,
any node whose interferers are a subset of those of the node served,
must either be empty or be served at full rate as well (on the fluid scale)
(Corollary~\ref{cor_dominate}).
For example, in the broken-diamond network, whenever node~3 is served,
node~4 must either be empty or be served at full rate as well,
and similarly for nodes~5 and~6.
These two properties combined yield a dominance property, saying that
if all the interferers of a particular node also interfere with some
other node that is currently being served at full rate, then the former
node must be empty or served at full rate (on the fluid scale) as well.
In the case of the broken-diamond network, under the assumption
$\lambda_3 > \lambda_4$, the queue of node~3 will therefore never be
smaller than that of node~4 after some finite amount of time,
and similarly for nodes~4 and~5.
Thus the fluid limit process will always reside in a natural state
after some finite amount of time.

Part~C of the proof focuses on the $M_1$-, $M_2$-, $M_3$-
and $M_4$-periods as described above.
Because of the sawtooth behavior, an $M_1$-period can only end when
both nodes~1 and~2 are empty (on the fluid scale).
Likewise, an $M_2$- or $M_3$-period can only end when node~3 or node~6
is empty, respectively.
An $M_4$-period can only end when node~4 or node~5 (or both) is empty.
It is then proven that at the end of an $M_1$-period, the fluid limit
process immediately switches to an $M_2$-, $M_3$- or $M_4$-period with
the probabilities specified in the previous subsection
(Theorem~\ref{thm_switchM0}).
When the process resides in a natural state, an $M_2$-period is always
instantaneously followed by an $M_1$- or $M_3$-period,
while an $M_3$-period is always instantaneously followed by an $M_1$-
or $M_2$-period.
In particular, it is concluded that an $M_4$-period cannot be entered
from an $M_2$- or $M_3$-period, and must always be preceded
by an $M_1$-period once the process has reached a natural state.
After an $M_4$-period, the process always immediately switches to
an $M_2$- or $M_3$-period.

There is no reason a priori however that the process is guaranteed to
actually ever enter an $M_1$-, $M_2$-, $M_3$- or $M_4$- period.
In fact, the process may very well spend time in different kinds of
states, but the final Part~D of the proof establishes that these
kinds of states are transient, and cannot occur once a natural state
has been reached, which is forced to happen in a finite amount of time
for particular arrival rates as was already shown in Part~B.
Note that an $M_1$-, $M_2$-, $M_3$- or $M_4$- period occurs as soon as
node~1, node~2, node~3, node~6 or nodes~4 and~5 simultaneously are
served at full rate.
In other words, the only ways for the process to avoid an $M_1$-,
$M_2$-, $M_3$- or $M_4$-period, are:
(i) for node~4 to be served at full rate, but not nodes~3 and~5;
(ii) for node~5 to be served at full rate, but not nodes~4 and~6;
(iii) for none of the nodes to be served at full rate.
Scenario~(i) requires node~3 to be empty (on the fluid scale)
and node~4 to be nonempty, which can not occur in a natural state.
Likewise, scenario~(ii) cannot arise in a natural state either.
Scenario~(iii) requires that every empty node~$i$ is served at rate~$\rho_i$
(on the fluid scale), while all nonempty nodes are served at rate~0.
Such a scenario is not particularly plausible,
but a rigorous proof turns out to be quite involved.
The insights rely strongly on the specific properties of the broken-diamond
network, and an extension to arbitrary graphs does not seem straightforward.
The core argument is that interfering empty and nonempty queues can
not coexist, since the empty nodes will frequently enter back-off periods,
offering the nonempty nodes abundant opportunities to gain access,
drain their queues, and cause the empty nodes to build queues in turn.

\section{Instability results for broken-diamond network}
\label{instab}

In the previous section we provided a detailed description of the
dynamics of the fluid limit process once it has reached a natural
state and entered an $M_1$-, $M_2$-, $M_3$ or $M_4$-period.
In this section we exploit the properties of the fluid limit process
in order to prove that it is unstable for $\rho$ sufficiently close to~1,
and then show how the instability of the original stochastic process
can be deduced from the instability of the fluid limit process.

\subsection{Instability of the fluid limit process}

In order to prove instability of the fluid limit process,
we first revisit the intuitive explanation discussed earlier,
see Figure~\ref{path2} for an illustration.
Denote $\rho_0 = \max\{\rho_1, \rho_2\}$, and recall that
$\rho_3 \geq \rho_4$ and $\rho_5 \leq \rho_6$ by assumption.
Since nodes 1, 2, 3 and 6 are only served during $M_1$-, $M_2$-
and $M_3$-periods, and not during $M_4$-periods, it is easily seen
that the fraction of time that the system spends in $M_1$-, $M_2$-
and $M_3$-periods must be no less than $\rho = \rho_0 + \rho_3 + \rho_6$
in order for these nodes to be stable.
Thus, the system cannot be stable if it spends a fraction of the
time larger than $1 - \rho$ in $M_4$-periods.
As it turns out, however, when the de-activation function is
sufficiently aggressive, e.g., $g(x) = \oo(x^{- \gamma})$,
with $\gamma > 1$, $M_4$-periods in fact persistently occur for
a fraction of time that does not tend to~0 as $\rho$ approaches~1,
which forces the system to be unstable.


Figure~\ref{path2} shows a fluid-limit sample path corresponding to
the switching sequence $M_1 \to M_2 \to M_1 \to M_4 \to M_3 \to M_1$.
The aggregate queue size starts building up in the $M_3$-period
that follows the $M_4$-period.

In order to prove instability of the fluid limit process, we adopt the
Lyapunov function $\Load(t) = \sum_{k = 1}^{3} \max_{i \in M_k} Q_i(t)$,
and will show that the load $\Load(t)$ grows without bound almost surely.
Note that the load $\Load(t)$ increases during $M_4$-periods while
the process is in a natural state.

In preparation for the instability proof, we first state two auxiliary
lemmas.
It will be convenient to view the evolution of the fluid limit process,
and in particular the Lyapunov function $\Load(t)$, over the course of
\emph{cycles}.
The $i$-th \emph{cycle} is the period from the start of the $(i - 1)$-th
$M_1$-period to the start of the $i$-th $M_1$-period
\emph{once the fluid limit process has reached a natural state}.
Denote by $t_i$ the start time of the $i$-th cycle, $i = 1, 2, \dots$.
Each $t_i$ is finite almost surely for $\rho$ sufficiently close to~1,
and in particular an infinite number of cycles must occur almost surely.
In order to see that, recall that the fluid limit process will reach
a natural state and enter an $M_1$, $M_2$, $M_3$- or $M_4$-period in
finite time almost surely for any initial state as stated in
Subsection~\ref{desc}.
The description of the dynamics of the fluid limit process provided in
that subsection then implies that $M_1$-periods and hence cycles must
occur infinitely often
(and if only finitely many $M_1$-periods occurred, then at least one
of the nodes would in fact never be served again after some finite time,
implying that the fluid limit process is unstable regardless).

The next lemma shows that the duration of a cycle and the possible
increase in the load over the course of a cycle are linearly
bounded in the load at the start of the cycle.

\begin{lemma}

\label{auxi1}

The duration of the $i$-th cycle, $\Delta t_i = t_{i + 1} - t_i$,
and the increase in the load over the course of the $i$-th cycle,
$ \Load(t_{i + 1}) - \Load(t_i) =
\Load(t_i + \Delta t_i) - \Load(t_i)$, are bounded from above by
\[
\Delta t_i \leq C_T \Load(t_i) \mbox{ and } \Load(t_{i + 1}) - \Load(t_i) \leq
C_{\Load} \Load(t_i),
\]
for all $\rho \leq 1$, where $C_T = \frac{1}{1 - \rho_3 - \rho_6}
\left(\frac{1}{1 - \rho_0} + \frac{1}{1 - \max\{\rho_4, \rho_5\}}\right)$
and $C_{\Load} = \frac{\rho}{1 - \max\{\rho_4, \rho_5\}}$.

\end{lemma}

The proof of the above lemma is presented in Appendix~E.

In order to establish that the durations of $M_4$-periods are
non-negligible, it will be useful to introduce the notion of
`weakly-balanced' queues, ensuring that the queues of nodes~4 and~5
are not too small compared to the queues of nodes~3 and~6.

\begin{definition}

Let $\beta^{\min}$ and $\beta^{\max}$ be fixed positive constants.
The queues are said to be weakly-balanced in a given cycle (with
respect to $\beta^{\min}$ and $\beta^{\max}$) if $\beta^{\min} \leq
\frac{Q_3(t)}{Q_5(t)}, \frac{Q_6(t)}{Q_4(t)} \leq \beta^{\max}$,
with $t$ denoting the time when the $M_1$-period ends that
initiated the cycle.

\end{definition}

The next lemma shows that over two consecutive cycles, the queues
will be weakly-balanced with probability at least~1/3.

\begin{lemma}

\label{auxi2}

Let
\[
\epsilon = \frac{\rho_2}{2 \left(\rho_2 + (\rho_3 + \rho_6)
\frac{1 - \min\{\rho_4, \rho_5\}}{1 - \max\{\rho_4, \rho_5\}}\right)} \geq
\frac{\rho_2}{\rho}
\frac{1 - \max\{\rho_4, \rho_5\}}{1 - \min\{\rho_4, \rho_5\}}.
\]
Then over two consecutive cycles, with probability at least~1/3,
the queues will be weakly-balanced in at least one of these cycles with
\[
\beta^{\max} =
\frac{\max\{\rho_3, \rho_6\} + (1 - \rho_2) (1 - \epsilon) / \epsilon}
{\min\{\rho_4, \rho_5\}},
\]
and $\beta^{\min} = \frac{1}{\beta^{\max}}$.

\end{lemma}

The proof of the above lemma is presented in Appendix~E.

As suggested by the above lemma, it will be convenient to consider
pairs of two consecutive cycles in order to prove instability of
the fluid limit process.

Let $D_k$ be the pair of cycles consisting of cycles $2 k - 1$ and $2 k$
as in Figure~\ref{cycle}, $k = 1, 2, \dots$.
With minor abuse of notation, denote by $T_k = t_{2 k - 1}$ the start
time of $D_k$ and $\Load_k = \Load(T_k)$.
Denote by $\Delta T_k = T_{k + 1} - T_k$ the duration of~$D_k$
and by $\Delta \Load_k = \Load_{k + 1} - \Load_k$ the increase
in $\Load(t)$ over the course of~$D_k$.

The next proposition shows that for $\rho$ sufficiently close to~1
the load cannot significantly decrease over a pair of cycles
and will increase by a substantial amount with non-zero-probability.
We henceforth assume
$(\rho_1, \rho_2, \rho_3, \rho_4, \rho_5, \rho_6) = \rho
(\kappa_1, \kappa_2, \kappa_3, \kappa_3 - \alpha, \kappa_6 - \alpha,
\kappa_6)$, with $\max\{\kappa_1, \kappa_2\} + \kappa_3 + \kappa_6 = 1$
and $0 < \alpha < \min\{\kappa_3, \kappa_6\}$, so that
$\rho = \rho_0 + \rho_3 + \rho_6$.

\begin{figure}
\centering
\includegraphics[width=4in]{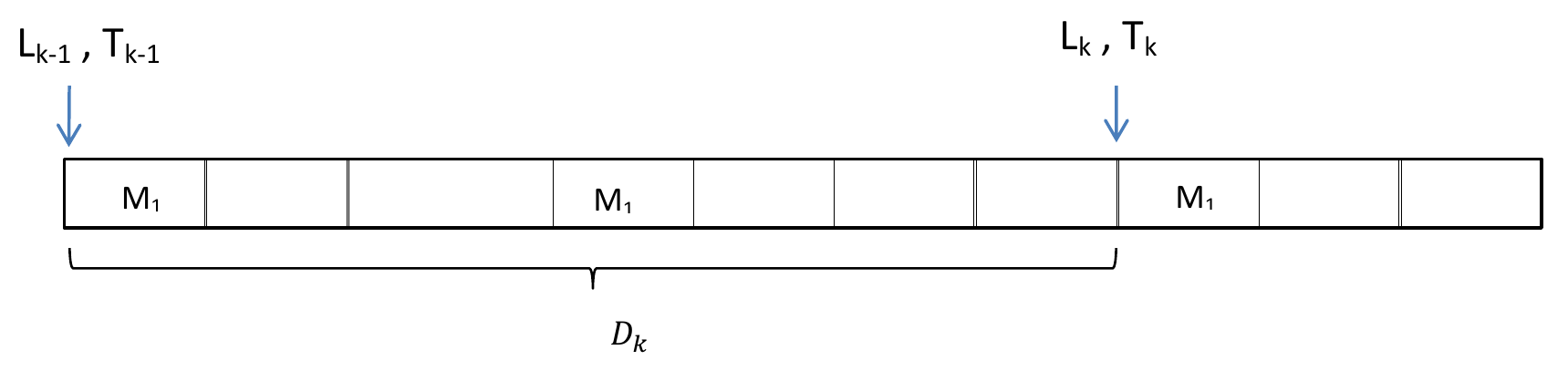}
\caption{A cycle $D_k$ consisting of a pair of consecutive cycles.}
\label{cycle}
\end{figure}

\begin{proposition}

\label{theorem:rho<1}

Let $C_{LT} = C_T (2 + C_\Load)$, with $C_T$ and $C_\Load$ as specified
in Lemma~\ref{auxi1}, $\theta = 1 - (1 - \rho) C_{LT}$, $p = 1/12$.
Over cycle pairs $D_k$, $k = 1, 2, \dots$,
\begin{itemize}
\item[(i)] $\Delta T_k \leq C_{LT} \Load_k$;
\item[(ii)] $\Load(t) \geq \theta \Load_k$ for all $t \in [T_k, T_{k+1}]$;
\item[(iii)] $\PP\bigl(\Load_{k + 1} - \theta \Load_k \geq
\delta(\rho) \theta \Load_k | \Load_k \bigr) \geq p$,
\end{itemize}
with $\delta(\rho)$ a constant, depending on~$\rho$,
and $\delta(\rho) \uparrow \delta =
\frac{1}{\beta^{\max}(1+\beta^{\max}) (1 + \alpha - \min\{\kappa_3, \kappa_6\})}$,
as $\rho \uparrow 1$.

\end{proposition}

\begin{proof}

We first show part~(i).
Using Lemma~\ref{auxi1}, we find
\[
\Delta T_k = \Delta t_{2 k - 1} + \Delta t_{2 k} \leq
C_T (\Load(t_{2 k - 1}) + \Load(t_{2 k})) \leq C_T (2 + C_{\Load}) \Load_k.
\]

In order to prove part~(ii), note that $\Load(t)$ cannot decrease at
a larger rate than $1 - \rho$, so that in view of part~(i),
\[
\Load(t) \geq \Load_k - (1 - \rho) (t - T_k) \geq
\Load_k - (1 - \rho) \Delta T_k \geq (1 - (1 - \rho) C_{LT}) \Load_k =
\theta \Load_k,
\]
for all $t \in [T_k, T_{k + 1}]$.

We now turn to part~(iii).
Suppose that the following event occurs:
the queues are weakly-balanced at the end of an $M_1$-period, say
time~$\tau$, during~$D_k$
(which according to Lemma~\ref{auxi2}) happens with at least
probability~1/3) and the system then enters an $M_4$-period
(which happens with probability~1/4).
Recalling that $\rho_3 > \rho_4$, $Q_3(t) \geq Q_4(t)$,
$\rho_5 < \rho_6$ and $Q_5(t) \leq Q_6(t)$, we find that during the
$M_4$-period $\Load(t)$ increases by
\[
\rho \min\left\{\frac{Q_4(\tau)}{1 - \rho_4}, \frac{Q_5(\tau)}{1 - \rho_5}\right\} \geq
\rho \frac{\min\{Q_4(\tau), Q_5(\tau)\}}{1 - \rho \min\{\kappa_3, \kappa_6\} + \rho \alpha}.
\]
Since the queues are weakly-balanced, we deduce
$Q_3(\tau) \leq \beta^{\max} Q_5(\tau) \leq \beta^{\max} Q_6(\tau)
\leq (\beta^{\max})^2 Q_4(\tau)$
and $Q_6(\tau) \leq \beta^{\max} Q_4(\tau) \leq \beta^{\max} Q_3(\tau)
\leq (\beta^{\max})^2 Q_5(\tau)$.
Noting that $Q_1(\tau) = Q_2(\tau) = 0$, we obtain
$$
\Load(\tau) = Q_3(\tau) + Q_6(\tau) \leq (1+\beta^{\max})Q_6(\tau) \leq
\beta^{\max}(1+\beta^{\max})Q_4(\tau),
$$
and also
$$
\Load(\tau) = Q_3(\tau) + Q_6(\tau) \leq (1+\beta^{\max})Q_3(\tau) \leq
\beta^{\max}(1+\beta^{\max})Q_5(\tau).
$$
So
$$
\Load(\tau) \leq \beta^{\max}(1+\beta^{\max}) \min \{Q_4(\tau), Q_5(\tau)\},
$$
and thus the increase in $\Load(t)$ during the $M_4$-period is no less
than $\delta(\rho) \Load(\tau)$, with
\[
\delta(\rho) = \frac{\rho}{\beta^{\max}(1+\beta^{\max}) (1 - \rho \min\{\kappa_3, \kappa_6\} + \rho \alpha)}.
\]
Using part~(i) once again, we conclude that with at least probability 1/12,
\ben
\Load_{k + 1} &\geq &\Load_k + \delta(\rho) \Load(\tau) - (1 - \rho) \Delta T_k \geq
\Load_k + \delta(\rho) (\Load_k - (1 - \rho) \Delta T_k) - (1 - \rho) \Delta T_k \\
&=& (1 + \delta(\rho)) (\Load_k - (1 - \rho) \Delta T_k) \geq
(1 + \delta(\rho)) (\Load_k - (1 - \rho) C_{LT} \Load_k) \\
&= & (1 + \delta(\rho)) \theta \Load_k.
\een
\end{proof}

Armed with the above proposition, we now proceed to prove that the
fluid limit process is unstable, in the sense that $\Load(T) \to \infty$
as $T \to \infty$.
In fact, $\Load(T)$ grows faster than any sub-linear function
$T^{\frac{1}{m}}$, $m>1$, as stated in the next theorem.

\begin{theorem}

\label{expectedrate}

For any $m > 1$, there exists a constant
$\rho^* = \rho^*(\kappa, m) < 1$, such that for all
$\rho \in (\rho^*, 1]$,
\[
\limsup_{T \to \infty} \expect{\frac{T}{\Load^m(T)}} = 0,
\]
for any initial state ${\bf Q}(0)$ with $\Lonenorm{{\bf Q}(0)} = 1$,
and $\Lonenorm{\cdot}$ denoting the $L_1$-norm.

\end{theorem}


\begin{proof}

Consider the cycle pairs $D_k$, $k = 1, 2, \dots$, as defined right
before Proposition~\ref{theorem:rho<1}.
Assume $\rho \in (1 - \frac{1}{C_{LT}}, 1]$, so that $\theta \in (0, 1]$
in Proposition~\ref{theorem:rho<1}.
For any time $t > T_1$, we can define a stopping time $N_t$ such that
$T_{N_t} < t \leq T_{N_t+1}$, i.e., $t$ is within the $N_t$-th cycle pair.
(This is possible almost surely, since $T_k \to \infty$ as $k \to
\infty$ almost surely, as will be proven below.)
Recall that $T_{N_t+1} \leq T_{N_t} + C_{LT} \Load_{N_t}$
and $\Load(t) \geq \theta \Load_{N_t}$ by parts~(i) and~(ii) of
Proposition~\ref{theorem:rho<1}, respectively, and trivially
$\Load_{N_t} \leq \Load(0) + \rho T_{N_t} \leq 2 T_{N_t}$ for $t$
sufficiently large.
Thus,
\be
 \limsup_{t \to \infty} \expect{t \Load^{- m}(t)} &\leq &
\limsup_{t \to \infty} \expect{T_{N_t+1} \theta^{- m} \Load^{- m}_{N_t}} \nonumber \\
&\leq & \theta^{- m} \limsup_{t \to \infty}
\expect{T_{N_t} \Load^{- m}_{N_t}} +
\theta^{- m} C_{LT} \limsup_{t \to \infty} \expect{\Load^{- m + 1}_{N_t}} \nonumber \\
&\leq& \theta (1 + 2 C_{LT}) \limsup_{t \to \infty}
\expect{T_{N_t} \Load^{- m}_{N_t}}.
\label{limsup}
\ee
So it suffices to prove that there exists
$\rho^* = \rho^*(\kappa, m) < 1$ such that (\ref{limsup}) is zero
for $\rho > \rho^*$, which we now proceed to show.

First of all, by Proposition~\ref{theorem:rho<1}, for any $m>0$,
\be
\expect{\Load_{k+1}^{- m}|\FF_k}
&\leq&
(1 - p) (\theta \Load_k)^{- m} + p ((\theta+\delta) \Load_k)^{- m} \nonumber \\
&=& \alpha_m \Load_k^{- m},
\label{martingale}
\ee
where $\FF_k$ is a suitable filtration
and $\alpha_m := (1 - p) \theta^{- m} + p (\theta+\delta)^{- m}$.

Since $\theta(\rho) \to \theta(1) = 1$
and $\delta(\rho) \to \delta(1) = \delta > 0$ as $\rho \uparrow 1$,
$\alpha_m(\rho)$ is a continuous function of $\rho$ in the vicinity
of~1.
Because $\alpha_m(1) < 1$, there must exist
a $\rho_m^* = \rho^*(\kappa, m) < 1$ such that $\alpha_m < 1$ for
all $\rho > \rho^*$.
This shows that, for $\rho > \rho_m^*$, $\Load^{- m}_k$ is a positive
(geometric) supermartingale with parameter $\alpha_m < 1$.
Taking expectations on both sides of~(\ref{martingale}) yields
\be
\expect{\Load^{- m}_k} \leq \alpha_m^k \Load^{- m}_0.
\label{eq:L_k martingale}
\ee
with $\Load_0 = \Load(t_{i_0}) > 0$ as noted earlier.
In particular, $ \lim_{k \to \infty} \expect{\Load^{- m}_k} = 0$,
and $1 / \Load_k \to 0$ almost surely as $k \to \infty$ by the Doob's
supermartingale-convergence Theorem (page 147 of~\cite{RW89B}).
This implies that $T_k \to \infty$ almost surely because
$\Load_k \leq \rho T_k + 1 \leq T_k + 1$.
Therefore, the stopping time $T_{N_t}$ is well-defined.

Next, consider the sequence of random variables $T_k \Load_k^{- m}$.
Using Proposition~\ref{theorem:rho<1},
\be
\expect{T_k \Load_k^{- m}|\FF_{k-1}}
&\leq&
(T_{k-1} + C_{LT} \Load_{k-1}) \expect{\Load_k^{-m}|\FF_{k-1}} \nonumber \\
&\leq& (T_{k-1} + C_{LT} \Load_{k-1}) \alpha_m \Load_{k-1}^{- m} \nonumber \\
&=&
\alpha_m T_{k-1} \Load_{k-1}^{- m} + \alpha_m C_{LT} \Load_{k-1}^{- m + 1}.
\label{almost-martingale}
\ee
Define $\epsilon_k := C_{LT} \alpha_m \Load_k^{- m + 1}$, then,
by~(\ref{martingale}) and~(\ref{eq:L_k martingale}), $\epsilon_k$
is a positive (geometric) super-martingale with parameter
$\alpha_{m - 1} < 1$ for $\rho > \rho_{m - 1}^* = \rho^*(\kappa, m - 1)$.
Then, $\sum_{k=1}^\infty \expect{\epsilon_k} \leq
C_{LT} \alpha_m \sum_{k=1}^\infty \alpha_{m - 1}^k < \infty$, which shows
that $\lim_{k \to \infty} T_k \Load_k^{- m} = 0$ almost surely.
In particular, define $\alpha := \max(\alpha_m, \alpha_{m - 1})$
and $\rho^* = \max\{\rho_m^*, \rho_{m - 1}^*\}$, then taking
expectations on both sides of~(\ref{almost-martingale}) yields
\be
\expect{T_k \Load_k^{-m}} \leq
\alpha \expect{T_{k-1} L_{k-1}^{-m}} + \alpha C_{LT} \alpha^{k-1},
\ee
which, by induction, shows that
\be
\expect{T_k \Load_k^{- m}} \leq \alpha^{k-1} (\expect{T_1 \Load_1^{- m}} + C_{LT} (k-1) \alpha),
\ee
for $\rho \in (\rho^*, 1]$.
Now observe that $T_1$ is strictly bounded and $\Load_1$ is
bounded away from zero, since a natural state is reached in finite time, before the system can empty, almost surely.
It then follows that $\lim_{k \to \infty} \expect{T_k \Load^{-m}_k} = 0$,

The fact that $T_k \Load_k^{- m}$ converges in $\mathcal{L}_1$ implies
that the sequence of random variables $T_k \Load_k^{- m}$ is
\textit{Uniformly Integrable} (UI) (page 147, Theorem 50.1 of~\cite{RW89B}).
It therefore follows, by adapting the arguments of Doob's optional
sampling theorem (page 159 of~\cite{RW89B}), that the family of random
variables $\{T_{N_t} \Load^{- m}_{T_{N_t}}\}$ is also UI.
Thus by definition, given $\varepsilon >0$, there exists
$K_\varepsilon$ such that
\[
\expect{T_{N_t} \Load^{- m}_{N_t}
\indi{T_{N_t} \Load^{-m}_{N_t} \geq K_\varepsilon}} \leq
\varepsilon, ~~ \forall t > 0
\]
We deduce
\begin{align*}
\expect{T_{N_t} \Load^{- m}_{N_t}}
&\leq \sum_{k=1}^\infty \expect{T_k \Load^{- m}_k \indi{N_t = k}
\indi{T_{N_t} \Load^{- m}_{N_t} \leq K_\varepsilon}} + \varepsilon \nonumber \\
&\leq K_\varepsilon \pr{N_t \leq D} +
\sum_{k=D+1}^\infty C_{LT} k \alpha^{k-1} + \varepsilon.
\end{align*}
Fixing~$\varepsilon$ and~$D$, we find that
\[
\limsup_{t \to \infty} \expect{T_{N_t} \Load^{- m}_{N_t}} \leq
(D+1) \frac{\alpha^{D}}{1-\alpha} + \varepsilon
\]
by the Monotone Convergence Theorem~\cite{Billingsley68}, and thus,
letting $D \to \infty$ and $\varepsilon \to 0$, we have
$\limsup_{t \to \infty} \expect{T_{N_t} \Load^{-m}_{N_t}} = 0$ for
$\rho > \rho^*$.
\end{proof}
\begin{corollary}

\label{coro}

For any $m>1$, there exists a constant
$\rho^* = \rho^*(\kappa, m) < 1$, such that for all
$\rho \in (\rho^*, 1]$,
\[
\liminf_{T \to \infty} \frac{\Load(T)}{T^{1/m}} = \infty,
\]
almost surely for any initial state ${\bf Q}(0)$ with $\Lonenorm{{\bf Q}(0)} = 1$.

\end{corollary}

\begin{proof}

Note that for any initial state ${\bf Q}(0)$ with $\Lonenorm{{\bf Q}(0)} = 1$,
\[
\liminf_{T \to \infty} \frac{\Load(T)}{T^{1/m}} \geq
\liminf_{k \to \infty} \frac{\theta \Load_k}{T_{k+1}^{1/m}},
\]
as can be seen from Proposition \ref{theorem:rho<1},
and so it suffices to show that $\limsup_k T_{k+1} \Load_k^{-m} = 0$.
But $T_{k+1} \leq T_k + C_{LT} \Load_k$, thus,
\be
\limsup_{k \to \infty} T_{k+1} \Load_k^{-m} \leq
\limsup_{k \to \infty} T_k \Load_k^{-m} + C_{LT} \limsup_{k \to \infty}
\Load_k^{-m+1}.
\ee
The right-hand side is zero because, as we saw in the proof of
Theorem~\ref{expectedrate}, both $T_k \Load_k^{-m}$ and $\Load_k^{-m+1}$
converge to zero almost surely for $\rho \in (\rho^*(\kappa, m), 1]$.

\end{proof}

\subsection{Instability of the original stochastic process}

In Theorem~\ref{expectedrate} we established that the fluid limit
process in unstable, in the sense that $\Load(T) \to \infty$ as
$T \to \infty$.
We now proceed to show how the instability of the original stochastic
process can be deduced from the instability of the fluid limit process.
The original stochastic process is said to be unstable when
$\{(U(t), \Qorig(t))\}_{t \geq 0}$ is transient,
and $\|\Qorig(t)\| \to \infty$ almost surely for any initial
state~$\Qorig(0)$.

We will exploit similar arguments as developed in Meyn~\cite{Meyn95}.
A notable distinction is that the result in~\cite{Meyn95} requires
that a suitable Lyapunov function exhibits strict growth over time.
In our setting the fluid limit is random, and the growth behavior
as stated in Theorem~\ref{expectedrate} is not strict, but only in
expectation and in an asymptotic sense, which necessitates a somewhat
delicate extension of the arguments in~\cite{Meyn95}.

The next theorem states the main result of the present paper,
indicating that aggressive deactivation functions cause the network
of Figure~\ref{network2} to be unstable for load values~$\rho$
sufficiently close to~1.

\begin{theorem}

\label{main}

Consider the network of Figure~\ref{network2}, and suppose that
$f_i(x) \equiv 1$, $x \geq 1$, and $g_i(x) = \oo(x^{- \gamma})$,
with $\gamma > 1$.
Let $(\rho_1, \rho_2, \rho_3, \rho_4, \rho_5, \rho_6) = \rho
(\kappa_1, \kappa_2, \kappa_3, \kappa_3 - \alpha, \kappa_6 - \alpha,
\kappa_6)$, with $\max\{\kappa_1, \kappa_2\} + \kappa_3 + \kappa_6 = 1$, and $0 < \alpha < \min\{\kappa_3, \kappa_6\}$.
Then there exists a constant
$\rho^\star(\kappa, \alpha) < 1$, such that for all
$\rho \in (\rho^\star(\kappa, \alpha), 1]$:
\[
\lim_{\|\Qorig(0)\| \to \infty}
\mathbb{P}_{\Qorig(0)}\{\liminf_{t \to \infty} \|\Qorig(t))\| = \infty\} = 1.
\]

\end{theorem}

Since our Markov Chain is irreducible, Theorem immediately implies that it is transient. The proof of Theorem~\ref{main} relies on similar arguments as
developed in the proof of Theorem~3.2 in~\cite{Meyn95}.
A crucial role is played by Theorem~3.1 of~\cite{Meyn95}, which is
reproduced below for completeness.

\begin{theorem}

\label{meyn}

Suppose that for a Markov chain $\{X(n); n = 0, 1, 2, \dots\}$ with
discrete state space~$S$, there exist positive functions $W(\cdot)$
and $\Delta(\cdot)$ on~$S$, and a positive constant $c_0$, such that
\be
\expect{W(X(n+1))|\FF_n} \leq W(X(n)) - \Delta(X(n)),
\label{sup-martingale}
\ee
whenever $X(n) \in S_{c_0} = \{x \in S: W(x) \leq c_0\}$,
with $\FF_n := \sigma(X(0), X(1), \dots, X(n))$.
Then for all $x \in S$,
\[
\PP_x\left\{\sum_{n=0}^{\infty} \Delta(X(n)) < \infty \right\} \geq
1 - W(x) / c_0.
\]

\end{theorem}

In order to apply the above theorem, we need to construct suitable
functions $W(\cdot)$ and $\Delta(\cdot)$.
The proof details are presented in Appendix~E.

\begin{remark}

Recall that the class of deactivation functions
$g_i(x) = \oo(x^{- \gamma})$ includes the random-capture scheme
with $g(x) \equiv 0$, $x \geq 1$, as considered in~\cite{FPR10}.
The result in Theorem~\ref{main} thus disproves the conjecture that the
random-capture scheme is throughput-optimal in arbitrary topologies.

\end{remark}

\section{Simulation experiments}
\label{simu}

We now discuss the simulation experiments that we have conducted to
support and illustrate the analytical results.
Consider the broken-diamond network as depicted Figure~\ref{network2}
and considered in the previous sections.
In the simulation experiments, the relative traffic intensities are
assumed to be $\kappa_1 = \kappa_2 = 0.4$, $\kappa_3 = 0.4$,
and $\kappa_6 = 0.2$ with $\alpha = 0$, for the components $M_1$, $M_2$,
and $M_3$, respectively, with a normalized load of $\rho = 0.97$.
At each node $i$, the initial queue size is $\Qorig_i(0) = 500$, the
activation function is $f_i(x) \equiv 1$, $x \geq 1$, and the de-activation
function is $g_i(x) = (1+x)^{-\gamma}$, where we set $\gamma = 2$.

Figure~\ref{sim1} plots the evolution of the queue sizes at the
various nodes over time, and shows that once a node starts
transmitting, it will continue to do so until the queue lengths of
all nodes in its component have largely been cleared.
This characteristic, and the associated oscillations in the queues,
strongly mirror the qualitative behavior displayed by the fluid limit.

\begin{figure}
\centering
\includegraphics[width=4in]{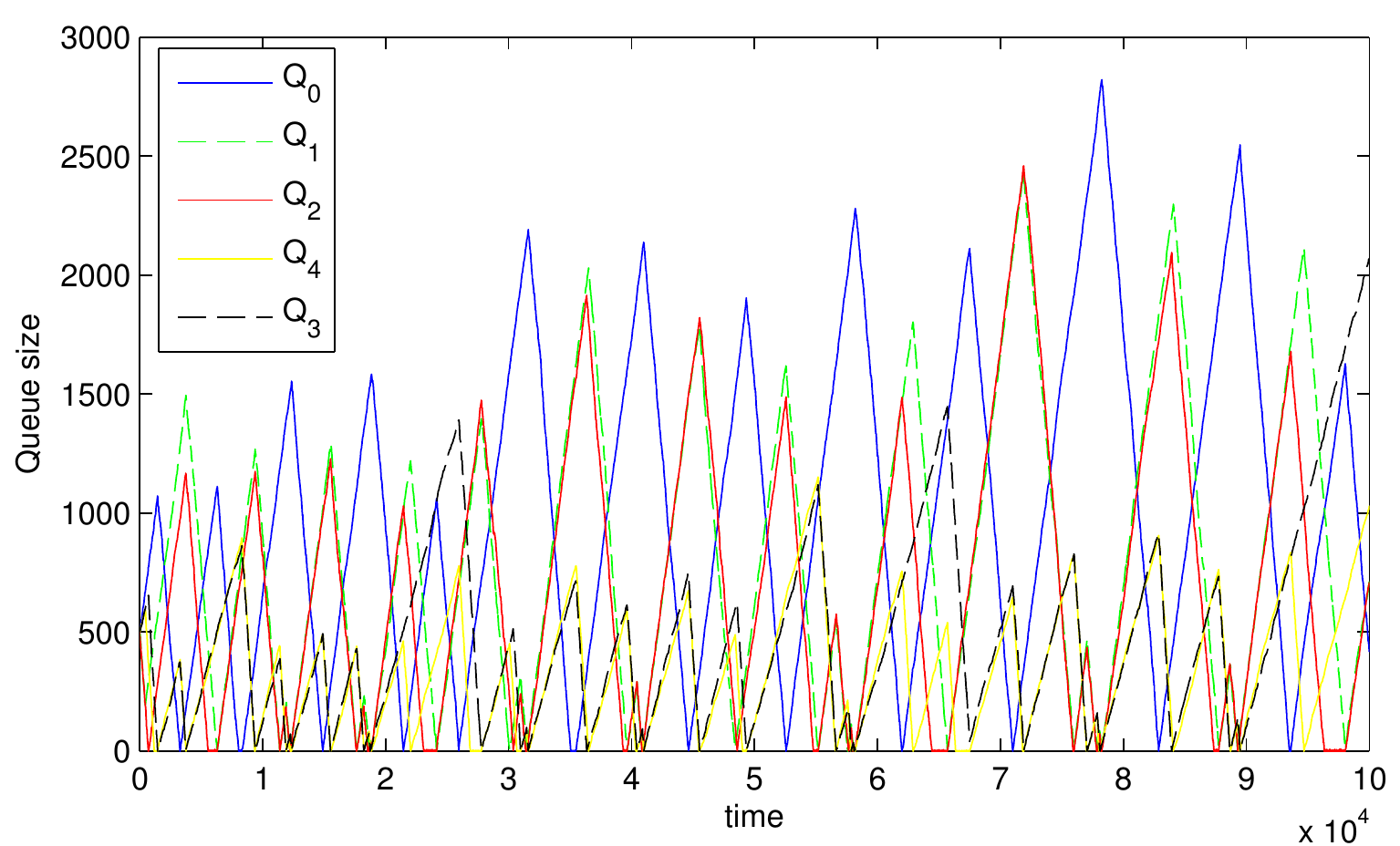}
\caption{Queue sizes at the various nodes as function of time for
the network of Figure~\ref{network2}.}
\label{sim1}
\end{figure}

\begin{figure}
\centering
\includegraphics[width=4in]{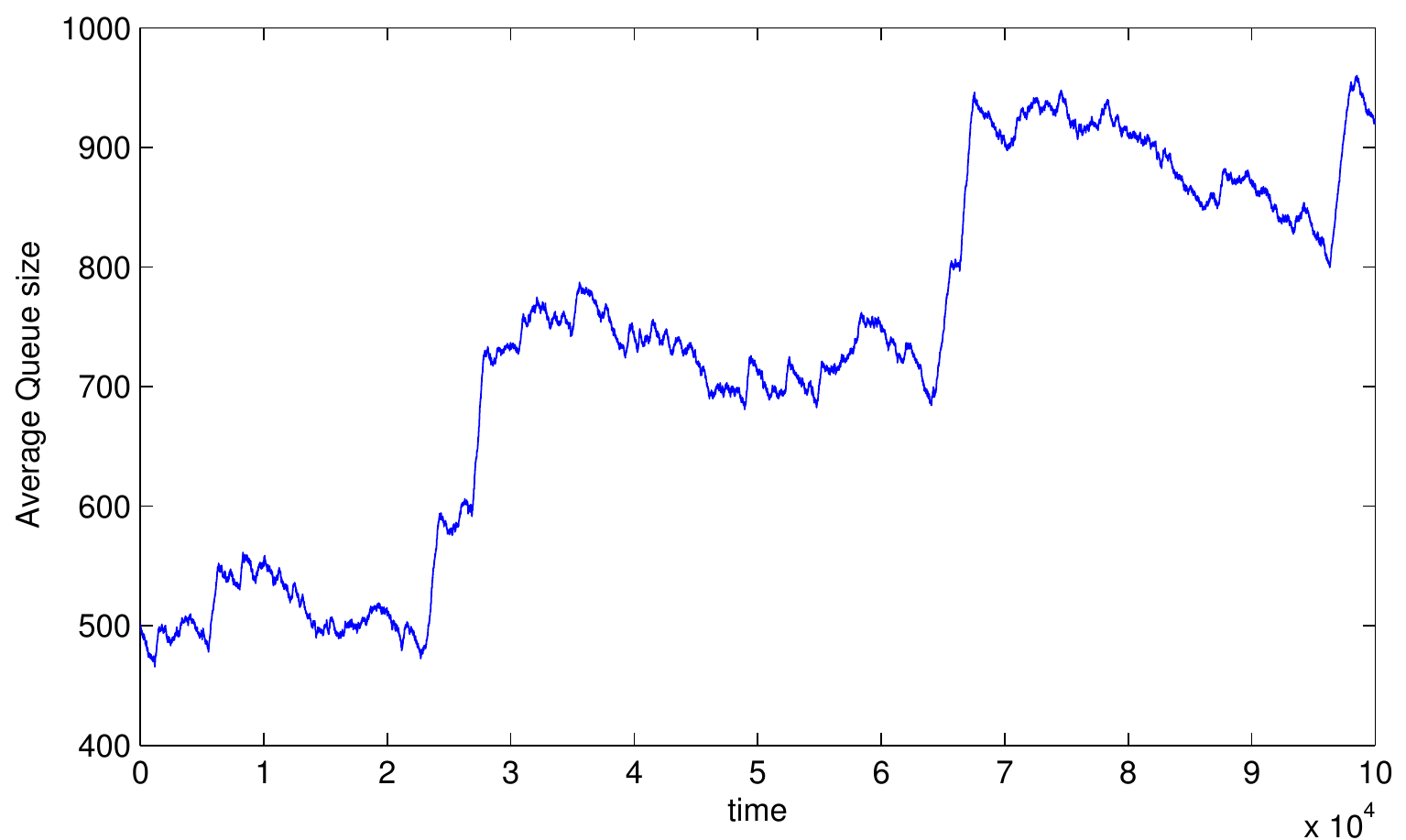}
\caption{Node-average queue size as function of time for the
network of Figure~\ref{network2}.}
\label{sim2}
\end{figure}

Although Figure~\ref{sim1} suggests an upward trend in the overall
queue lengths, the fluctuations make it hard to discern a clear picture.
Figure~\ref{sim2} therefore plots the evolution of the node-average
queue size over time, and reveals a distinct growth pattern.
Evidently, it is difficult to make any conclusive statements
concerning stability/instability based on simulation results alone.
However, the saw-tooth type growth pattern in Figure~\ref{sim2}
demonstrates strong signs of instability, and corroborates the
qualitative growth behavior exhibited by the fluid limit.
Indeed, careful inspection of the two figures confirms that the large
increments in the node-average queue size occur immediately after
$M_4$-periods, exactly as predicted by the fluid limit.
We further observe that in between these periods, the node-average
queue size tends to follow a slightly downward trend, consistent
with the negative drift of rate $(\rho-1)/3$ in the fluid limit.

\section{Concluding remarks and extensions}
\label{conc}

We have used fluid limits to demonstrate the potential instability
of queue-based random-access algorithms.
For the sake of transparency, we focused on a specific six-node
network and super-linear activity functions.
Similar instability issues can however arise in a far broader class of
interference graphs, as we will discuss in Subsection~\ref{generalgraph1} below.
The proof arguments further suggest that instability can in fact
occur for any activity factor that grows as a positive power
$1 / K$ of the queue length for network sizes of order~$K$,
as will be described in Subsection~\ref{polyfunction1}.

\subsection{Instability in general interference graphs}
\label{generalgraph1}

The instability of random access, with aggressive de-activation functions,
is not restricted to the broken-diamond network, and can arise in many
other interference graphs.
Consider a general interference graph $G = (V, E)$.
Without loss of generality, we can assume $G$ is connected,
because otherwise we can consider each connected subgraph separately.
For $\gamma >1$, the fluid limit sample paths still exhibit the
\textit{sawtooth behavior}, i.e., when a node starts transmitting,
it does not release the channel until its entire queue is cleared
(on the fluid scale).
Let $\calM = \{M_1, \dots, M_K\}$ denote the set of maximal
independent sets (maximal schedules) of~$G$.
We say the network operates in $M_i$ if a subset $W \subseteq M_i$ of
nodes are served at full rate (on the fluid scale), and $W$ does not
belong to any other maximal schedules $M_j$, $j \neq i$.
Under the random-access algorithm, at any point in time the network
operates in one of the maximal schedules and switches to another
maximal schedule when one or several of the queues in the current
maximal schedule drain (on the fluid scale).
More specifically, assume the network operates in a maximal schedule $M_i$.
If $M_i$ interferes with all other maximal schedules, i.e.,
$M_i \cap M_j = \emptyset$ for all $1 \leq j \leq K$, $j \neq i$,
then a transition from $M_i$ to any maximal schedule $M_j$, $j \neq i$,
is possible when all the queues in $M_i$ hit zero (on the fluid scale).
On the other hand, if $M_i$ overlaps with a subset of maximal schedules
$\calM_i^\prime:= \{M_j \in \calM: M_i \cap M_j \neq \emptyset\}$,
then the activity process can make a transition to $M_j \in \calM_i^\prime$
when all the queues in $M_i \backslash M_j$ drain (on the fluid scale).

The capacity region of the network is the convex set
$\calC = \mbox{conv}(S)$, which is full-dimensional because all
the basis vectors of $\PRR^N$ belong to that set.
The incidence vectors of the sets $\calM$ correspond to the extreme
points of $\calC$ as they can not be expressed as convex combinations
of other points.
Consider a covering of $V = \{1,2,\dots,N\}$ using the maximal schedules.
Formally, a set cover $C$ of $V$ is a collection of maximal schedules
such that $V \subseteq \cup_{M_i \in C} M_i$.
A set cover $C$ is minimal if removal of any of the elements
$M_i \in C$ leaves some nodes of $V$ uncovered.
Consider the class of graphs in which $|C| \leq K-1$ for some minimal
set cover $C$, i.e., we do not need all $M_i$'s for covering $V$.
Without loss of generality,
let $\calM^*=\{M_1, M_2, \dots, M_{K^*}\}$ denote such a minimal cover
with $K^* \leq K-1$.
Consider a (strictly positive) vector of arrival rates
$\lambda = \rho \sum_{i=1}^{K^*} \sigma_i 1_{M_i}$ where $\sigma_i > 0$,
$1 \leq i \leq K^*$, such that $\sum_{i=1}^{K^*} \sigma_i = 1$,
and $0 < \rho < 1$ is the load factor.
Hence, a centralized algorithm can stabilize the network by scheduling
each $M_i \in \calM^*$ for at least a fraction $\rho \sigma_i$ of the time.
However, under the random-access algorithm, the network might spend
a non-vanishing fraction of time in the schedules
$\calM \backslash \calM^*$, which can cause instability as $\rho$
approaches~1.
This phenomenon is easier to observe in graphs with a \textit{unique}
minimal set cover $\calM^*$ and with a maximal schedule $M_1$ interfering
with all the other maximal schedules, hence $M_1 \in \calM^*$.

\begin{figure}[!t]
\label{examples}
\centering
\subfigure[$\{1\}, \{2,3\}, \{4,5\}$]{\label{unstable1}
\includegraphics [width = 1.5 in]{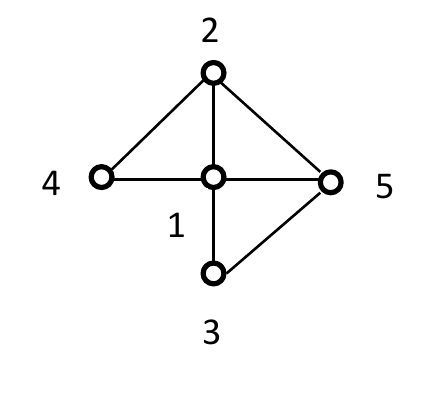}
}
\subfigure[$\{1\}, \{2,3,4\}, \{5,6\}$]{\label{unstable2}
\includegraphics [width = 1.5 in]{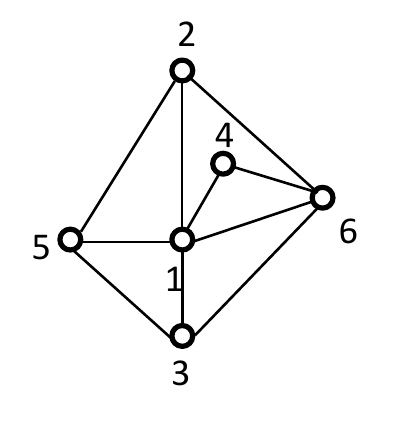}
}
\subfigure[$\{1\}, \{2,3,4\}, \{5,6\}$]{\label{unstable3}
\includegraphics [width = 1.5 in]{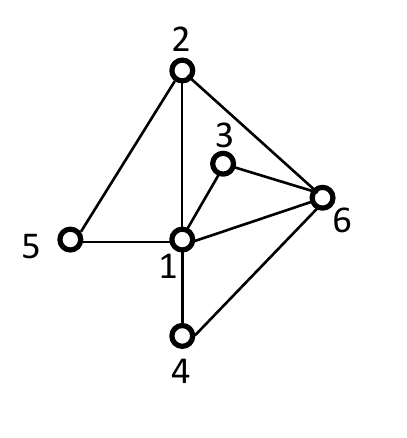}
}
\subfigure[$\{1\}, \{2,3,4\}, \{5,6,7\}$]{\label{unstable4}
\includegraphics [width = 2 in]{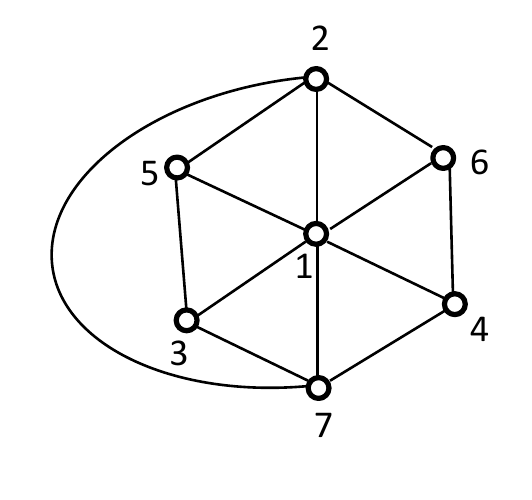}
}
\caption{A few unstable networks with their unique minimal cover
$\calM^*$ using the maximal schedules.}
\end{figure}

This means any valid covering of $V$ must contain $\calM^*$.
Therefore, considering arrival rate vectors of the form
$\lambda = \rho \sum_{i=1}^{K^*} \sigma_i 1_{M_i}$, $\sigma_i > 0$,
$\sum_{i=1}^{K^*} \sigma_i = 1$, \textit{the only way} to stabilize the
network is to use $M_i$ for a time fraction greater than $\rho \sigma_i$.
Visits to $M_1$ have to occur infinitely often, otherwise the network
is trivially unstable, and at the end of such visits, a transition
to any other maximal schedule is possible, including the schedules
in $\calM \backslash \calM^*$ with positive probability.
Then, upon entrance to schedules in $\calM \backslash \calM^*$,
the network spends a positive time in such schedules because the queues
in $\calM \backslash \{M_1\}$ build up during visits to $M_1$.
Hence, the arguments in the instability proof of the broken-diamond
network can be extended to such networks, although a rigorous proof
of the fluid limits in such general cases remains a formidable task.
Figure~\ref{examples} shows a few examples of such unstable networks
with unique minimal set covers.

\subsection{Instability for de-activation functions with polynomial decay}
\label{polyfunction1}

Consider any unstable network $G = (V,E)$, for example the broken-diamond
network or a graph as described in the previous subsection.
Let $\mathcal{I}(i)$ denote the set of neighbors of node~$i$ in~$G$.
We construct a \textit{$k$-duplicate graph} $G^{(k)}$, $k \in \PNN$,
of~$G$ as follows.
For each node $i \in V$, add $k$~duplicate nodes
$d^{(i)}_1, \dots, d^{(i)}_k$ to the graph, with the same arrival rate
$\lambda_i$ and the same initial queue length $\Qorig_i(0)$, such that
each node $d^{(i)}_j$ is connected to all the neighbors of node~$i$
and their duplicates, i.e.,
$\mathcal{I}(d^{(i)}_j) = \mathcal{I}(i) \cup_{l \in \mathcal{I}(i)} \{d^{(l)}_1, \cdots,d^{(l)}_k\}$,
for all $1 \leq j \leq k$.
For notational convenience, we define
$D^{(k)}_i := \{i, d^{(i)}_1, \dots,d^{(i)}_k\}$ and call it the
duplicate collection of node~$i$.
Note that the duplicate graph has the same number of maximal schedules
as the original graph.
In fact, each maximal schedule $M_i^{(k)}$ of $G^{(k)}$ consists of
nodes in the maximal schedule $M_i$ of~$G$ and their duplicates, i.e.,
$M_i^{(k)} = \cup_{l \in M_i} D^{(k)}_l$.
Next, we show that the duplicate graph is unstable for de-activation
functions that decay as $\oo(x^{-\gamma})$, for $\gamma > 1/(k+1)$.
Essentially, for such a range of~$\gamma$, each duplicate collection
acts as a super node with $\gamma > 1$, i.e., (i) if one of the queues
in a duplicate collection $D^{(k)}_i$ starts growing, all the queues
in $D^{(k)}_i$ grow linearly at the same rate~$\lambda_i$ (on the
fluid scale), (ii) if a nonempty queue in $D^{(k)}_i$ starts draining,
then all the queues in $D^{(k)}_i$ drain at full rate until they all
hit zero (on the fluid scale).
Then the instability follows from that of the original network,
as we can simply regard the duplicate collections as super nodes.
An informal proof of claims~(i) and~(ii) is presented below.

Claim~(i) is easy to prove as all the queues in a duplicate collection
share the same set of conflicting neighbors and the fact that one of
the queues grows, over a small time interval, implies that some
conflicting neighbors are transmitting over such interval.
To show~(ii), note that if one of the queues in the duplicate
collection drains over a non-zero time interval, no matter how small
the interval is, all the conflicting neighbors must be in backoff for
$O(R)$ units of time in the pre-limit process.
This guarantees that all the queues in the duplicate collection will
start a packet transmission during such interval almost surely.
As long as the duplicate collection does not lose the channel,
each queue of the collection follows the fluid limit trajectory of
an M/M/1 queue.
Suppose all the queues of the duplicate collection are above
a level~$\epsilon$ on the fluid scale for some fixed small $\epsilon > 0$.
Thus, in the pre-limit process, the amount of time required for the
queues to fall below a threshold $\epsilon R$ is $O(R)$ with high
probability as $R \to \infty$.
The duplicate collection loses the channel if and only if all $k+1$
nodes in the collection are in backoff and a conflicting node acquires
the channel by winning the competition between the backoff timers.
The probability that a node goes into backoff at the end of a packet
transmission is $O((\epsilon R)^{-\gamma})$, or approximately the fraction
of time that a node spends in backoff is $O((\epsilon R)^{-\gamma})$.
Therefore, the fraction of time that all $k+1$ nodes of the duplicate
collection are simultaneously in backoff is $O((\epsilon R)^{-k\gamma})$
because the nodes in the duplicate collection act independently from
each other.
Therefore, over an interval of duration $O(R)$, the amount of time
that all $k+1$ nodes are in backoff is $O( R^{1-(k+1)\gamma})$,
which goes to zero as $R \to \infty$ if $\gamma > 1/(k+1)$.
Thus, the nodes in the duplicate collection follow the fluid limits of
an M/M/1 queue until their backlog is below $\epsilon$ on the fluid scale.
Since $\epsilon$ could be made arbitrarily small, we can view the
duplicate collection as a super node that does not release the channel
until its backlog hits zero.
This demonstrates the instability of fluid limits for the initial
queue lengths described above for the duplicate network.

To rigorously prove instability of the original process using the
framework of Meyn~\cite{Meyn95}, we need to show instability of the
fluid limit for any initial state.
Handling arbitrary initial states for general activity functions
and interference graphs is more involved than in the specific
broken-diamond network considered here.
An alternative option would be to extend the methodology and develop
a proof apparatus where it suffices to show instability of the fluid
limit for one particular initial state.
The framework of Dai~\cite{Dai96} offers the advantage that instability
of the fluid limit only needs to be shown for an all-empty initial state.
However the characterization of the fluid limit for an all-empty
initial state appears to involve additional complications.

The above proof arguments suggest that instability can in fact occur
for any $\gamma >0$ as $k$ can be chosen arbitrarily large.
This indicates that the growth conditions in Ghaderi \& Srikant~\cite{GS10}
are sharp in the sense that backoff probabilities of the form
$\frac{1}{O(\log(x))}$ are essentially the most aggressive
de-activation functions that guarantee maximum stability
of queue-based random access in arbitrary graphs.
In terms of backoff probabilities $\frac{1}{1+{\rm e}^{w(\Qorig)}}$ used
in~\cite{GS10}, this means the weight functions $w(x) = \log(1+x)/h(x)$,
where $h(x)$ is an arbitrarily slowly increasing function, are essentially
the most aggressive weight functions that the random-access algorithm
can use while preserving maximum stability in general topologies.


\appendices

\newlength{\figwidth}
\newlength{\figheight}
\newlength{\philwidth}
\newlength{\philheight}

\setlength{\figwidth}{0.96\columnwidth}
\setlength{\figheight}{2.0 in}
\setlength{\philwidth}{2.0 in}
\setlength{\figheight}{2.0 in}

\def\mb#1{\mbox{\boldmath $#1$}}
\def\mbb#1{\mathbb{#1}}
\def\mc#1{{\mathcal{#1}}}
\def\mr#1{{\mathrm{#1}}}

\def\mm#1{\boldsymbol{\mr{#1}}}
\def\mmh#1{\mm{\hat{#1}}}
\def\mmt#1{\mm{\tilde{#1}}}

\def\mv#1{\boldsymbol{{#1}}}
\def\mvh#1{\mv{\hat{#1}}}
\def\mvt#1{\mv{\tilde{#1}}}

\def\lrb#1{\ensuremath{\left({#1}\right)}}%
\def\lrc#1{\ensuremath{\left\{{#1}\right\}}}%
\def\lrs#1{\ensuremath{\left[{#1}\right]}}%
\def\lrv#1{\ensuremath{\lvert{#1}\rvert}}%
\def\blrv#1{\ensuremath{\Bigg\lvert{#1}\Bigg\rvert}}%
\def\lrvv#1{\ensuremath{\lVert{#1}\rVert}}%

\newcommand{\1}{{\rm 1\hspace*{-0.4ex}%
\rule{0.1ex}{1.52ex}\hspace*{0.2ex}}}
\newcommand{\one}[2]{  {\rm 1\hspace*{-0.4ex} \rule{0.1ex}{1.52ex}^{ #1}_{  #2}}}
\newcommand{\philpr}[2]{\ensuremath{\langle {#1},{#2}\rangle}}
\newcommand{\norm}[1]{\ensuremath{\|#1\|}}

\newcommand{\tr}{\mathrm{Tr}}
\newcommand{\VN}{{\bf V}_N}
\newcommand{\ZI}{{\mathbb Z}}
\newcommand{\ZIO}{{\mathbb Z}_0}
\newcommand{\ZZ}{{\mathbb N}_0} 
\newcommand{\ZBZ}{{\mathbb Z}^B_0}
\newcommand{\babs}[1]{\blrv{#1}}
\newcommand{\abs}[1]{\lrv{#1}}%
\newcommand{\half}{\frac{1}{2}}

\def\Exp#1{\mbb{E}[{#1}]}%
\def\Var#1{\mbb{V}[{#1}]}%
\def\Prb#1{\mbb{P}[{#1}]}%
\def\Explr#1{\mbb{E}\lrs{{#1}}}%
\def\Varlr#1{\mbb{V}\lrs{{#1}}}%
\def\Prblr#1{\mbb{P}\kern-2pt\lrs{{#1}}}%
\def\dPrb{d\mbb{P}}%
\def\lrvert#1{\ensuremath{\lvert{#1}\rvert}}%

\newcommand{\RR}{\mathbb R}
\newcommand{\RRp}{{\mathbb R}_+}
\newcommand{\Qrat}{{\mathbb Q}}
\newcommand{\sgn}{\text{sgn}}
\newcommand{\perm}[2]{\left(\begin{array}{c}
                             #1 \\ #2
                             \end{array} \right)}
\newcommand{\MINF}[1]{I \left[ #1 \right]}
\newcommand{\HNTRPY}[1]{H \left[ #1 \right]}
\newcommand{\prob}[1]{{\mathbb P}\{#1\}}
\newcommand{\Px}[2]{{\mathbb P}_{#1}\{#2\}}
\newcommand{\Ex}[2]{{\mathbb E}_{#1}\left[ #2\right]}
\newcommand{\qed}{$\diamond$~\\}
\newcommand{\ceil}[1]{\lceil #1 \rceil}
\newcommand{\floor}[1]{\lfloor #1 \rfloor}
\newcommand{\euclidnorm}[1]{\lvert \lvert #1 \rvert \rvert}
\newcommand{\spnorm}[1]{\lvert \lvert #1 \rvert \rvert_1}
\newcommand{\maxnorm}[1]{\lvert \lvert #1 \rvert \rvert_\infty}
\newcommand{\opnorm}[1]{\lvert \lvert #1 \rvert \rvert_2}
\newcommand{\lc}{\left\{}
\newcommand{\rc}{\right\}}
\newcommand{\lb}{\left(}
\newcommand{\rb}{\right)}
\newcommand{\ls}{\left[}
\newcommand{\rs}{\right]}
\newcommand{\cmplx}{\mathbb C}
\newcommand{\RRe}[1]{\Re\left( #1 \right)}
\newcommand{\IIm}[1]{\Im\left( #1 \right)}
\newcommand{\cF}{{\cal F}}
\newcommand{\cH}{{\cal H}}
\newcommand{\Filt}{{\cal F}_t}
\newcommand{\Borel}[1]{{\cal B}_{#1 }}
\newcommand{\CtopM}{{\cal C}}
\newcommand{\rC}{\rho_C}  
\newcommand{\cC}{{\cal C}}
\newcommand{\bldsym}[1]{\mbox{\boldmath $ #1 $}}
\newcommand{\olQR}{\overline{Q}^{(R)}} 


\newcommand{\olQBR}{\overline{{\bf Q}}^{(R)}}
\newcommand{\olIBR}{\overline{{\bf I}}^{(R)}}



\newcommand{\Icell}{{\cal I}_\ell}



\newcommand{\tK}{\tilde{K}}
\newcommand{\tO}{\tilde{O}}
\newcommand{\Btelep}{B^\ell_{t,\epsilon}}
\newcommand{\BtQelep}{B^\ell_{t,Q,\epsilon}}
\newcommand{\vQT}{\varepsilon_{Q_T}}
\newcommand{\kRx}{k_{R,\xi}}

\newcommand{\QBR}{{\bf Q}^R}
\newcommand{\IBR}{{\bf I}^R}
\newcommand{\IBRB}{{\bf I}^{(R)}}
\newcommand{\QRB}{Q^{(R)}}
\newcommand{\QR}{Q^R}
\newcommand{\UBR}{{\bf U}^{(R)}}
\newcommand{\IbarR}{\overline{I}^{(R)}}
\newcommand{\olQ}{\overline{Q}}
\newcommand{\bR}[1]{b^{(R)}_{ #1 } }
\newcommand{\bRR}[2]{b^{( #1 R)}_{ #2}}
\newcommand{\tR}{\tau^{(R)}}
\newcommand{\sT}[1]{T_{ #1}}
\newcommand{\olp}{\overline{p}}
\newcommand{\GR}{g^{(R)}} 
\newcommand{\UQ}{Q^{(R)}} 
\newcommand{\NN}{{\mathbb N}}
\newcommand{\dE}[2]{D^{ #1}_n( #2 )}


\newcommand{\Zstop}[1]{Z^{ #1}_{m,n}}
\newcommand{\Astop}[1]{A^{ #1}_{m,n}}
\newcommand{\Vstop}[1]{V^{ #1}_{m,n}}
\newcommand{\Bstop}[1]{B^{ #1}_{m,n}}
\newcommand{\Pstop}[1]{P^{ #1}_{m,n}}
\newcommand{\Zsigma}[1]{{\cal Z}^{ #1 }_{m,n}}
\newcommand{\Zsgma}[1]{{\cal Z}^{ #1 }_{m,n,[0,\infty)}}
\newcommand{\Kell}[1]{K^{(\ell)}_{#1}}
\newcommand{\Uell}{A^{(\ell)}}



\newcommand{\tE}{\tilde{E}}



\newcommand{\JS}[2]{J^{(#1)}_{#2}}
\newcommand{\FZ}[2]{{\cal F}_{ #1}^{ #2}}
\newcommand{\FZL}{{\cal F}_{Z_{n,0}^{(L)}}}
\newcommand{\Gctm}{G_{c,t}^{(m)}}
\newcommand{\Gct}{G_{c,t}}
\newcommand{\Cct}[1]{C^{( #1 )}_{c,t,R}}
\newcommand{\Ict}[1]{I^{#1}_{c,t}}

\newcommand{\Wct}[1]{W^{( #1 )}_{c,t}}
\newcommand{\Hlt}[2]{ H^{  #1 }_{ #2}}

\newcommand{\muR}[1]{\mu_R \left\{ #1 \right\}}
\newcommand{\muu}[1]{\mu \left\{ #1 \right\}}
\newcommand{\muuu}[2]{\mu_{ #1}\Big\{ #2 \Big\}}
\newcommand{\FoK}{{\cal F}^{(\Omega_k)}_{Z^0_{m,n}}}
\newcommand{\FwK}{{\cal F}^{(\Omega_k)}_{w_k}}

\newcommand{\Bkoff}[1]{B^{(\ell)}\left(  #1 \right)}
\newcommand{\DQell}[2]{D^{(\ell)}_{ #1}\left( #2 \right)}


\newcommand{\Dvarsgma}[1]{D^{(\ell)}_{ #1 }}
\newcommand{\Glsh}{G^{(\ell)}_{s,h}}
\newcommand{\Gllsh}[1]{G^{(\ell)}_{#1}}
\newcommand{\FltQM}{F^{(\ell)}_{t,Q,M}}



\newcommand{\WcR}{W^R}
\newcommand{\hC}{\hat{C}}



\newcommand{\PhiLin}[2]{\Phi^{ #1 }_{ #2 }}


\newcommand{\QellsQ}{Q^{(\ell)}_{s,Q}}


\newcommand{\Cn}{C^{(n)}}
\newcommand{\Wck}[1]{W^{#1}}
\newcommand{\BBR}[1]{B_R^{ #1 }}
\newcommand{\CctR}[1]{C_{c,t,R}^{ #1 }}
\newcommand{\QQR}[1]{Q_R^{ #1}}
\newcommand{\NNR}[1]{N_R^{#1}}
\newcommand{\muG}[1]{\mu_G \left\{ #1 \right\}}
\newcommand{\muO}[1]{\mu_{\Omega_k} \left\{ #1 \right\}}
\newcommand{\muGR}[1]{\mu_G^{(R)} \left\{ #1 \right\}}


\newcommand{\bM}[1]{b^{(M)}_{#1}}
\newcommand{\bI}[1]{b^{(#1)}}
\newcommand{\bth}{b^{(3)}}
\newcommand{\piinf}[1]{\pi^\infty_{ #1}}
\newcommand{\olQf}{\overline{Q}_f}
\newcommand{\ovpR}{\overline{p}^R}
\newcommand{\ovp}{\overline{p}}
\newcommand{\ovqR}{\overline{q}^R}
\newcommand{\ovq}{\overline{q}}

\newcommand{\odR}{\overline{d}^R}
\newcommand{\veQT}{\varepsilon_{Q_T}}
\newcommand{\sqR}{\sqrt{R}}
\newcommand{\dfrfv}{\delta_{4,5}}
\newcommand{\dRfrfv}{\delta^R_{4,5}}
\newcommand{\Tpot}{\tau^{(1,2)}_p}
\newcommand{\EZplus}[1]{E^Q_{#1}}
\newcommand{\EZStop}[2]{\upsilon^{(#1)}_{#2}} 
\newcommand{\TauStop}[1]{\tau_S^{#1}}
\newcommand{\TauStops}[2]{\tau_{ #2}^{#1}}
\newcommand{\FinEv}[1]{F^{ #1}_E}
\newcommand{\FVV}[1]{F^V_{ #1}}
\newcommand{\YY}[1]{Y^{ #1}_M}
\newcommand{\pae}[1]{p_{#1}(a,h,e,m,c,q,r)}
\newcommand{\pgb}{p_{g,b}}
\newcommand{\ovt}{\overline{t}}
\newcommand{\udt}{\underline{t}}
\newcommand{\unf}{\underline{f}}
\newcommand{\pRe}{p^R_\epsilon}
\newcommand{\etafrac}{5}

\newcommand{\sW}{\tau^{(S)}}


\appendices
\section{Fluid limit proofs: Part A}

\subsection*{A.I. Prelimit model}
\label{sec_prelim}

We start with the time-homogeneous Markov process $({\bf U}(t),{\bf \Qorig}(t)), t \geq 0$
with state space ${\cal S} = S \times \ZZ^N$ where $N=6$ and $S
\subseteq \lc 0,1 \rc^N$ is the set of feasible activity states,
which has already been fully described earlier in Section~\ref{mode}. We recap to state that service times are unit exponential as are backoff periods. In addition the Poisson arrival
processes are determined by the vector of arrival rates $\bldsym{\lambda}$ and the probability of backoff is determined as a function of queue length $1/(1+Q)^\gamma$ with $\gamma \in (1, \infty)$. As
mentioned earlier, the case $\gamma = \infty$ corresponds to the random capture algorithm,
considered in \cite{FPR10}.


The fluid limit will not be obtained directly from the above process but rather via
the {\em jump chain} of a {\em uniformized version} with ``clock ticks''  from a Poisson clock with constant rate,
\be \label{beta-definition}
\beta \doteq \sum_{\ell=1}^N \lambda_\ell + N,
\ee
independent of state, with null (dummy) events introduced as needed.

With minor abuse of notation, denote by $\lb{\bf U}(n), {\bf\Qorig}(n)\rb \in {\cal S}$ to be the state
of the jump chain at $n$th clock tick.
For our subsequent construction, it will be convenient to replace ${\bf U}(n)$ with the cumulative state ${\bf I}(n) = \sum_{k=0}^n {\bf U}(n) \in \ZZ^N$, which is by definition
increasing. It determines and is determined by the sequence ${\bf U}(n)$ and the associated jump chain
is Markov if the state is altered to be $({\bf I}(n), {\bf I}(n-1), {\bf\Qorig}(n))$ with ${\bf I}(-1) = 0$. {\em Note that the process ${\bf I}(n)$ counts the number of steps where the queue process is active. It is not a count of
  the number of service completions by step $n$.} 

From the jump chain, we obtain a continuous stochastic process in $C[0,\infty)$ by linear interpolation and by accelerating time by a factor $\beta$. To be specific, at an arbitrary intermediate
time $t > 0$ between two clock ticks $t_l = (k-1)/\beta \leq t \leq t_u = k/\beta,~ k\in \NN$, the interpolated process takes the values
\begin{eqnarray*}
\overline{{\bf Q}}(t) & \doteq & \beta (t_u - t) {\bf\Qorig}(k-1) + \beta (t - t_l) {\bf\Qorig}(k), \\
\overline{{\bf I}}(t) & \doteq & \beta (t_u - t) {\bf I}(k-1) + \beta (t - t_l) {\bf I}(k).
\end{eqnarray*}

From this construction we can obtain a sequence of such processes, indexed by $R \in \NN$,
with the usual fluid limit scaling
\begin{equation}
\lb \QBR(t), \IBR(t) \rb \doteq \lb \frac{1}{R} \olQBR(Rt), \frac{1}{R}\olIBR(Rt) \rb.
\label{eqn_scaleprocess}
\end{equation}
This is obtained together with a corresponding sequence of initial queue lengths
\begin{equation}
{\bf Q}^R(0) = \frac{1}{R} \olQBR(0) \rightarrow {\bf q}.
\label{eqn_initload}
\end{equation}
Recall  that the underlying jump chain $\lb {\bf U}(n), {\bf\Qorig}(n) \rb_{n \geq 0}$ is
affected only through the initial state. Its transition probabilities are unaffected. The convergence in (\ref{eqn_initload}) is with respect to the Euclidean norm and without loss of generality we may take  $\euclidnorm{{\bf q}} = 1$.

%
%

For every $R$ and time  $t \geq 0$, $\lb {\bf Q}^R(t), {\bf I}^R(t) \rb$ take values in
$E \doteq \RR_+^N \times \RR_+^N$, which is therefore the {\em state space} of the process. $E$ has the usual Euclidean metric and associated topology and we will denote the Borel sets by $\Borel{E}$. Furthermore the underlying jump chain $\lb {\bf U}(n), {\bf \Qorig}(n) \rb_{n \geq 0}$ of the uniformized Markov process satisfies the ``skip-free property''
\cite{Meyn95} which ensures that the jumps between states are bounded in ${\cal L}_1$.  It
follows that the interpolated paths are Lipschitz continuous with Lipschitz constant $3\beta < \infty$.  This property is conferred on the sample paths $\omega$ themselves as stated below
\begin{equation}
\euclidnorm{\QBR(t,\omega) - \QBR(s,\omega) } + \euclidnorm{\IBR(t,\omega) - \IBR(s,\omega)}
\leq 3\beta \lb t - s \rb
\label{eqn_Lipschitz}
\end{equation}
which holds $\forall \omega, 0 \leq s < t, R \in \NN$. The factor 3 appears since
at most two queues can be active at the same time and at each clock tick at most one queue
can be in(de)cremented.

To summarize, the scaled sequence of processes as defined in (\ref{eqn_scaleprocess}) 
take values in the space $C[0,\infty)$ of continuous paths taking values in
$E$, endowed with the supnorm topology, and $\sigma$-algebra $\CtopM$ generated by the open sets.
This is  obtained through the usual metric $\rC$ as defined
in \cite{Whitt69}, page 6.  This space is both separable and complete, see
\cite{Whitt69} Theorem 2.1. 
The probability measure induced on $\CtopM$ by the $R$th  interpolated process
(\ref{eqn_scaleprocess}) is
denoted $\mu_R$ so that $\mu_R(A)$ is the probability of an event $A \in \CtopM$. Finally, it is of
course the case that the jump chain sequence determines and is determined by the corresponding interpolated path. Hence  $\mu_R$ and the jump chain probabilities are equivalent, given the initial
conditions.


\subsection*{A.II. Fluid limit}

\label{sec_fluid}

%
%
If there is an infinite subsequence, $R_{k_1}, R_{k_2}, \ldots$ such that $\mu_{R_{k_n}} \Rightarrow \mu$
where $\Rightarrow$ denotes weak convergence, then $\mu$ is said to be a fluid limit measure.
If such a fluid limit exists then the corresponding process can be defined as follows. Its {\em state
  space} is again $E$ 
with underlying sample space $C[0,\infty)$ and corresponding $\sigma$-algebra $\cC$ 
generated by the open sets under the metric, $\rC$, as mentioned earlier.
This is the same space as for the sequence of prelimit processes. With the fluid limit
measure $\mu$ (including the deterministic initial conditions) we have an underlying  probability space  $\lb C[0,\infty), \cC, \mu \rb$. 
The stochastic process, $\lb {\bf Q}, {\bf I} \rb$  is the mapping $[0,\infty) \times C[0,\infty) \rightarrow E$ with values $\lb {\bf Q}(t,\omega),
{\bf I}(t,\omega) \rb \in E$. The curves $\lb {\bf Q}(.,\omega), {\bf
  I}(.,\omega) \rb$ and $\omega$ itself are the same. 
Whilst these definitions are somewhat redundant, nevertheless in  what follows, it will be convenient to think of a sample path as either a point $\omega$ or as a random function. Finally, on some occasions, we will  use the notation $X \in m\cC$ to indicate that $X: C[0,\infty) \rightarrow \RR$ is
measurable. 

%
%
The proof of the next Theorem is standard and follows from Lipschitz continuity, Theorem 8.3 of [1], and Lemma 3.1 of [40]. The details are omitted for brevity.
\begin{theorem}
\label{thm_tightmuR}
The sequence of measures $\mu_R$ defined on $\lb C[0,\infty),\CtopM \rb$ is tight.
\end{theorem}
~~~\\
Thus, it follows from Prohorov's Theorem (Theorem 6.1 of \cite{Billingsley68}) that the sequence $\mu_R$ is relatively compact and fluid limit measure $\mu$ must exist. We suppose without loss of generality that
$\mu_R \Rightarrow \mu$. The sample paths under $\mu$ have the same Lipschitz constant $3\beta$.
It follows that the sample
paths of $\mu$ are differentiable a.e., almost surely \cite{Royden87}. 

Lipschitz continuity also implies that there are only a countable number of closed intervals $[a,b]$, $0 \leq a < b$, such that
$Q_\ell(a, \omega) = Q_\ell(b,\omega) = 0$, $\ell = 1,\cdots,N$, and
$Q_\ell(x,\omega) > 0, \forall x \in (a,b)$, $\ell=1,\cdots,N$, holding
almost surely. 

We denote by $\lc \cF_t \rc_{t \in [0,\infty)},~{\cal F}_t \subset \CtopM$, the filtration of
sub $\sigma$-algebras generated by the open sets restricted to the interval $[0,t]$. 
The process $\lb {\bf Q}, {\bf I} \rb$ is adapted to $\lc \Filt \rc_{t \in [0,\infty)}$
(In fact it is $\Filt$-progressive as the process is continuous, see 
\cite{EK85}). 

By consideration of the weak law of large numbers and the existence of the fluid limit measure
$\mu$, it holds that
\begin{equation}
{\bf Q}(t) = {\bf Q}(0) + \bldsym{\lambda} t - \frac{1}{\beta} {\bf I}(t),~t \geq 0.
\label{eqn_accountq}
\end{equation}
This equation can be thought of as an accounting identity. If queue $\ell$ is active for
a unit interval then $I_\ell$ increases by $\beta$, which corresponds (almost surely) to departures at unit rate. During the same period the arrival rate is $\lambda_\ell$ of course.

Since $I_\ell(t+h) - I_\ell(t) \leq \beta h$ for any node $\ell$,
and any times $t\geq 0$ and $h >0$, it follows from (\ref{eqn_accountq}) that
\begin{equation}
Q_\ell(t+h) \geq  Q_\ell(t) + \lambda_\ell h - h,~\mu~a.s.
\label{eqn_accountcon}
\end{equation}

We now derive an elementary property of the fluid limit process. Given $t \geq 0, h > 0$, define
\be \label{eq Ynotationhelp}
Y^\ell_{t,h} \doteq \lc \omega: I_\ell(t+h, \omega) - I_\ell(t,\omega) = \beta h  \rc
\ee
to be the event that queue $\ell$ is being served (at maximum rate) during the interval
$[t,t+h]$, i.e., the node is fully active during the given interval. Since many of the events that
we consider later are in terms of activity, we adapt the following notation throughout the paper. In the case of (\ref{eq Ynotationhelp}),
\begin{equation}
Y^\ell_{t,h} = \JS{\ell}{=}(t,h,\beta h)
\label{eqn_Jdef}
\end{equation}
where the superscript ``$\ell$'' denotes the node, ``$t$'' time and ``$h$'' duration. ``$\beta h$'' is the
amount of activity which must be met with equality here, as indicated by the subscript ``=''.
The subscript ``='' may be replaced by $>$, $\geq$, $<$, or $\leq$, depending on the event.


\begin{lemma}[No Conflict Lemma]
\label{lemma_conflict}
Let $\ell_1 \neq \ell_2 \in \lc 1,\cdots,N \rc$ be two neighbors in the interference
graph $G$, and
 $h > 0, t\geq 0$,
then
\begin{equation}
\muu{Y^{\ell_1}_{t,h} \cap  Y^{\ell_2}_{t,h} } = 0.
\label{eqn_AcapA}
\end{equation}\
\end{lemma}
\begin{proof}
This follows by definition, and the existence of the fluid limit. The event
$Y^{\ell_1}_{t,h} \cap  Y^{\ell_2}_{t,h}$ contradicts the inequality that
for all $t \geq 0, h > 0$,
$$
\ls I_{\ell_1}(t+h,\omega) - I_{\ell_1}(t,\omega) \rs + \ls I_{\ell_2}(t+h,\omega) - I_{\ell_2}(t,\omega) \rs \leq  \beta h,
$$
which holds $\mu$ almost surely.
\end{proof}

To obtain more detailed information with respect to the sample paths of $\mu$, we proceed to
the construction of sequences of stopping times.

\subsection*{A.III. Sequences of stopping times}
\label{sec_stopseq}

The following definition is in connection with the amount of time a sample path for
$Q_\ell$ is positive, immediately prior to a time $z > 0$.
\begin{definition}
Given a time $z > 0$, and $v,0 < v \leq z$, and an $\ell =1,\cdots,N$, define
$$
\Kell{z,v} \doteq \lc \omega : Q_\ell(z-s,\omega) > 0,~\forall s \in (0,v) \rc.
$$
In words, $\Kell{z,v}$ is the set of sample paths for 
$Q_\ell$ which
are strictly positive in the interval $(z-v,z)$; if $z=0$, $\Kell{z,v}$ is taken to
be $\emptyset$.
\end{definition}
Observe that it could be the case that either $Q_\ell(z,\omega) = 0$ or $Q_\ell(z-v,\omega)=0$
(or both) and still $\omega \in \Kell{z,v}$. Finally note that it is possible for a given $\omega$
that no such $v$ can be found, which requires that $Q_\ell(z,\omega) = 0$ on account of continuity.
It can be shown that
$$
\Kell{z,v} = \cap_{n : 2/n < v} \ls \cup_{m=1}^\infty \lc \omega: Q_\ell(z-q,\omega) \geq 1/m,~
q \in [1/n,v-1/n] \cap \Qrat \rc \rs \in \cF_z,
$$
for $z > 0$, where $\Qrat$ is the set of rational numbers.


Given a time $z \geq 0$ and a path $\omega$, we define the mapping $\Uell(z,\omega): C[0,\infty) \rightarrow [0,z]$
to be $\Uell(z,\omega) \doteq \sup \ls \lc v : \omega \in \Kell{z,v} \rc \cup \lc 0 \rc \rs$, which is the time for which $Q_\ell$ was positive immediately prior to $z$. By definition, if $z \geq u > 0$ then
$$
\lc \omega : \Uell(z,\omega) \geq u \rc = \Kell{z,u},
$$
from which it follows that $\Uell(z,\omega) \in m\cF_z$. So far $z$
has been fixed.  However $\Uell : [0,\infty) \times C[0,\infty)
\rightarrow \RR_+$ is a stochastic process carried by the underlying
probability space $\lb C[0,\infty), \CtopM, \mu \rb$ and
$\cF_t$-adapted as we have just seen. This process is piecewise linear
and left-continuous (It falls to 0 immediately after $Q_\ell$ returns
to 0 from being positive). It follows that $\Uell$ is
$\cF_t$-progressive, see for example \cite{EK85}. 
We are now in a position to make the following definition.
\begin{definition}\label{defn_Tmap}
Given an  ${\cal F}_t$ stopping time $\sigma$, a queue $\ell \in 1,\cdots,N$ and
$m \in \ZIO \doteq \ZI - \lc 0 \rc$, define $\sT{\ell,m}(\omega,\sigma) : C[0,\infty) \times [0,\infty] \rightarrow [0,\infty]$ as follows
$$
\sT{\ell,m}(\omega,\sigma) \doteq \inf \lc z \geq \sigma(\omega) : Q_\ell(z) = 0,
\Uell(z,\omega) \in (e_m,f_m] \rc \leq \infty,
$$
where
\begin{eqnarray*}
f_m & = & \frac{1}{m},~e_m= \frac{1}{m+1}; \mbox{ for }m \in \ZIO,~m > 0, \\
f_m & = & \abs{m-1},~e_m = \abs{m};\mbox{ for } m \in \ZIO,~m < 0,
\end{eqnarray*}
where again empty sets have an infinite infimum.
\end{definition}
In words, $\sT{\ell,m}$ is the earliest right-hand end of an open interval, with value $z$, such that $Q_\ell$ is positive for a period $\Uell(z,\omega) \in (e_m , f_m]$, immediately prior to $\sT{\ell,m}$. If $z-f_m$ is
the first time prior to $z$ that $Q_\ell=0$, then $z$ is in the set on the RHS. However, if this occurs at
$z - e_m$, this is not the case. 

It is plausible that $\sT{\ell,m}$ is also an ${\cal F}_t$ stopping time, and we will subsequently prove
this with particular choices for $\sigma$.  We now state a construction lemma using a sequence of stopping times. These are returns to 0 following a fixed positive interval, in which we wait for a particular event $A_k$ to occur.
\begin{lemma}[Stop and Look Back]
\label{lem_stoplookback}
Let $\sigma \geq 0$ be an ${\cal F}_t$ stopping time and $a > 0$ a constant. Proposition
1.5 in Ethier \& Kurtz \cite{EK85} ensures that the following
inductively defined sequence is a sequence of $\cF_t$ stopping times:
$s_0,s_1,s_2,\ldots$,
\begin{eqnarray}
s_0 & \doteq  & \sigma \\
\label{eqn_zerosstop}
s_k & \doteq & \tau_c(\lc 0 \rc, s_{k-1} + a),~k=1,2,\cdots \nonumber
\end{eqnarray}
Here, given an $\cF_t$ stopping time $\sigma_1 > 0$, $\tau_c(\lc 0 \rc,\sigma_1) = \inf \lc t \geq \sigma_1, Q(t,\omega)= 0 \rc$.
Now let $A_k \in {\cal F}_{s_k}, k=1,2,\ldots$ be a sequence of events in the pre-$T$ $\sigma$-algebras of the above stopping time sequence. Finally, define $\tau \doteq s_k$ if $A_k$ occurs for the first
time at step $k$ and $\tau = \infty$ otherwise. Then $\tau$ is an $\cF_t$ stopping time.
\end{lemma}

Note that we do not check to see if $A_k$ has occurred if $s_k = \infty$ at any stage, as $\tau$ is assigned this value regardless.

%

We now proceed to show the following.
\begin{lemma}
\label{lem_defstoptime}
Let $\sigma_0 \geq 0$ be an $\cF_t$ stopping time such that $Q_\ell(\sigma_0(\omega),\omega) = 0$
or $\sigma_0 = \infty$ and suppose $\ell,m$ are given. Let $a = e_m$ and $\sigma \doteq a + \sigma_0$ which is therefore an $\cF_t$ stopping time, and $\sT{\ell,m}$ be the mapping given in Definition \ref{defn_Tmap}. Then $\sT{\ell,m}(\omega,\sigma)$ is an $\cF_t$ stopping time.
\end{lemma}

\begin{proof}
Given $\sigma$ we will obtain a sequence of stopping times as
in the first part of Lemma \ref{lem_stoplookback}. However as we have already discussed,
$\Uell$ is $\cF_t$-progressive, from which it follows by
Proposition 1.4 of \cite{EK85} that
$$
A_k := \Uell(s_k(\omega), \omega) \in m\FZ{s_k}{}, \forall~k = 1,2,3,\ldots
$$
so that $\Uell(s_k(\omega),\omega) \in (e_m,f_m] \in \FZ{s_k}{}$. Hence $\tau$ as defined in Lemma
\ref{lem_stoplookback} is an $\cF_t$ stopping time.

It remains to show that $\tau$ coincides with $T_{\ell,m}$ as
defined. First suppose $\tau < \infty$, and immediately, $\tau \geq
\sigma, Q_\ell(\tau,\omega) = 0, \Uell(\tau,\omega) \in (e_m,f_m]$ by
definition. The fact that there is no earlier time satisfying these
conditions follows since each $s_k$ is a zero of $Q_\ell$ and the
construction rules out that the event could have taken place any
earlier. The case $\tau = \infty$ coincides with there being no zero
satisfying the required conditions.
\end{proof}

We now make the following recursive definitions.
\begin{definition} \label{defn_Zseq}
Given $m \in \ZIO $ and queue $\ell \in \lc 1,\cdots,N \rc$, let $\tau_0$ be the first entry of $Q_\ell(t, \omega)$ into 0 ($\tau_0$ is an $\cF_t$ stopping time). Then $Z^\ell_{m,0}$ is defined as
\begin{eqnarray*}
Z^\ell_{m,0} & \doteq & \sT{\ell,m}(\omega,0); \mbox{ if $Q_\ell(0, \omega) = 0$}\\
Z^\ell_{m,0} & \doteq & \tau_0; \mbox{ if $Q_\ell(0,\omega) > 0$, $\tau_0 \in (e_m,f_m]$}\\
Z^\ell_{m,0} & \doteq & \sT{\ell,m}(\omega,\tau_0); \mbox{ if $Q_\ell(0,\omega) > 0$, $\tau_0 \notin (e_m,f_m]$},
\end{eqnarray*}
and subsequent stopping times are defined as
$$
Z^\ell_{m,n}  \doteq  \sT{\ell,m}(\omega,Z^\ell_{m,n-1}),~n=1,2,3,\ldots.
$$
The value of the stopping time is taken to be $\infty$ if the events
do not occur. With obvious notation, we also define
$$
A^\ell_{m,n}(\omega) \doteq \Uell(\Zstop{\ell},\omega)
$$
to be the actual amount of time that the queue $\ell$ is positive prior to $Z^\ell_{m,n}$, and $A^\ell_{m,n} = \infty$ in case $\Zstop{\ell} = \infty$. Finally define the time
at which 
$Q_\ell$ last enters $(0,\infty)$ prior to $\Zstop{\ell}$ to be
$$
V^\ell_{m,n} \doteq Z^\ell_{m,n} - A^\ell_{m,n},
$$
when $Z^\ell_{m,n} < \infty$ and $V^\ell_{m,n} = \infty$ otherwise. Note that
$V^\ell_{m,n} \in m\cF_{\Zstop{\ell}}$ and is thus a
non-negative random variable but {\em not} a stopping time.
\end{definition}

The following corollary follows immediately from Lemma \ref{lem_defstoptime} and Definition \ref{defn_Zseq}.
\begin{corollary}
\label{corr_Zok}
$Z^\ell_{m,n}, n \in \ZZ$ is a strictly increasing sequence of $\cF_t$ stopping times,
$\forall \ell \in N, m \in \ZIO$.
\end{corollary}

This completes our goal of constructing sequences of stopping times for the queue processes. For any
queue $\ell$, by construction and by Lipschitz continuity, it follows
that the set of stopping times, $Z^\ell_{m,n}$ determine all intervals  where $Q_\ell$ is positive for any sample path almost surely.

For each $m \in \ZIO$, and $\ell \in \lc 1,\cdots,N \rc$, define
$$
B^\ell_m \doteq \sup_n \lc \Zstop{\ell} : \Zstop{\ell} < \infty \rc
$$
to be the supremum of the finite stopping times for positive intervals with duration in
$(e_m,f_m]$. If there is an $m$ such that $B^\ell_m = \infty$, then  $Q_\ell$ returns to 0 infinitely often.
Otherwise there is a $B > 0, B > B^\ell_m, \forall m\in \ZIO$. In this case, either 
$Q_\ell$ remains at 0 as $t \rightarrow \infty$, or queue $\ell$ never returns to 0.

\subsection*{A.IV. Piecewise linearity and no backoff until empty}
\label{sec_pcwselinnobackoff}

So far the backoff exponent $\gamma > 1$ has not been taken into consideration, but from now on it will be. The following lemma bounds the probability, for the jump chain, that node $\ell$ has a backoff before its queue gets ``small'' provided that it was active earlier.
Given some number $\Qorig_T \in\NN$ define,
$$
K_{\Qorig_T} \doteq \lc \exists n: \Qorig_\ell(k) \geq \Qorig_T, 0 \leq k \leq n, U_\ell(n) = 0, U_\ell(n-1) = 1 \rc
$$
to be the event that queue $\ell$ has remained above $Q_T$ and has had a backoff at step $n$. We then have the following lemma,
\begin{lemma}[No Early Backoff]
\label{lemma_noearlybackoff}
Given any $\Qorig_0 > \Qorig_T$,
\begin{equation}
\prob{ K_{\Qorig_T} | \Qorig_\ell(0) \geq \Qorig_0} \leq \frac{1}{1 - \lambda_\ell} \times \sum_{r=\Qorig_T}^\infty \frac{1}{(1+r)^\gamma} =\varepsilon_{\Qorig_T}.
\label{eqn_QTveps}
\end{equation}
\end{lemma}

Since the sum is convergent, $\varepsilon_{\Qorig_T} \downarrow 0$ as $\Qorig_T \uparrow  \infty$.

\begin{proof}
It is convenient to consider the packets being served in generations. That is given a target
packet, suppose that we serve the packets which arrive during its service with preemptive priority up to and including the target
packet. This makes no difference to queue behaviour as the service
times are exponential and we are only interested in the {\em first
  occasion} when node $\ell$ goes into backoff.

Suppose there are $\Qorig_T > 0$ packets in queue at the time the service of a given target packet
starts. Consider the busy period of this packet, i.e. the time to serve the target packet and the subsequent high-priority packets (without backoff). It is easy to see that the mean number of packet arrivals during this busy period is $\frac{1}{ 1 - \lambda_\ell }$, including the target packet itself.

The service of each packet ends with a random decision to backoff with probability less than $\frac{1}{\lb 1 + \Qorig_T \rb^\gamma}$, since the queue length is never shorter than $\Qorig_T$ until the target packet has departed. Thus, by the
union bound, the probability of a backoff occurring before or immediately after the target packet departs, is less than $ \frac{1}{\lb 1 + \Qorig_T \rb^\gamma} \times \frac{1}{1 - \lambda_\ell }$. The probability of a backoff, starting with $\Qorig_0>\Qorig_T$ packets, and before the queue drops below the level $\Qorig_T$, is therefore smaller than $\frac{1}{ 1 - \lambda_\ell } \sum_{r=\Qorig_T}^{r=\Qorig_0} \frac{1}{(1+r)^\gamma}$ which implies the statement of Lemma.

\end{proof}

The following lemma will also be useful. First given times $t_2 > t_1 \geq 0$ on the fluid
scale, let $B_\ell([t_1,t_2])$
be the event that node $\ell$ {\em starts a backoff}
in the interval $[t_1,t_2]$. This event occurs in the prelimit process $\lb Q^R_\ell, I^R_\ell \rb$
if for some jump chain index
$n, U_\ell(n) = 1, U_\ell(n+1) = 0$ with $\floor{R t_1 \beta} \leq n \leq \ceil{Rt_2 \beta}$.
Let $\DQ{\varsigma}{[t_1,t_2]}$
be the event that $Q^R_\ell(u) > \varsigma$ (or equivalently $\bar{Q}_\ell(Ru) \geq R\varsigma$) for all $u$ in the interval $[t_1,t_2]$.

\begin{lemma}
\label{lemma_NoBackOffQLong}
Given the above definitions,
\begin{equation}
\lim_{R \rightarrow \infty} \muR{B_\ell{[t_1,t_2]} \cap \DQ{\varsigma}{[t_1,t_2]}} = 0
\end{equation}
\end{lemma}
\begin{proof}
First we may suppose node $\ell$ becomes active at some stage or there is nothing to prove.
The lemma then follows from the union bound. Since there are at most
$R_{t_1,t_2} \doteq \ceil{R(t_2-t_1)\beta}+2$ departures in the entire interval, the union
bound implies that the probability of a backoff is smaller than,
$$
\muR{B_\ell{[t_1,t_2]} \cap \DQ{\varsigma}{[t_1,t_2]}}  \leq \frac{R_{t_1,t_2}}{\lb 1 + R \varsigma \rb^\gamma} \rightarrow 0.
$$
This completes the proof.
\end{proof}

\begin{definition}
\label{defn_poi}
Given a queue $\ell$, a time $t \in [0,\infty)$ on the fluid scale, and a queue
length $Q_\ell(t) = Q > 0$, we say that $t$ is a point of increase for the activity
process of queue $\ell$ if the event
$$
P_{t,Q}^{(\ell)} \doteq \cap_{M=1}^\infty \left\{\JS{\ell}{>}(t,\frac{1}{M},0)\right\} \cap Q^{(\ell)}_{t,Q}
$$
occurs, with $Q^{(\ell)}_{t,Q} \doteq \lc \omega: Q_\ell(t,\omega) > Q \rc$. In words, queue $\ell$ is active in any arbitrarily small interval $(t,t+1/M)$ and
$Q_\ell$ is greater than $Q$ at time $t$. 


Furthermore, given a time $s \in [0,\infty)$ and $h > 0$, we say that queue $\ell$ is under active, with duration $h > 0$ if the following event occurs
\be \label{eq Glsh}
\Glsh \doteq \JS{\ell}{<}(s,h,\beta h).
\ee
\end{definition}
Points of increase rule out that there is a sequence $t_n \downarrow
t$ such that $I_\ell(t_n,\omega) = I_\ell(t,\omega)$, as there is
activity no matter how small the interval. Under activity means
that there was some idling during the interval. Given our choice of
$\gamma$, it will be shown that a point of increase cannot be followed by a period of under activity
until queue $Q_\ell$ has drained. This is because the probability of even a single
backoff once service has begun, is effectively 0 until the queue has
drained on the fluid scale.
\begin{lemma}
\label{lmma_NoBackOff}
Suppose $s \in \lb t,t+Q/(1-\lambda_\ell) \rb$. Then $\forall h, 0 < h < t+Q/(1-\lambda_\ell) - s$,
and for all sufficiently large  $M$,
\begin{equation}
\muu{\JS{\ell}{>}(t,1/M,0)  \cap Q^{(\ell)}_{t,Q} \cap \Glsh} = 0.
\end{equation}
\end{lemma}
\begin{proof}
Consider the sequence of prelimit processes. We will choose $M$ large enough and $\varsigma$ small enough so that $[s,s+h] \subset (t+1/M,t+(Q-\varsigma)/(1-\lambda_\ell)]$
(for some small constant $\varsigma > 0$). Then for $R,M$ sufficiently large, and then by definition,
occurrence of $\Glsh$, for the $R$-th prelimit process, implies occurrence of  $B_\ell([t+1/M,s+h])$.
Hence we obtain,
\begin{eqnarray}
\muR{\JS{\ell}{>}(t,1/M,0)  \cap Q^{(\ell)}_{t,Q} \cap \Glsh} & \leq & \muR{ \JS{\ell}{>}(t,1/M,0)  \cap Q^{(\ell)}_{t,Q} \cap B_\ell([t+1/M,s+h]) }
\label{eqn_BDinc} \\
       & \leq & \mu_{B,R}+\mu_{F,D,R}, \nonumber
\end{eqnarray}
where
\begin{equation}
\mu_{B,R} \doteq \muR{\JS{\ell}{>}(t,1/M,0) \cap \Bkoff{t+1/M,s+h} \cap \DQ{\varsigma}{[t,s+h]}},
\label{eqn_muBR}
\end{equation}
and
\begin{equation}
 \mu_{F,D,R} \doteq \muR{ Q^{(\ell)}_{t,Q} \cap \lb  \DQ{\varsigma}{[t,s+h]} \rb^c}.
\label{eqn_FcapD}
\end{equation}
Thus, in order to prove the lemma, it is sufficient to show that both $ \mu_{B,R}\rightarrow 0$ and $\mu_{F,D,R} \rightarrow 0,$
as $R \rightarrow \infty$, because then we may conclude that
$$
\muu{\JS{\ell}{>}(t,1/M,0)  \cap Q^{(\ell)}_{t,Q}  \cap G^{(\ell)}_{s,h} }   \leq \liminf
\muR{ \JS{\ell}{>}(t,1/M,0)  \cap Q^{(\ell)}_{t,Q}  \cap G^{(\ell)}_{s,h} } = 0
$$
on applying Theorem 2.1 in \cite{Billingsley68} and since the sets
$\JS{\ell}{>}(t,1/M,0),Q^{(\ell)}_{t,Q},G^{(\ell)}_{s,h}$ are all open.

The fact that $\mu_{F,D,R} \rightarrow 0$ follows from
(\ref{eqn_accountq}) and then by definition of $Q^{(\ell)}_{t,Q}$ and
additionally by the choice of $s,h,\varsigma$. As far as $\mu_{B,R}$ 
is concerned, the event $\JS{\ell}{>}(t,1/M,0)$ implies that service has started during the interval $[t,t+1/M]$. On the other hand, the event $B_\ell([t+1/M,s+h])$ 
implies that at some time in $[t,s+h]$ node $\ell$ starts to backoff. Setting $t_1 = t$ and $t_2 = s + h$,
we may invoke Lemma \ref{lemma_NoBackOffQLong} as by definition the event
$D_{\ell,\varsigma}([t,s+h])$ 
implies $Q_\ell$ did not go below $\varsigma$ in the interval $[t_1,t_2]$. It follows that
$\mu_{B,R} \rightarrow 0$ as required.  
\end{proof}

The implication of Lemma \ref{lmma_NoBackOff} is that any positive period of transmission, no matter how
short,  must be followed by full activity until the queue has drained on the fluid scale. This
implies that there is no period of under activity, until the queue has drained, with probability 1.

\subsection*{A.V. Piecewise linear paths with probability 1}
\label{sec_linpathprob1}

The aim of this section is to show that the queue sample paths follow a certain bilinear path during the interval prior to the queue becoming zero again. The bilinear path depends on the duration of the interval and on the arrival rate for the given queue.

To make the above statements precise, given $\ell \in \lc 1,\cdots, N \rc$ define
 the bilinear path $\PhiLin{\ell}{t_0,t_1}$ for the interval $[t_0,t_1]$ to be,
\begin{eqnarray}\label{eqn_bilinear1}
 \PhiLin{\ell}{t_0,t_1}(s)  = \left\{\begin{array}{ll}
 \lambda_\ell \lb s - t_0 \rb;&  t_0 \leq s \leq s_0, \\
\lambda_\ell \lb s - t_0 \rb -(1-\lambda_\ell)(s-s_0);&  s_0 \leq s \leq t_1,
\end{array}
\right.
\end{eqnarray}
where $s_0 \doteq  t_1-\lambda_\ell(t_1-t_0)$.
In words, $Q_\ell$ builds up linearly in the interval $[t_0,s_0]$ at
rate $\lambda_\ell$ and
drains at rate $1-\lambda_\ell$ in the interval $[s_0,t_1]$. 

Given $\eta > 0$, and $\ell \in \lc 1,\cdots,N  \rc$, define
$\1_{t_0,t_1}^{(\eta,\ell)}$ to be the indicator for the event
$$
\lc \omega: \sup_{s\in [t_0,t_1]}  \abs{ Q_\ell(s,\omega) - \PhiLin{\ell}{t_0,t_1}(s) } < \eta \rc \in {\cal F}_{t_1}
$$
In words, $\1_{t_0,t_1}^{(\eta,\ell)}(\omega) = 1$ iff the absolute difference between $\PhiLin{\ell}{t_0,t_1}$
and the sample path for $Q_\ell$ is smaller than $\eta$ in supnorm over the interval $[t_0, t_1]$.

We now examine the conditional probability that $\1_{\Vstop{\ell},\Zstop{\ell}}^{(\eta,\ell)}(\omega) = 1$, given $\Zstop{\ell} < \infty$ and $\Astop{\ell}$ (the case $\Zstop{\ell} =\infty$ is irrelevant). Define, $\Zsigma{\ell} \doteq \sigma\lb \Zstop{\ell}, \Astop{\ell}  \rb \subset \CtopM$
~and also $\Zsigma{\ell,\infty}=  \Zsigma{\ell} \cap \lc \omega: \Zstop{\ell}(\omega) < \infty \rc$. 

It will be enough to show that the sample paths lie in an arbitrarily small tube around
$\PhiLin{\ell}{\Vstop{\ell},\Zstop{\ell}}$ conditional on  $\Astop{\ell},\Zstop{\ell}$ lying in
some small rectangle $Z_{(s,t)}^{(a,b)} \doteq \lc \omega:\Astop{\ell}(\omega) \in (a,b], \Zstop{\ell}(\omega) \in (s,t] \rc \in \Zsigma{\ell,\infty}$.
\begin{theorem}
\label{thm_linearpaths}
Given $n \geq 1, m \in \ZIO$, then $\forall \eta > 0$,
\begin{equation}
\muu{ \1_{\Vstop{\ell},\Zstop{\ell}}^{(\eta,\ell)}  =1 | \Zsigma{\ell, \infty}} =1~a.s.
\label{eqn_linearsamp}
\end{equation}
In words, given the stopping time $\Zstop{\ell}$ and the time prior to this that $Q_\ell$
was positive, $\Astop{\ell}$, the probability 
that  $\PhiLin{\ell}{\Vstop{\ell},\Zstop{\ell}}$ is followed, starting at
$\Vstop{\ell}$ and ending at $\Zstop{\ell}$, is 1 under the fluid limit
measure $\mu$.
\end{theorem}
\begin{proof}
For any given $\varepsilon > 0$, the sets $Z_{(s,t)}^{(a,b)}$ 
$0 <  s < t$, $0 < a < b$, $0 < t-s, b-a < \varepsilon$,
are a $\pi$-system \cite{Billingsley95} (i.e. closed under finite
intersections) and which generate $\Zsigma{\ell,\infty}$. Hence  we only need to show that
$$
\muu{\1_{\Vstop{\ell},\Zstop{\ell}}^{(\eta,\ell)}(\omega)=1; Z_{(s,t)}^{(a,b)}} = \muu{Z_{(s,t)}^{(a,b)}},
$$
for suitably chosen $\varepsilon$ given $\eta > 0$.

Let $\Bstop{\ell}(\omega) \leq \Astop{\ell}$ be the additional time, following strict
entry of $Q_\ell$ into $(0,\infty)$ at $\Vstop{\ell}$, until the first point of increase of $I_\ell$ is reached. 
$\Bstop{\ell} \in m\FZ{\Zstop{\ell}}{}$ as can be seen on consideration of its definition,
\be \label{eq Bstop}
\Bstop{\ell}(\omega) \doteq \inf \lc u \in (0,\Astop{\ell}(\omega)) \cap \Qrat~: I_\ell(\Vstop{\ell} + u,\omega) - I_\ell(\Vstop{\ell},\omega) > 0 \rc, 
\ee
when $\Zstop{\ell} < \infty$.

By definition of $\Bstop{\ell}$, Lemma \ref{lmma_NoBackOff} and then
(\ref{eqn_accountq}), we may deduce that for $\omega \in Z_{(s,t)}^{(a,b)}$ 
\begin{eqnarray}
Q_\ell(\Vstop{\ell}(\omega) + u,\omega) & = & \lambda_\ell u, ~u\in [0,\Bstop{\ell}(\omega)], \label{eqn_bilinear2}  \\
Q_\ell(\Vstop{\ell}(\omega) + u,\omega) & = & \lambda_\ell \Bstop{\ell}(\omega) -(1-\lambda_\ell)(u-\Bstop{\ell}) , ~
u\in [\Bstop{\ell}(\omega), \frac{\Bstop{\ell}(\omega)}{1 - \lambda_\ell }],
\nonumber
\end{eqnarray}
$\mu$ almost surely. 
Moreover $\Bstop{\ell}(\omega)$ must satisfy
$$
t - s + b  \geq \frac{\Bstop{\ell}(\omega)}{1 - \lambda_\ell } \geq  s-t+a,~\mu~a.s.
$$
in order to reach 0 in $[s,t]$.

Therefore, given any $\eta > 0$, we may choose $\varepsilon_\eta > 0$ such that for all
$v \in [s-b,t-a], z \in [s,t]$ with $b -a, t -s < \varepsilon_\eta$
$$
\sup_{u \in [v,z]} \abs{ Q_\ell(u,\omega) - \PhiLin{\ell}{v,z}(u)} < \eta,
$$
$\mu$ almost surely, using Lipschitz continuity. Since $\omega \in Z^{(a,b)}_{(s,t)}$ implies  $\Vstop{\ell} \in [s-b,t-a], \Zstop{\ell}\in [s,t]$,
we obtain that
$$
\muu{\1_{\Vstop{\ell},\Zstop{\ell}}^{(\eta,\ell)}(\omega)=1; Z_{(s,t)}^{(a,b)}} = \muu{Z_{(s,t)}^{(a,b)}},
$$
for all such $a,b,s,t$ as required.
\end{proof}
~~\\
A similar result can be obtained when $n=0$, where the possibility occurs that $Q_\ell(V^\ell_{m,0}) > 0$.

\subsection*{A.VI. Brief discussion of results}
%
%
Theorem \ref{thm_linearpaths} applies to general networks and relies
only on the assumption that $\gamma > 1$. 
The theorem implies that the
sample paths are more or less determined given the sequences of
stopping times $\Zstop{\ell}$. Only in the case where the (finite)
stopping times have a common upper bound is the process not
completely defined, as otherwise the queue returns to 0 infinitely
often, determining the path completely. If there is such a bound, either the queue remains at 0, or increases linearly, as there can be no subsequent point of increase of $I_\ell$.

Indeed, since there are only countably many stopping times, and since
for each finite $\Zstop{\ell} < \infty$ the queue sample paths follow
$\PhiLin{\ell}{.,.}$ for some finite interval with probability 1, we
may confine sample path realizations to countable successions of such
intervals. These either determine the entire sample path; or
the queue remains at 0 following the final return; or as
the final alternative, the queue remains zero for some interval and then
increases linearly at rate $\lambda_\ell$ thereafter.  We define the
set of such sample paths by $P \subset C[0,\infty)$. The probability of any event $F \in
{\cal C}$ can as well be taken as
$$
\muu{ F} = \muu{F \cap P},
$$
and, therefore, we suppose that the probability space is defined on ($P,
{\cal C}_P$) with topology relativized in the usual way to $P$ which
is a subset of ${\cal C}[0,\infty)$. This establishes that the queue-length trajectory
of each of the individual nodes exhibits \textit{sawtooth} behavior in the fluid limit. This
concludes Part~A.

In Part~B, we will
show that we can in fact confine ourselves to a smaller set of paths
which reflect the constraints resulting from the underlying
interference graph.
\section{Fluid limit proofs: Part B}

\subsection*{B.I. No idling property and zero delay capture}
\label{sec_idlecapture}

From Lemma \ref{lemma_conflict}, it follows that if queue 2 is draining, then queues 3, 4, 5, 6
are increasing linearly. However, we also expect that queue 1 is either draining or remaining at 0 until time $t$, and this is indeed the case as we now show.

More generally, given a node $\ell$, let ${\cal I}_\ell$ be the set of its interfering nodes,
i.e. the set of its neighbors in the interference graph $G$.
The following lemma shows that if $Q_\ell(s) > Q$, and all its interferers are idle in some interval $[s,t]$ then
node $\ell$ is fully active until its queue drains.
\begin{lemma}[No Idling Property]
\label{lemma_NoDelay}
Given a node $\ell$ with interference set ${\cal I}_\ell$
and an interval $[s,t]$, define
$$
D^{(\ell)}_{s,t} \doteq \cap_{j \in {\Icell}} \JS{j}{=}(s,t-s,0),
$$
that is, no activity for any node in $\Icell$ during $[s,t]$.
Further, given $Q >0$, define
$h^{(\ell)}_{s,t,Q} \doteq Q/(1-\lambda_\ell)  \wedge (t - s)$ so that the queue at most empties
over this period, and let
$$
S^{(\ell)}_{s,t} \doteq  \JS{\ell}{<}(s,h^{(\ell)}_{s,t,Q},\beta h^{(\ell)}_{s,t,Q}),
$$
which implies that node $\ell$ is under active.
Then
$$
\muu{D^{(\ell)}_{s,t} \cap S^{(\ell)}_{s,t} \cap Q^{(\ell)}_{s,Q} } = 0,
$$
where $\QellsQ$ is the event $\lc \omega: Q_\ell(s,\omega) > Q \rc$, as defined earlier.
\end{lemma}
\begin{proof}

Given $n \in \NN$ such that $n > 1/(t-s)$,
fix an arbitrary $\zeta,$ $0 < \zeta < \frac{1}{2N}$ (Recall that $N$ is the number of nodes in the network). Clearly,
$$
D^{(\ell)}_{s,t} \subset \tilde{D}^{(\ell)}_{\zeta,n} := \cap_{j \in {\cal I}_\ell} \JS{j}{<}(s,1/n,\zeta/n).
$$
Hence, for arbitrary $\epsilon_n > 0$ depending on $n$, to be fixed later,
$$
D^{(\ell)}_{s,t} \subseteq \JS{\ell}{>}(s,1/n,0) \cup \lb \JS{\ell}{<}(s,1/n,\epsilon_n) \cap \tilde{D}^{(\ell)}_{\zeta,n} \rb.
$$
Next, observe that for all $n_S \in \NN$ sufficiently large,
$$
S^{(\ell)}_{s,t} = \cup_{n > n_S} G_n,
$$
with $G_n \doteq  \Gllsh{s+2/n,h^{(\ell)}_{s,t,Q}-2/n}$ and $\Glsh$ as defined in (\ref{eq Glsh}).
The union bound thus implies that
\begin{eqnarray}
\muu{D^{(\ell)}_{s,t} \cap S^{(\ell)}_{s,t} \cap Q^{(\ell)}_{s,Q}}
& \leq & \sum_{n > n_S} \muu{ G_n \cap Q^{(\ell)}_{s,Q} \cap \JS{\ell}{>}(s,1/n,0)} \label{eqn_muDJbnd} \\
& + & \sum_{n > n_S} \muu{\tilde{D}^{(\ell)}_{\zeta,n} \cap \JS{\ell}{<}(s,1/n,\epsilon_n) \cap Q^{(\ell)}_{s,Q}}. \nonumber
\end{eqnarray}
Provided $n_S$ is sufficiently large, each term in the first sum must be 0, else Lemma \ref{lmma_NoBackOff} is contradicted. To complete the proof, it is
therefore sufficient to show that each of the terms in the second sum is 0 as well by suitable
choice of $\epsilon_n$.

Given $n$, it is sufficient to find $\epsilon_n > 0$ so that
$$
\lim_{R \rightarrow \infty} \muR{ \tilde{D}^{(\ell)}_{\zeta,n} \cap \JS{\ell}{<}(s,1/n,\epsilon_n) \cap Q^{(\ell)}_{s,Q}} = 0,
$$
because $\tilde{D}^{(\ell)}_{\zeta,n}$, $\JS{\ell}{<}(s,1/n,\epsilon)$, and $Q^{(\ell)}_{s,Q}$ are all open, so that Theorem 2.1 \cite{Billingsley68} implies that
$$
\muu{\tilde{D}^{(\ell)}_{\zeta,n} \cap \JS{\ell}{<}(s,1/n,\epsilon_n) \cap Q^{(\ell)}_{s,Q}} = 0.
$$
The event $\tilde{D}^{(\ell)}_{\zeta,n}$ implies that there must have been at least
\begin{equation}
\frac{R\beta}{n} \lb 1 - N\zeta  \rb > \frac{R\beta}{2n}
\label{eqn_lotsbackoff}
\end{equation}
steps in the jump chain (if we allow for no overlap between active periods and since
$\abs{{\cal I}_\ell} < N$)  at which all queues in ${\cal I}_\ell$ are in backoff for the interval
$[s,s+1/n]$. Also, 
\begin{equation}
Q^{R}_\ell > Q - \frac{\beta}{n} > \varsigma > 0,
\label{eqn_Qbeta}
\end{equation}
throughout $[s,s+1/n]$ since there can be at most $R\beta/n$ departures. 

But if (\ref{eqn_lotsbackoff}) occurs, we may suppose that node $\ell$ becomes active within
$R\beta/(4n)$ such steps, as the probability converges to 1 as $R \rightarrow \infty$ that it does so.
But if we take $0 < \epsilon_n < \beta/(4n)$ the implication is that there is a subsequent backoff.
Since (\ref{eqn_Qbeta}) also occurs, Lemma \ref{lemma_NoBackOffQLong} with $t_1 = s, t_2=s+1/n$
and $\varsigma$ above shows that the probability of a subsequent backoff goes to 0 which
establishes the result.
\end{proof}


Since $s,t,Q$ are arbitrary in Lemma \ref{lemma_NoDelay}, it follows from continuity
that node $\ell$ begins service the instant its interferers become idle, if it has a positive
queue-length.

%
%

Lemmas \ref{lemma_conflict} and \ref{lemma_NoDelay} carry an implication for the node pairs $(1,2),(3,4),(5,6)$ in our network. We say that node $\ell_1$ dominates node $\ell_2$, $\ell_1 \neq \ell_2$ if ${\cal I}_{\ell_2} \subseteq {\cal I}_{\ell_1}$.
Hence, if (say) queue 3 (the dominant queue) is draining, then no other queue than 4 may be active as a consequence of  Lemma \ref{lemma_conflict}. But this implies all interferers of queue 4 are inactive. Hence, if $Q_4 > 0$, it will therefore begin to drain immediately, i.e. if queue 3 is draining so is queue 4. 
Also  if $Q_4$ becomes 0 before $Q_3$, then it must remain at 0, until queue 3 drains.

This result is formally stated in the following corollary, the proof of which is omitted for brevity.

Given any node $k\in \lc 1,\cdots, N\rc$, $Q_k \geq 0$, and  time $t$ define,
\begin{equation}
\Psi^k_{t,Q_k}(u) \doteq \ls Q_k - (u-t)(1-\lambda_k) \rs_+,~u \geq t,
\end{equation}
and given $v > t$, let $F_{t,v,\eta}^k$ be the event that
$\abs{\Psi^k_{t,Q_k(t,\omega)}(u) - Q_k(u,\omega)} < \eta$ for
$u \in [t,v]$.

\begin{corollary}
\label{cor_dominate}
Given a queue $\ell$, let $k$ be any other queue with  ${\cal I}_k \subseteq {\cal I}_\ell$. $\forall t \geq 0, Q > 0, \eta > 0$, define $v = t + Q/(1-\lambda_\ell)$, 
then with $P^{(\ell)}_{t,Q}$ as in Definition \ref{defn_poi}, it holds that,
$$
\muu{ P^{(\ell)}_{t,Q} \cap \lb F^{(k)}_{t,v,\eta} \rb^c} = 0. 
$$
\end{corollary}

Corollary \ref{cor_dominate} implies that $\mu$ almost surely the dominated node $k$
follows $\Psi^k$ the moment that dominating node $\ell$ becomes active. 

In case the arrival rates satisfy,
\begin{eqnarray}
\lambda_1 & = & \lambda_2 = \lambda > 0 \label{eqn_arrivalrates} \\
\lambda_4 & = & \lambda_5 \nonumber \\
\lambda_6 & > &  \lambda_5\nonumber \\
\lambda_3 & > & \lambda_4 \nonumber
\end{eqnarray}
Corollary \ref{cor_dominate} may be used to show
that the network enters a {\em natural state} (as defined in
Section~\ref{instab}) $\mu$ a.s. This result is
proved in the following theorem.
\begin{theorem}[Almost Sure Natural State]
\label{thm_natural}
Given the initial condition ${\bf Q}(0) = {\bf q}$ with $\euclidnorm{{\bf q}}=1$, there exits a $T_N > 0$ such that $\mu$ a.s. for all $t \geq T_N$,
\begin{eqnarray*}
Q_3(t) & \geq & Q_4(t), \\
Q_6(t) & \geq & Q_5(t).
\end{eqnarray*}
Moreover (recalling the definition of $\rho$ given in section \ref{brokendiamond})
$\exists \rho^* < 1$ such that for all $\rho \in [\rho^*,1)$, $\vee_\ell Q_\ell(T_N) > 0$
i.e. the network is non-empty at time $T_N$.
\end{theorem}
\begin{proof}
This result follows from Lipschitz continuity and more particularly  from the fact that
the sample paths are piecewise linear. Hence, apart from a set of measure
0, the derivatives of all queue lengths exist.

Consider now nodes 3 and 4. Where the derivatives exist and $Q_4 > 0$,
it holds that
$$
\frac{d Q_3}{dt} > \frac{d Q_4}{dt},
$$
since $\lambda_3 > \lambda_4$ and since $Q_4$ is decreasing at
linear rate whenever $Q_4 > 0$ and $Q_3$ is decreasing at a linear rate,
as shown in Lemma \ref{lemma_NoDelay}. We may therefore deduce $\mu$ a.s. and where
differentiability holds that,
$$
\frac{d \ls Q_4(t) -Q_3(t) \rs_+}{dt} \leq \lambda_4 - \lambda_3 < 0 ,~~
$$
until some time $T_3$, such that $\ls Q_4(t) -Q_3(t) \rs_+ = 0,~t\geq T_3$.
The same holds for queues 5 and 6, with corresponding time $T_6$ and the following inequalities
are satisfied,
$$
T_3 \leq \frac{\ls Q_4(0) - Q_3(0) \rs_+}{\lambda_3 - \lambda_4},~T_3 \leq \frac{\ls Q_5(0) - Q_6(0) \rs_+}{\lambda_6 - \lambda_5}.
$$
We may therefore take
$$
T_N = T_3 \vee T_6,
$$
and by taking worst case values in the above inequalities,
we obtain a uniform bound on $T_N$. This concludes the first part of the lemma.

We now show that $T_E$, the time to empty, can be taken arbitrarily large. Define
$L_P(t) \doteq \max \lb Q_1(t), Q_2(t) \rb + Q_3(t) + Q_6(t)$. Then $L_P$ can be reduced at most at rate 1, since service of nodes (1, 2), 3 and 6 is mutually exclusive, and grows at rate $\rho = \rho_0 + \rho_3 + \rho_6$, which can be made arbitrarily close to 1. Hence
$T_E \rightarrow \infty$ as $\rho  \uparrow 1$ if $L_P(0) > 0$. It can be the case that
$L_P(0) = 0$ but then $Q_4(0)+Q_5(0) = 1$, so that $L_P(1/2) = \rho_0/2$ and $T_E \geq \frac{1}{2} \lb 1 + \frac{\rho_0}{1 - \rho} \rb$ and
again $T_E \rightarrow \infty$ as $\rho  \uparrow 1$.
\end{proof}

This shows that a non-empty natural state can be reached in finite time, because of the dominance property. Given Theorem \ref{thm_natural} we can and will suppose that the state is natural at time 0,
without loss of generality.



We define the set of paths which additionally satisfy the constraints of Lemmas \ref{lemma_conflict} and \ref{lemma_NoDelay} to be $P_L \subset P \subset C[0,\infty)$. As previously, we now restrict the set of sample paths to $P_L$, so that the probability of an event $F \in {\cal C}$ can be
determined as $\muu{F} = \muu{F \cap P_L}$. This concludes Part~B.

\subsection*{B.II. Discussion}

We now give a largely informal description of the paths in $P_L$. Section \ref{desc}
gives a detailed description of the periods $M_k, k=1, 2, 3, 4$.  The ends of periods $M_1,M_2,M_3$
are marked by the corresponding stopping times $\Zstop{(1,2)}, \Zstop{(3)},\Zstop{(6)}$. For $M_4$
periods, the following construction is needed. (It is needed because $\Zstop{(4)}$ stopping times may be part of an $M_2$ period and hence do not mark the end of a $M_4$ period.)

We first define $\Pstop{(\ell)} = \Vstop{(\ell)} + \Bstop{(\ell)} \in m\FZ{\Zstop{(\ell)}}{}$ to be
the time prior to $\Zstop{(\ell)}$ when service begins (recall the definition of $\Bstop{(\ell)}$ in (\ref{eq Bstop})).
\begin{definition}
\label{defn_M4period}
A stopping time $\Zstop{(4)}$ is a $M_4^{(5)}$ stopping time, denoted $\Zstop{4,M_4^{(5)}}$
if the following holds,
\begin{eqnarray}
Q_5(\Zstop{4} - \Pstop{(4)}) & \geq  & Q_4(\Zstop{4} - \Pstop{(4)}) \label{eqn_M4condn},\\
I_\ell(\Zstop{4} - \Pstop{(4)}) & = & I_\ell(\Zstop{4}),~\ell=3,6 \label{eqn_M4II}.
\end{eqnarray}
That $\Zstop{4,M_4^{(5)}}$ is an $\cF_t$ stopping time follows as both the above events lie in
$\FZ{\Zstop{(4)}}{}$.
\end{definition}
This is consistent with an $M_4$-period taking place in which queue 4 emptied first (or at
the same time as queue 5) by (\ref{eqn_M4condn}).
If this is a strict inequality then we say this is a strict $M_4^{(5)}$ stopping time.
(\ref{eqn_M4II}) ensures that queue 5 is being served throughout $[\Pstop{(4)},\Zstop{(4)}]$ as a
consequence of Lemma \ref{lemma_NoDelay}. Similary we may define $\Zstop{5,M_4^{(4)}}$.

It is also possible that some subset of queues are all empty
with the remaining queues growing linearly. For example, at the end
of an $M_1$-period, it could be the case that both queues 1 and 2 remain at 0,
whilst the other queues continue to grow linearly. Similarly,
it could be the case that all of queues 3,4,5,6 become and remain 0
whilst queues 1 and 2 grow at rate $\lambda_1$.
In Part C, we will derive the probabilities according to which one
period is followed by another with no delay (on the fluid scale)
in switching from one period to the next. We concentrate on the case of switching
out of $M_1$ where the probability of the
next period depends only on the residual backoff times. 
%

%

%
%
%

\section{Fluid limit proofs: Part C}
\label{sec_M1switch}

We begin with some preliminary results. The first is for measures
constructed from closed continuity sets. Given a set of sample paths
$G$, define the improper probability measures,
$$
\muG{F} \doteq \muu{F \cap G},~~\mu^{(R)}_G(F) = \mu_R \lc F \cap G \rc.
$$
The following lemma shows that weak convergence is conferred on $\mu^{(R)}_G$ provided $G$ is closed
and a $\mu$-continuity set.
\begin{lemma}
\label{muG_lemma}
Suppose $\mu^{(R)}$ is a sequence of probability measures on a metric space,
$(\Omega, {\cal F})$,  such that
$$
\mu^{(R)} \Rightarrow \mu
$$
where $\mu$ is also a probability measure on the same space. Let $G \in {\cal F}$ be closed
and a $\mu$-continuity set. Then it holds that
$$
\mu_G^{(R)} \Rightarrow \mu_G
$$
\end{lemma}
In particular, the weak convergence definitions (iii), (iv), and (v), in Theorem 2.1 of \cite{Billingsley68}, all
equivalently hold.

The next lemma is concerned with the following. Suppose a pair of non-interfering queues in the
network are operating in isolation e.g. queues (1,2). Then each queue will be empty
and in fact will then subsequently be empty infinitely often, almost surely. Given that the
evolutions of the two queues are independent, it is plausible that the total number of steps in
the jump chain for which both queues are backed off together increases to infinity in a period which is negligible on the fluid scale.

Given a start time taken to be 0, define $W^R(u)$ to be the total number of steps that queues 1 and 2 are both
in backoff, starting at time 0 and ending at time $u > 0$ on the fluid scale, in
$\lb {\bf Q}^R(t), {\bf I}^R(t) \rb$. 
Partial periods between one clock tick and the next, at the start and
at the end are neglected. The following lemma supposes queues 1 and 2 are in isolation so
that no other nodes may gain the medium.
%
%
\begin{lemma}[Total Backoff]
\label{lemma_TotalBackOff}
Given $Q > 0$, define $t \doteq Q/(1-\lambda)$, and suppose that $
Q^R_\ell(0) \leq Q,~\ell=1,2,
$
and both queues are active at time 0. Then for any $Q, \xi > 0$,
$$
\lim_{R \rightarrow \infty} \muR{\WcR(t+2\xi) \geq 2 \sqrt{R} } = 1.
$$
\end{lemma}
\begin{proof}
Let $\tau_{0,0}$ be the stopping index in the jump chain for the first occurrence of
\begin{equation}\label{eqn_kempty}
X_1(\tau_{0,0}) = X_2(\tau_{0,0}) = 0.
 \end{equation}
Given any $\xi >0$, define $p^R_{\xi,Q} \doteq \prob{ \tau_{0,0} \leq \floor{\beta(t+\xi)R} | X_\ell(0) \leq RQ }.$
It will be enough to show that $p^R_{\xi,Q} \rightarrow 1$ as $R \rightarrow \infty$. To see this, note that any queue in isolation is positive recurrent, as a consequence of Lemma \ref{lemma_noearlybackoff}. Thus, the jump chain restricted to nodes 1 and 2 in isolation (i.e., with remaining queues barred from gaining the medium) is also positive recurrent. Let $m_0$ be the mean number of steps between indices $k$ such that
(\ref{eqn_kempty}) is again satisfied. Also let $K^R_\xi$ be the random number of such steps in
the next interval of $\floor{\beta \xi R}$ steps. It is easily seen from the weak law of large numbers that
$$
\lim_{R \to \infty} \muR{  K^R_\xi > \frac{\floor{\beta \xi R}}{2 m_0}}= 1,
$$
which implies the statement of Lemma.

Thus, to complete the proof, we just need to show that $p^R_{\xi,Q} \rightarrow 1$. Fix $\varepsilon_{X_T} > 0$ and choose $X_T :=X_T(\lambda, \gamma)< \infty$ as in (\ref{eqn_QTveps}) so that the probability of even a single backoff before either queue reaches $X_T$ is no more than $\varepsilon_{X_T}$. Moreover let $\tau_{T,\ell}, \ell=1,2$ be the stopping indices for
$X_\ell(\tau_{T,\ell}) = X_T$. Then, given any $\eta>0,$ and $\varepsilon_{R,\eta}>0$, it can be seen that
$\tau_{T,1}  \vee \tau_{T,2} \leq \beta R(t+\eta)$ occurs with probability larger than $1 - 2\varepsilon_{R,\eta} -2\vQT$, with $\varepsilon_{R,\eta} \to 0$ as $R\to \infty$ by the weak law of large numbers.

Next, given any $\varepsilon_L > 0$, there exists a $X_L$ large enough such that $\prob{ X_\ell(\tau_{T,\ell} + k ) \leq X_L} > 1 - \varepsilon_L$ for all $k \in \PNN$.
This follows from the fact that the jump chain in isolation is positive recurrent, and thus the corresponding sequence of infinite probability vectors is tight as they are converging to the steady-state distribution. Hence, with probability larger than $1 - 2\varepsilon_{R,\eta} -2\varepsilon_{X_T} - 2\varepsilon_L$, $X_\ell( \floor{\beta R(t+\eta)} \leq X_L, \ell=1,2$.

Moreover, again by the positive recurrence of the isolated jump chain, the mean number of steps for queues 1 and 2 both to become 0, starting from any state with $X_\ell\leq X_L,$ $\ell=1, 2$, is bounded by some constant $m_L := m_L(X_L) < \infty$.
Thus, by Markov's inequality, with a probability less than than  $m_L/(\eta R)$, in a further $\eta R$ steps both queues will become 0 (and thus inactive).

Finally, given any $\epsilon>0$, choose $X_T$ and $X_L$ large enough so that $\varepsilon_{X_T}< \epsilon /8$ and $\varepsilon_L < \epsilon/8$ and then $R$ sufficiently large
so that $\varepsilon_{R,\eta}<\epsilon/8$ and $m_L/(\eta R) <  \epsilon/8$. Hence, with probability larger than $1-\epsilon$, $\tau_{0,0} < (t+2\eta)R$ for all $R$ sufficiently large. Since $\epsilon$ and $\eta$ are arbitrary, the proof is complete.
\end{proof}

\subsection*{C.I. Transition from an $M_1$-period}
\label{sec_M1trans}


%
%
In what follows 
we will further suppose that the lengths of queues 1 and 2
and their activity are both equal, as the following arguments are readily modified where this is not the case. We therefore denote their common queue length as $Q(u) = Q_1(u) = Q_2(u)$ in what follows
and similarly for the activity $I(u) = I_1(u) = I_2(u)$.
Finally, in the following $t,c$ and hence $s$ are fixed,
\begin{eqnarray*}
s & \doteq &  t - \frac{c}{1 - \lambda }, \\
\delta_k & \doteq & \alpha_k c,~0 < \alpha_k < 1,~k=0,1,\\
h & \doteq & \nu c, \\
\zeta & \doteq & \chi c,~\nu > \chi > 0,
\end{eqnarray*}
for some small positive constants $\alpha_k$, $\nu$, and $\chi$ to be determined later.
We are now ready to define the following closed set of paths,
\begin{equation}
\Gct \doteq \lc \omega : 0 < c - \delta_0 \leq Q(s,\omega) \leq c + \delta_1 \rc \cap
\lc \omega : I(s+h, \omega) - I(s,\omega) \geq \beta(h-\zeta) \rc.
\label{eqn_Gctdefn}
\end{equation}
$G_{c,t}$ is constructed to correspond to an $M_1$-period.

Now given $0 < s_1 < s_2$, and $\ovt$ (which will be specified later), define
\begin{equation}
I^{(3,4)}_{c,t} \doteq \JS{3}{=}(\ovt + s_1,s_2-s_1,\beta(s_2-s_1)) \cap \JS{4}{=}(\ovt + s_1,s_2-s_1,\beta(s_2-s_1)).
\label{eqn_Ictdefn}
\end{equation}
$\Ict{(3,4)}$ is a (closed) set of paths for which queue 3 (and also queue 4) are fully active
during the interval $[\ovt+s_1,\ovt+s_2]$. Similar definitions, using the same $s_1$, $s_2$, and $\ovt$,
can be made for $\Ict{(4,5)},\Ict{(5,6)}$.

The first set of paths, $\Gct$, is illustrated in the dashed lines in Figure \ref{fig_GI}.
\begin{figure}[!Ht]
  \begin{center}
    \includegraphics[width=8 cm]{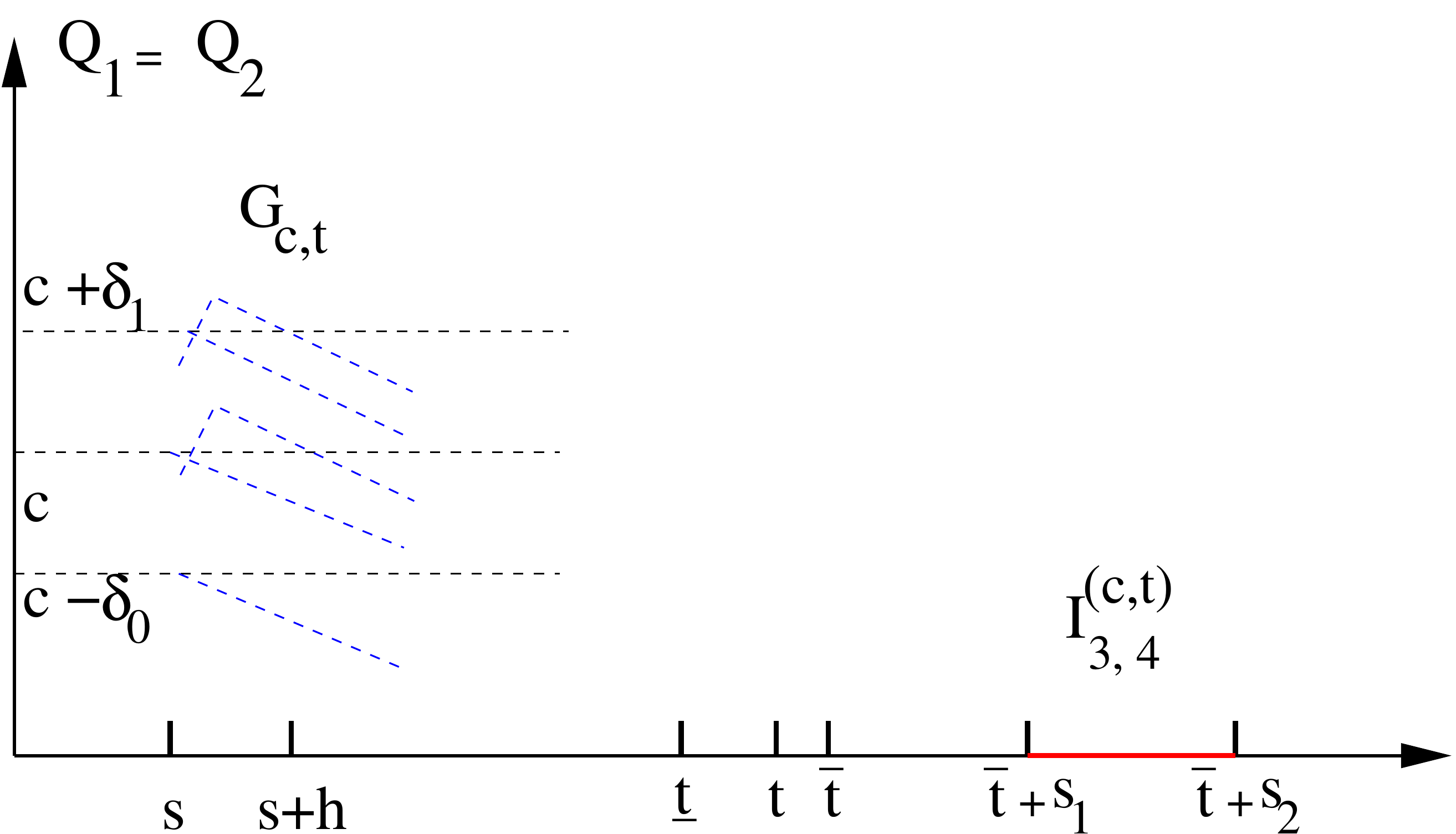}
    \caption{Sample Paths for the sets $\Gct$ and $\Ict{(3,4)}$.}
    \label{fig_GI}
  \end{center}
\end{figure}
Note all sample paths must pass through the interval $[c-\delta_0,c+\delta_1]$ at time $s$, but may continue to increase for a brief period at the beginning. After $s+h$ the two queues must be draining
at rate $1 - \lambda$ almost surely  as shown in Lemma \ref{lmma_NoBackOff}. The
red interval to the right indicates periods where one of the other three queue pairs
are expected to have the medium during the interval $[\ovt+s_1,\ovt+s_2]$. Only one such pair will be active during this period as a result of the forthcoming construction.

The following is the earliest time that queues 1 and 2 can drain if the sample paths are constrained
to lie in $\Gct$,
\begin{equation}
\udt \doteq t - \frac{\alpha_0}{1 - \lambda} c.
\label{eqn_udtdefn}
\end{equation}
As far as additional queue build up is concerned, under the fluid limit,
$$
  Q(s+h, \omega)  \leq  Q(s,\omega) + \lambda h = Q(s,\omega) + \lambda \nu c
$$
holds for sample paths in $G_{c,t}$ (see (\ref{eqn_accountq})). It then follows that queues
1 and 2 will reach 0 under the fluid limit no later than
\begin{equation}
\ovt \doteq t + \frac{\alpha_1 + \lambda \nu}{1 - \lambda} c,
\label{eqn_ovtdefn}
\end{equation}
which is the definition for $\ovt$. We thus conclude that, under the fluid limit, queues 1 and 2 will
reach 0 in the interval $(\udt,\ovt)$ (for the first time after $s+h$ on occurrence of the
event $\Gct$). We formalize the above in the following lemma,
\begin{lemma}[Queue Bounds]
\label{lemma_qbnds}
Let $\tau^0_{c,s} \doteq \tau_c(s,\lc 0 \rc ) = \inf \lc t \geq s: Q(t) = 0 \rc$ be the first contact time with 0 for $Q = Q_1 = Q_2$.
Then,
$$
\muu{ \Gct \cap \lc \omega : \tau^0_{c,s}(\omega) \not\in [\udt,\ovt] \rc } = 0.
$$
Additionally, $\forall \ell = 3, 4, 5, 6,$
\begin{equation}
\muu{\Gct \cap \lc \omega: Q_\ell(\udt,\omega) < \Delta t \lambda_\ell \rc  } = 0,
\label{eqn_Qtbnd}
\end{equation}
where,
$$
\Delta t \doteq \udt - (s+h) = c(1-\alpha_0-\nu(1-\lambda))/(1-\lambda).
$$
\end{lemma}
\begin{proof}
By definition of $G_{c,t}$, $Q(s,\omega) \geq c - \delta_0, \forall \omega \in G_{c,t}$. It follows
that 
$Q$ cannot reach 0 before $\udt$, as sample paths by definition lie in $P_L$ (see Part A.VI, following Theorem \ref{thm_linearpaths}). A similar argument applies to $\ovt$.

For the last part,  Lemma \ref{lemma_conflict}, shows that nodes $3,4,5,6$ must be idle in the period
$[s+h,\udt]$. Since the sample paths are restricted to  lie in $P_L$, it follows that their queues must satisfy the stated inequality at time $\udt$. The proof is complete.
\end{proof}

The time for queue $\ell$ to reach 0 following $\udt$ 
is therefore at least,
$$
f_\ell \doteq \Delta t \times \frac{\lambda_\ell}{1 - \lambda_\ell},~\ell = 3, 4, 5, 6.
$$
Clearly $ f_\ell \rightarrow (c \times \lambda_\ell)/\lb (1- \lambda) \times (1 - \lambda_\ell) \rb$ as $\alpha_0,\nu \downarrow 0$,
and so this expression is bounded from below as $\alpha_0,\alpha_1,\nu > \chi$ are made arbitrarily small. For future use, we define
$$
\unf \doteq \Delta t \wedge_{\ell=3}^6 \lambda_\ell/(1-\lambda_\ell),
$$
as a lower bound on the time needed to drain any queue $\ell = 3,4,5,6$.

Our results thus far do not rule out the possibility that there is an idle period during which
queues 3, 4, 5, 6 fail to obtain the medium. In order to make allowance for this,
we introduce a period $\xi c$, $ \xi > 0$, which comes following queues 1 and 2 draining, and to be definite, we set $\xi c = \unf/8$. Hence, if it is the case that
\begin{equation}
\ovt - \udt < \unf/4
\label{eqn_tdiff}
\end{equation}
and that service of queue $\ell$ cannot start before $\udt - \xi c$ and must have started no later
than $\ovt + \xi c$, then it follows that
service will continue throughout the interval $[\ovt + \xi c, \ovt + \xi c +\unf/2]$. In this
case, we may take $s_1 = \unf/4, s_2 = \unf/2$ again to be definite. Further, set $s_3 = \xi c + \unf/2$.
To summarize, if (\ref{eqn_tdiff})
holds, on occurrence of $\Gct$ and that service of queues 3 and 4 commences in the interval
$[\udt - \xi c,\ovt + \xi c]$, then the event $\Ict{(3,4)}$ must take place.
The same is true in case service commences for either queue pair $(4,5)$ or $(5,6)$ in
$[\udt - \xi c,\ovt + \xi c]$.

Let $\hC_k,k=3,4,5,6$, be the residual time to backoff for queues 3, 4, 5, 6, at time $s+h$, with
$\hC_1=\hC_2=0$ as these queues will be almost surely active. Define $S_M$ to be the number
of steps in the jump chain before one of these nodes gains the medium and also define
\begin{eqnarray}
W^{(3,4)} & \doteq &  \lc \hC_3 <  \wedge_{k=4}^6 \hC_k \rc  \cup \lb \lc \hC_4 <  \hC_3 \wedge \hC_5 \wedge \hC_6 \rc
\cap \lc \hC_3 < \hC_5 \rc \rb,\\
\label{def_Wct}
\Cct{3,4} & \doteq & W^{(3,4)} \cap \lc S_M 
\leq \sqrt{R} \rc. \nonumber
\end{eqnarray}
$W^{(3,4)}$ is the event that queues 3 and 4 win the backoff competition to take the medium first from
queues 1 and 2. Similar definitions can be made for queues $(4,5)$ and for queues $(5,6)$ in addition. The probabilities of these events are
\begin{equation}
\prob{W^{(3,4)}} = \frac{3}{8} = \prob{W^{(5,6)}},~\prob{W^{(4,5)}} = \frac{1}{4},
\label{eqn_clocks}
\end{equation}
as the backoff periods are unit mean i.i.d. exponential random variables. $\Cct{3,4}$ is the event that queues $(3,4)$ win the backoff competition and  that it does so
in no more  than $\sqrt{R}$ of the jump chain steps when queues 1 and 2 are in backoff together.

Next let
$$
\BBR{(1,2)} \doteq \NNR{(1,2)}(s+h,\udt-\xi c) \cap
\{W^R(s+h,\ovt+\xi c) \geq 2 \sqrt{R}\}
$$
be the intersection of the event $\NNR{(1,2)}(s+h,\udt-\xi c)$
that neither queue~1 nor queue~2 starts to backoff during the time
interval $[s+h,\udt-\xi c]$ and the event
$\{W^R(s+h,\ovt+\xi c) \geq 2 \sqrt{R}\}$ that queues~1 and~2
operating in isolation would be simultaneously in backoff for
a cumulative period of time of at least $2 \sqrt{R}$ during the
interval $[s+h,\ovt+\xi c]$.
Informally speaking, the event $\BBR{(1,2)}$ ensures that there
is sufficient backoff by queues~1 and~2 and that they do not begin
to backoff whilst there are a significant number of queue 1 or queue 2
packets remaining.

Next define $c_Q$ to be, 
$$
c_Q \doteq \frac{ s_3 - s_2 }{2} \times \lb 1- \lambda_3 \rb > 0,
$$
which is at least half the content of queues 3 and 4 on the fluid scale at time $\ovt+s_2$, given our construction. Further, define the following event
\begin{equation}
\QQR{(3,4)}(\udt,\ovt+s_2) \doteq \lc \omega: \inf \lc Q^R_m(u,\omega),
u \in [\udt,\ovt+s_2] \rc > c_Q, m = 3, 4 \rc \in \cF_{\ovt+s_2},
\label{eqn_Q34defn}
\end{equation}
for which we obtain the following corollary.
\begin{corollary}
\label{cor_Q34cGct}
$$
\lim_{R \rightarrow \infty} \muR{\Gct \cap \lb \QQR{(3,4)} \rb^c } = 0
$$
\end{corollary}
\begin{proof}
Lemma \ref{lemma_qbnds} implies that for all $n$ sufficiently large,
$$
\limsup_{R \rightarrow \infty} \muR{\Gct \cap \lc \omega: Q^R_\ell(\udt,\omega) \leq \Delta t \lambda_\ell -1/n \rc} = 0,~\ell = 3,4,
$$
on using Theorem 2.1 in \cite{Billingsley68} and that both the above sets are closed. Hence we need
only show that,
\begin{equation}
\lim_{R \rightarrow \infty} \muR{\lc \omega: Q^R_\ell(\udt,\omega) > \Delta t \lambda_\ell-1/n, \ell=3,4 \rc  \cap \lb \QQR{(3,4)} \rb^c } = 0,
\label{eqn_cor3}
\end{equation}
for sufficiently large $n$. However (\ref{eqn_cor3}) follows from the weak law of large
numbers, from the definition of $\Delta t$, $c_Q$, and the event $\QQR{(3,4)}$.
\end{proof}

Finally, define $\NNR{(3,4)}(\udt,\ovt+s_2)$ to be the event that
neither queue~3 nor queue~4 has a backoff during the time interval
$[\udt,\ovt+s_2]$ (on the fluid scale). Clearly, equivalent definitions for this and the above
and the above corollary can be made for queue pairs $(4,5),(5,6)$.

In what follows it will be convenient to write $G := G_{c,t}$.
%
Our aim now is to show that no matter what trajectory the fluid limit path followed earlier, if it
lies in $G$ so that queues 1 and 2 almost surely reach 0 in the interval
$[\udt,\ovt]$, marking the end of an $M_1$ period,  then the probability of the next period depends only on the residual backoff times, which is a Markov property.

\begin{lemma}
\label{lem_backclock}
Suppose that $G$ is a set of paths as defined in (\ref{eqn_Gctdefn}), with parameter values so
that (\ref{eqn_tdiff}) holds,
 and is also a $\mu$-continuity set. In addition, let $F \in {\cal F}_s$ be an arbitrary closed,
finite-dimensional set of paths defined by times $s$ and earlier. It then holds that
\begin{eqnarray*}
\muG{F \cap \Ict{(3,4)}} & \geq & \frac{3}{8} \muG{ F^o} \\
\muG{F \cap \Ict{(5,6)}} & \geq & \frac{3}{8} \muG{ F^o} \\
\muG{F \cap \Ict{(4,5)}} & \geq & \frac{1}{4} \muG{ F^o}
\end{eqnarray*}
In case $F$ is a $\mu$-continuity set, the interior can be dropped and $\geq$ replaced with
equality.
\end{lemma}
\begin{proof}
We first show the last part of the lemma, assuming the first part to be true. If $F$ is a $\mu$-continuity set, then by definition,
$0 = \muu{\partial F} \geq \muu{G \cap \partial F}$ and it follows that $F$ is a
$\mu_G$-continuity set as well. Since the factors
sum to 1 and the events on the left are almost surely exclusive as a consequence of
Lemma \ref{lemma_conflict}, we can now replace the inequality sign with equality.

We move to the first part of the lemma, which we will prove for queues 3 and 4. The proof for the other
queue pairs is similar.

First observe that
$$
\Cct{3,4} \cap \BBR{(1,2)} \cap \NNR{(3,4)}(\udt,\ovt+s_2) \subseteq \Ict{(3,4)},
$$
since $\Cct{3,4} \cap \BBR{(1,2)}$ implies that queues~3 and~4
activate before time $\ovt+s_1$, while $\NNR{(3,4)}(\udt,\ovt+s_2)$
ensures that neither queue~3 nor queue~4 has a backoff during the time
interval $[\udt,\ovt+s_2]$.
We thus obtain the following chain of inequalities
\begin{eqnarray}
\muGR{F \cap \Ict{(3,4)}}
&\geq&
\muGR{F \cap \Cct{3,4} \cap \BBR{(1,2)} \cap \NNR{(3,4)}} \label{eqn_Fineq} \\
&\geq&
\muGR{F \cap W^{(3,4)}} -
\muGR{\lb\lc S_M \leq \sqrt{R}\rc \cap \BBR{(1,2)} \cap \NNR{(3,4)}\rb^c} \nonumber \\
&\geq&
\frac{3}{8} \muGR{F} - \muGR{S_M > \sqrt{R}} -
\muGR{\lb \BBR{(1,2)}\rb^c} - \muGR{\lb \NNR{(3,4)}\rb^c}, \nonumber
\end{eqnarray}
with $\NNR{(3,4)} \equiv \NNR{(3,4)}(\udt,\ovt+s_2)$ for compactness.
The first line follows by inclusion, the second using
$\muG{A \cap B} \geq \muG{A} - \muG{B^c}$, and the third from
(\ref{eqn_clocks}) by independence of the back-off clocks
and by using the union bound in conjunction with de Morgan's laws.
We now proceed to show that
\begin{eqnarray*}
\mu^{(R)}_{G,1} &\doteq& \muR{ S_M
> \sqrt{R}} \rightarrow 0 \\
\mu^{(R)}_{G,2} &\doteq& \muR{\lb \BBR{(1,2)} \rb^c \cap \Gct} \rightarrow 0 \\
\mu^{(R)}_{G,3} &\doteq& \muR{\lb N_R^{(3,4)} \rb^c \cap \Gct} \rightarrow 0
\end{eqnarray*}
The first limit is immediate. 

In order to deal with the second limit, define the event
$$
\QQR{(1,2)}(s+h,\udt-\xi c) \doteq \lc \omega: \inf \lc Q^R_m(u,\omega),
u \in [s+h,\udt-\xi c] \rc > \varsigma, m = 1, 2 \rc \in \FZ{\udt - \xi c}{}
$$
for some small constant $\varsigma > 0$, and use the upper bound
$$
\mu^{(R)}_{G,2} \leq
\muR{\lb\BBR{(1,2)}\rb^c \cap \QQR{(1,2)}(s+h,\udt-\xi c) \cap \Gct} +
\muR{\lb\QQR{(1,2)}(s+h,\udt-\xi c)\rb^c \cap \Gct}
$$
The limit of the second term is 0 by definition of $\udt$ as
the earliest time that queues~1 and~2 can drain under the event $\Gct$ and on making a suitable
choice for $\varsigma$. It suffices then to show that the limit of the first term is 0.
In order to prove this, we invoke the definition of the event
$\BBR{(1,2)}$ to obtain that the first term is bounded from above by
$$
\muR{\lb\NNR{(1,2)}(s+h,\udt-\xi c)\rb^c \cap \QQR{(1,2)}(s+h,\udt-\xi c)} +
\muR{\{W^R(s+h,\ovt+\xi c) \leq 2 \sqrt{R}\} \cap \Gct}
$$
That the first term converges to 0 follows by definition of the events and Lemma \ref{lemma_NoBackOffQLong}. Lemma \ref{lemma_TotalBackOff} shows
that the limit of the second term (i.e. the event there is insufficient backoff by queues 1 and
2 on occurrence of $\Gct$) is 0.

In order to handle the third limit, we apply the upper bound
$$
\mu^{(R)}_{G,3} \leq \muR{\lb\NNR{(3,4)}(\udt,\ovt+s_2)\rb^c \cap \QQR{(3,4)}} +
\muR{\lb\QQR{(3,4)}\rb^c \cap \Gct}
$$
Lemma \ref{lemma_NoBackOffQLong} shows that the limit of the first term is 0, whilst
the statement of Corollary \ref{cor_Q34cGct} is that the limit of the second term is 0.

Taking limits in (\ref{eqn_Fineq}) with respect to $R$, and using Lemma \ref{muG_lemma}, it follows that,
\begin{eqnarray}
\muG{F \cap \Ict{(3,4)}} &\geq& \frac{3}{8} \limsup_R \muGR{F} \\
              &\geq& \frac{3}{8} \liminf \muGR{F^o} \nonumber \\
              &\geq& \frac{3}{8} \muG{F^o} \nonumber
\end{eqnarray}
where the first inequality follows from the fact that $F$
and $\Ict{(3,4)}$ are both closed and the third since $F^o$ is open
and again from Lemma \ref{muG_lemma}.
\end{proof}

Let $\mu$ be the fluid limit measure and  proceed to define for any given $t \geq 0$ the following
class of sets, the {\em finite-dimensional continuity rectangles} ${\cal K}_{\mu,t}$ which are a
subset of the {\em finite-dimensional sets}, ${\cal H}_t$.
\begin{definition}
\label{defn_rectangle}
Define the  class of finite closed rectangles
${\cal R}$ to be the sets,
$$
\lb \prod_{j=1}^{N} [q_{j,L},q_{j,H}]\rb \times \lb \prod_{j=1}^{N} [r_{j,L},r_{j,H}]\rb \subset
\RRp^N \times \RRp^N,
$$
where $q_{j,L} \leq q_{j,H}, r_{j,L} \leq r_{j,H}$, otherwise we obtain the empty set.

Given times
$0\leq t_1 < t_2 < \cdots < t_J \leq t$, define $\pi_{J,t}: C[0,\infty) \rightarrow E^J$ to be the (continuous) projection map taking the sample path to its position at times $t_1,\cdots,t_J$
$$
\pi_{K,t}(\omega) = \Big( \lb {\bf Q}(t_1,\omega),{\bf I}(t_1,\omega)\rb, \cdots,
\lb {\bf Q}(t_J,\omega), {\bf I}(t_J,\omega) \rb \Big).
$$
Finally, take ${\cal R}^J$ to be $J$-products of closed rectangles.
Define ${\cal K}_t$ to be sets of the form $\pi_{K,t}^{-1} R_J, R_J \in {\cal R}^J$ and
finally ${\cal K}_{\mu,t} \subset {\cal K}_t$ to be those $H \in {\cal K}_t$ such that
$\muu{\partial H} = 0$. Clearly ${\cal K}_{\mu,t} \subset {\cal K}_t \subset {\cal F}_t$.
\end{definition}
Returning to Lemma \ref{lem_backclock}, we see that it is satisfied by
all sets $F \in {\cal K}_{\mu,s}$ with equality since they are by
definition closed $\mu$-continuity sets.  Furthermore, since the terms on the
left and on the right are measures and since ${\cal K}_{\mu,s}$
generates ${\cal F}_s$, the following corollary holds.
\begin{corollary}
\label{cor_Fwk}
$\forall F \in {\cal F}_s$, Lemma \ref{lem_backclock} holds with equality, i.e.
\begin{eqnarray*}
\muG{F \cap \Ict{(3,4)}}  & = & \frac{3}{8} \muG{ F} \\
\muG{F \cap \Ict{(5,6)}}  & = & \frac{3}{8} \muG{ F} \\
\muG{F \cap \Ict{(4,5)}}  & = & \frac{1}{4} \muG{ F}
\end{eqnarray*}
\end{corollary}
\begin{proof}
First note that the measures on the LHS and RHS are both finite and are therefore both
$\sigma$-finite, with respect to the sets in ${\cal K}_{\mu,s}$. It is readily shown that
${\cal K}_{\mu,s}$ is a $\pi$-system and $\sigma({\cal K}_{\mu,s}) = \cF_s$. Theorem 10.3 in
\cite{Billingsley95} thus shows that LHS and RHS agree on $\cF_s$.
\end{proof}


To continue towards Theorem \ref{thm_switchM0}, we now define paths that one of which is followed immediately on completion of a (positive) $M_1$-period at time
$s$, $\mu$ a.s. First define
\begin{equation}
\vartheta_{s,t}^{({\bf Q},M_k)}(u),~u\in [s,t], k=2,3,4,
\end{equation}
to be the path which is at ${\bf Q}$ at time $s$ and then follows $M_k$ until time $t$,
e.g., if $k=1$, queues 1 and 2 are decreasing linearly at rate $(1-\lambda)$ and any other
queue $\ell=3,4,5,6$ is increasing at rate $\lambda_\ell$. Precise
definitions we omit as the form of the sample paths have already been discussed. The
next definition is for an indicator function that the above path is being followed in
an interval $[s,s+h], h > 0$.
\begin{equation}
\label{eqn_vartheta}
\one{(s,h,\eta)}{M_k, {\bf Q}} \doteq \1\lc \omega: \euclidnorm{{\bf Q}(v,\omega)
 - \vartheta_{s,s+h}^{({\bf Q}(s,\omega),M_k)}(v)} < \eta, v\in [s,s+h] \rc.
\end{equation}
In words  $M_k$ is "followed" for an interval of duration $h$ starting at $s$ to
a closeness $\eta$.

%
%

%
%

Note that the result of Corollary \ref{cor_Fwk} applies only to events in some $\sigma$-algebra $\cF_w$ where $w \geq 0$ is fixed. However, we require
that equivalent results be established for all events $F \in \FZ{\Zstop{(1,2)}}{}$. This issue can be approached as follows.

Given $s < t, a < b$, with $n \in \NN_0$ and recalling Definition \ref{defn_Zseq}, set
$$
C^{(n)} \doteq \lc \omega : (\Zstop{(1,2)}, \Astop{(1,2)}) \in  [s,t) \times [a,b) \rc,~~
$$
then it is readily seen that,
$$
F \cap \Cn \in \cF_t,~\forall F \in  \FZ{\Zstop{(1,2)}}.
$$

Suppose that $s < t$ in the definition of $\Cn$ satisfy
$
t-s < \lambda_1 e_m,
$
and $e_m <  a < b  \leq f_m$ ($e_m$ and $f_m$ are given in Definition \ref{defn_Zseq}). Then, it can be seen that for all paths $\omega \in F \cap \Cn$, for any $F \in \FZ{\Zstop{(1,2)}}{}$,
we can find a $w < s$ such that $I_\ell(s,\omega) - I_\ell(w,\omega) = \beta(s-w)$, $\ell=1,  2$, i.e., the queues and activity components constitute a set of parallel lines over the interval $[w,s]$.

The above intuitive argument can be formalized by establishing the existence of an equivalent $\sigma$-algebra. We say that the $\sigma$-algebra $\FZ{\Zstop{(1,2)}}{} \cap \Cn$ is {\em
equivalent} to a sub $\sigma$-algebra, ${\cal H}_w \subset \cF_w$,
if to each event $H \in \FZ{\Zstop{(1,2)}}{} \cap \Cn$ there is an event
$H_w \in \cH_w$ so that $H = H_w$.\\

\begin{lemma}[Equivalent $\sigma$-algebra]
\label{lemma_sigmaequiv}
Given $n \in \ZZ$, arbitrary $t > s \geq 0$ such that $t-s < \lambda_1 e_m$, and $e_m <  a < b  \leq f_m$ and arbitrary $w\in (t-\lambda_1 e_m , s)$, there is a $\sigma$-algebra, $\cH_w \subset \cF_w$ equivalent to $\FZ{\Zstop{(1,2)}}{} \cap \Cn$.\\
\end{lemma}

We omit the proof. 
Set $t_Z = \Zstop{(1,2)}$, $h^{(1,2)} = \lambda_1 e_m \times \wedge_{\ell=3}^6 \lambda_\ell/(1-\lambda_\ell),$
and ${\bf Q}_Z = {\bf Q}(\Zstop{(1,2)})$. Next define $\one{(\eta,M_1)}{M_k,m,n}$ to be $\one{(t_Z,h^{(1,2)},\eta)}{M_k,{\bf Q}_Z}, k=2, 3, 4$ if $\Zstop{(1,2)} < \infty$. Define ${\cal F}^\infty_{\Zstop{(1,2)}} \doteq   {\cal F}_{\Zstop{(1,2)}}\cap\lc \Zstop{(1,2)} < \infty \rc$  as we are only interested in finite stopping times.\\
\begin{theorem}
\label{thm_switchM0}
$\forall n \in \ZZ, m \in \ZIO$, $\exists \eta_m$ such that $\forall \eta,~\eta_m > \eta > 0$,
\begin{eqnarray}
\muu{\one{(M_1,\eta)}{M_2,m,n} | {\cal F}^\infty_{\Zstop{(1,2)}}}  & =  & \frac{3}{8}, ~\mu~a.s.,
\label{eqn_condprobM1M2} \\
\muu{\one{(M_1,\eta)}{M_3,m,n} | {\cal F}^\infty_{\Zstop{(1,2)}}}  & =  & \frac{3}{8}, \nonumber \\
\muu{ \one{(M_1,\eta)}{M_4,m,n} | {\cal F}^\infty_{\Zstop{(1,2)}}}  & =  & \frac{1}{4}.\\ \nonumber
\end{eqnarray}
\end{theorem}
Since $\eta > 0$ can be taken arbitrarily small, the conclusion is that one of the
periods $M_2,M_3,M_4$ start immediately at $\Zstop{(1,2)}$ on occurrence of
$\Zstop{(1,2)} < \infty$ and with probabilities determined solely by the residual backoff times.

\begin{proof}
Given $F \in \FZ{\Zstop{(1,2)}}{\infty}$, we may write $F = \cup_k F_k$ as a countable union of disjoint sets.
$F_k$ is obtained by intersection of $F$ with the disjoint sets,
$$
C_k \doteq \lc \omega : (\Zstop{(1,2)}, \Astop{(1,2)}) \in  (s_k,t_k] \times (a_k,b_k] \rc,
$$
where $e_m \leq a_k < b_k \leq f_m$ and $e_m \leq s_k < t_k$ are chosen according to $\eta$ in a way to be  described subsequently. $C_k \cap C_m = \emptyset,~m\neq k$ is constructed by first choosing the intervals for the stopping time $\Zstop{(1,2)}$ to be disjoint and then likewise the durations into disjoint semi-open intervals.
Thus, $
F_k  \doteq  F\cap C_k \in {\cal F}_{t_k}.
$

We turn to $F_k$ and will suppose that $t_k-s_k$ is sufficiently small so that we may find a time
$w_k \in (t_k - e_m \lambda_1,s_k)$
as in Lemma \ref{lemma_sigmaequiv}.
$w_k$ will be a constant determined   by $a_k,b_k,s_k,t_k$ and $\eta$ only. For the moment suppose that $w_k$ and $\eta$ are used to determine constants $c_\eta,t_{w_k}$ and then a set $G_{c_\eta,t_{w_k}}$, satisfying the conditions of Lemma \ref{lem_backclock}, such that in addition,
\begin{equation}
C_k \subset G_{c_\eta,t_{w_k}},
\label{eqn_CsubG}
\end{equation}
with the $s$ in the definition of $G_{c,t}$, see (\ref{eqn_Gctdefn}), taken to be $w_k$.
It can then be seen that the following chain of equalities hold,
\begin{eqnarray}
\muu{F_k \cap I^{(3,4)}_{c_\eta,t_{w_k}}} & = &  \muuu{G_{c_\eta,t_{w_k}}}{F_k \cap I^{(3,4)}_{c_\eta,t_{w_k}} } \label{eqn_threeeight} \\
                  & = & \muuu{G_{c_\eta,t_{w_k}}}{F^{(w_k)}_k \cap I^{(3,4)}_{c_\eta,t_{w_k}} } \nonumber \\
                 & = & \frac{3}{8}\muuu{G_{c_\eta,t_{w_k}}}{F^{(w_k)}_k} \nonumber \\
                 & = & \frac{3}{8} \muuu{G_{c_\eta,t_{w_k}}}{F_k } \nonumber \\
                 & = & \frac{3}{8} \muu{F_k } \nonumber
\end{eqnarray}
The first equality follows from (\ref{eqn_CsubG}), the second from Lemma \ref{lemma_sigmaequiv} as there exits a $F^{(w_k)}_k\in \cF_{w_k}$ such that $F^{(w_k)}_k=F_k $, the third from Corollary \ref{cor_Fwk} and by definition of $G_{c_\eta,t_{w_k}}$, the fourth  equality is again from Lemma  \ref{lemma_sigmaequiv}, and the final one follows again from
(\ref{eqn_CsubG}). 
Corresponding results follow for $I^{(4,5)}_{c_\eta,t_{w_k}}$ and $I^{(5,6)}_{c_\eta,t_{w_k}}$. Once one of these events has occurred, $\mu$ almost surely the queues corresponding to the active period proceed to empty because they lie in  $P_L$ and therefore $P$. Moreover, $\eta_m$ is determined depending on the duration (at least $\lambda_1 e_m$) of the $M_1$ period. $\eta_m$ is taken sufficiently small,
 so that if we take any $\eta, \eta_m > \eta > 0$, only the node pair (and corresponding $M_k$ period) can satisfy the constraints in (\ref{eqn_vartheta}), for the interval $[\Zstop{(1,2)}, \Zstop{(1,2)}+h^{(1,2)}]$.

The above steps may be taken provided that i) $G_{c_\eta,t_{w_k}}$ is a closed $\mu$-continuity set,
ii) $G_{c_\eta,t_{w_k}}$ contains $C_k$ and hence $F_k$ iii) $\overline{t}+ s_3 - w_k$ is sufficiently small, so that the paths $\vartheta_{.,.}^{({\bf Q},M_k)}$ satisfy the constraints as in
(\ref{eqn_condprobM1M2}), and iv) the condition (\ref{eqn_tdiff}) must hold so that the
conditions of Lemma \ref{lem_backclock} and also of Corollary \ref{cor_Fwk} are met.

To show that $c_\eta, t_{w_k}$ and a corresponding $G_{c_\eta,t_{w_k}}$ exist, given $\eta>0$, set $t_k - s_k = A_1 \eta$ and $c_\eta = A_2 \eta$, where $A_1$ and $A_2$ will be fixed later. Next fix the time $t_{w_k}=(s_k+t_k)/2$. $w_k$ is now determined using
$t_{w_k} - w_k = (1-\lambda) c_\eta$.
A brief calculation shows  that $\alpha_0,\alpha_1$ must be chosen so that
\begin{equation}
\alpha_0,\alpha_1 > \frac{A_1 (1 - \lambda)}{2 A_2},
\label{eqn_alphabnds}
\end{equation}
in order that condition ii) above is met.

As far as i) is concerned, $\Gct$ is an intersection of two sets, a queue constraint and an activity
constraint, so that it is enough to
obtain each as a $\mu$-continuity set. With respect to the queue constraint set, at time
$w_k$, there are uncountably many choices for $\alpha_0,\alpha_1$
which may be taken as close as we like to the constraint in (\ref{eqn_alphabnds}) given
$c_\eta$. For the activity set, we may choose $\nu$ arbitrarily small and
having fixed it, then there are uncountably many choices for $\chi > 0$ which we may also take arbitrarily small. Thus the activity set can also be chosen to be a $\mu$-continuity set. Putting
the above together, $G$ may be constructed as a $\mu$-continuity set for given $A_1,A_2, \eta$
and so that $C_k \subset G$.

We turn to condition iv), where it can be checked that it is satisfied provided that
\begin{equation}
\frac{ \alpha_0+\alpha_1+\lambda \nu}{1-\lambda} < \frac{1}{4}\lb  \frac{1- \alpha_0 - \nu(1-\lambda)}{1-\lambda} \wedge_{\ell=3}^6 \frac{\lambda_\ell}{1-\lambda_\ell} \rb,
\label{eqn_cond47}
\end{equation}
which obviously holds by making $\alpha_0$, $\alpha_1$, and $\nu $ sufficiently small, and then choosing $A_1/A_2$ sufficiently small according to (\ref{eqn_alphabnds}).

As far as iii) is concerned, an examination of the construction preceding Lemma \ref{lem_backclock}
shows that  $\ovt + s_3 - w_k \propto c_\eta = A_2\eta$. We may thus proceed by taking $A_2 > 0$
sufficiently small to ensure that iii) is met and then choose $A_1 > 0$ sufficiently small so as to
meet (\ref{eqn_alphabnds}), (\ref{eqn_cond47}). This fixes $c_\eta,\nu$. The rest follows on choice of
$\alpha_0,\alpha_1,\chi$.

Note that common choices may be made for $A_1,A_2,\nu$ for each $C_k$ and once these are fixed, common values may be chosen for $\alpha_0,\alpha_1,\chi$ as there are uncountably many possibilities and only a countable number of choices can have positive probability for any $C_k$. We have thus shown that
a suitable $\Gct$ can be found for each $k$, given $\eta > 0$.

The rest of the proof follows on summing (\ref{eqn_threeeight}) over $k$, to obtain
\begin{eqnarray}
\sum_k \muu{ \one{(t_Z,h_Z,\eta)}{M_2,{\bf Q}_Z}; F_k} & = &  \frac{3}{8}\sum_k  \muu{F_k} \label{eqn_fin38} \\
 \muu{ \one{(t_Z,h_Z,\eta)}{M_2,{\bf Q}_Z}; F} & = &  \frac{3}{8}\muu{F}, \nonumber
\end{eqnarray}
and similarly for $M_3$ and $M_4$. This is the required result as $F$ is arbitrary.
\end{proof}

%
%

%
%

\subsection*{C.II. Switchover from $M_2, M_3, M_4$}

Here we will only state our results, moreover $M_2$- and $M_3$-periods are analogous and so
we will only deal with the former. To state our theorem for switching out of a $M_2$-period, define $\one{(M_2,\eta)}{M_k,m,n}, k=1,3$  as was done for switching out of $M_1$. Also
define $\FZ{Z^3_{m,n}}{\infty} \doteq \FZ{Z^3_{m,n}}{} \cap \lc Z^3_{m,n} < \infty \rc$.

%
%
\begin{theorem}
\label{thm_switchM2}
$\forall n \in \NN_0, m \in \ZI - \lc 0 \rc$, $\exists \ovp,\ovq > 0, \ovp + \ovq = 1$ and $\exists \eta_m$ such that $\forall \eta,~\eta_m > \eta > 0$
\begin{eqnarray}
\muu{\one{(M_2,\eta)}{M_1,m,n} | \FZ{\Zstop{3}}{\infty}}  & =  & \ovp, ~\mu~a.s.,
\label{eqn_condprobM2M1} \\
\muu{\one{(M_2,\eta)}{M_3,m,n} | \FZ{\Zstop{3}}{\infty}} & =  & \ovq \nonumber.\\ \nonumber
\end{eqnarray}
\end{theorem}

The quantities $\ovp, \ovq$ are determined as follows,
\begin{eqnarray}
\ovp & = & \sum_{X=0}^\infty \sum_{X_4=0}^\infty \sum_{U_4=0,1}
\bth(X) \piinf{4}(X_4,U_4) c_1^X(X_4,U_4) \label{eqn_Lebsesguelim}  \\
\ovq & = & \sum_{X=0}^\infty \sum_{X_4=0}^\infty \sum_{U_4=0,1}
\bth(X) \piinf{4}(X_4,U_4) c_5^X(X_4,U_4) \nonumber
\end{eqnarray}
where $\piinf{4}(X_4,U_4)$ is the equilibrium jump chain probability that node 4 is in state
$(X_4,U_4)$ when operating
in isolation (i.e., when node 4 is the only node in the network). $\bth(X)$ is the limiting probability as $X^3_0 \uparrow \infty$ that a first backoff of node 3 occurs when $X_3=X$, service starting with $X^3_0$ packets. $c_1^X(X_4,U_4)$ is the probability
that node 1 or 2 first gain the medium when node 3 has a first backoff with $X_3=X$ packets
and the state of node 4 as given. The remaining definitions for $\ovq$ are similar. Thus in this
case there is no simple formula and $\ovp, \ovq$ depend on the backoff parameter $\gamma$ as well as the
arrival rates to nodes 3 and 4.

%

For the case of switching out of $M_4$ we have the following result, again making the corresponding
definitions as in Theorem \ref{thm_switchM0}.
\begin{theorem}
\label{thm_switchM4}
For any $\Zstop{(4,M_4^{(5)})}$ stopping time, there is a $\eta_m >0$ sufficiently small
so that, for all $\eta_m > \eta > 0$
\begin{equation}
\muu{\one{(M_2,\eta)}{M_4,m,n} \vee \one{(M_3,\eta}{M_4,m,n)}~~|~\FZ{\Zstop{(4,M_4^{(5)})}}{\infty}} = 1 \,~\mu~ a.s.,
\label{eqn_M4_M2orM3}
\end{equation}
and so that
\begin{eqnarray}
\muu{\one{(M_1,\eta)}{M_4,m,n} | \FZ{\Zstop{(4,M_4^{(5)})}}{\infty}}  & =  & 0, ~\mu~a.s..
\label{eqn_condprobM1M4} \\
\end{eqnarray}
\end{theorem}
A similar result holds for $\Zstop{(5,M_4^{(4)})}$ stopping times.

This concludes Part~C.

%
%
%
%
%
%

\section{Fluid limit proofs: Part D}

In parts A-C we have established a) saw tooth properties  and some constraints on
those sample paths, b) what will occur at the end of a given $M_k$-period, $k=1, 2, 3, 4$ and c) that a natural
state will be entered in finite time before the network can empty almost surely. What has not been shown, is whether any $M_1$-period would
ensue at all. The purpose of this section is to show that $M_1$-periods will
occur $\mu$ a.s. following a natural state, provided $\rho$ is sufficiently
close to 1.

In fact, establishing this result is not strictly needed to prove instability. If there is a last
visit to queues 1 and 2 (which might occur when they are both empty), then these two queues
must grow linearly and therefore the fluid system is unstable. Nevertheless, we will show that an infinite sequence of $M_1$-periods will
occur $\mu$ almost surely and in strictly bounded time, following $T_N$ to enter a non-empty natural state,
%

%
%
Denote by $\Tpot: C[0,\infty] \rightarrow [0,\infty]$ as the first point of increase of either
$I_1, I_2$, as in Definition \ref{defn_poi}, following $T_N$. It is easily shown that $\Tpot$ is a
$\cF_{t^+}$ stopping time, and corresponds to the start of a positive $M_1$-period. Our main result is:
\begin{theorem}
There exists a $0 < T_V <\infty$ such that $\Tpot < T_N+T_V,$ $\mu~a.s$.\\
\end{theorem}
We first show that the issue of occurrence of an $M_1$-period arises only when there is at least one
zero queue. To see this, consider the network at time $T_N$, and, without loss of generality, suppose $T_N > 0$. If $Q_\ell(T_N) = q_\ell > 0,~\ell=1,\cdots,6$ then continuity implies
that this actually holds for some small interval $[T_N-\xi,T_N]$, depending on $\omega$, with
$\xi >0$. During any such interval, one of the periods $M_1,\cdots,M_4$ has to be active, with probability 1, as a consequence of Lemma \ref{lemma_NoDelay}. Moreover if the active period is not $M_1$, then one will follow  in bounded time $\Tpot < T_N+ T_S$, as a consequence of Theorems
\ref{thm_switchM0}, \ref{thm_switchM2} and \ref{thm_switchM4} and because the subnetwork determined
by nodes 3, 4, 5 ,6 is work conserving once a natural state is entered.
Next, consider the cases where a subset of queues $\{q_1,q_2,q_3,q_6\}$ are 0. Rather than deal with all every such cases, we will
consider just one. The approach for the remaining cases will then become apparent. We therefore focus on the case $q_3 = q_4 = q_5 = q_6 = 0$ with
$q_1 \vee q_2 > 0$ and show that $\Tpot$ is effectively $T_N$ on the fluid scale. Other cases are simpler to address.

For the above case, it will be sufficient to prove the following lemma,
\begin{lemma}
\label{lemma_ThereisM1}
Given $\epsilon > 0, \eta > 0$, and $Q^{(\ell)}_L <  Q^{(\ell)}_H, \ell =1,2$ with $ Q^{(1)}_L \vee Q^{(2)}_L > 0 $ there exists $\delta > 0$, such that,
$$
\liminf_{R \rightarrow \infty} \muR{ \Tpot > T_N + \epsilon; Q_\ell^R(T_N) < \delta, \ell =3, 4, 5, 6,
Q^R_\ell(T_N) \in (Q^{(\ell)}_L,Q^{(\ell)}_H), \ell = 1,2 } < \eta.
$$
\end{lemma}
It then follows that $\muu{\Tpot > T_N + \epsilon;~Q_\ell(T_N) = 0, \ell=3, 4, 5, 6; Q_\ell \in
(Q^{(\ell)}_L,Q^{(\ell)}_H), \ell = 1,2 } = 0$ by Theorem 2.1 of \cite{Billingsley68}. This implies our
result since $\epsilon >0$ was arbitrary.

The proof of Lemma \ref{lemma_ThereisM1} is strongly specific to our choice of network and so we only provide a sketch of the proof, which relies on the following definition.
Consider the sequence of queue activations in the
network and in particular node 4 and 5 activations.
\begin{definition}
A control step, at index $k$ of the jump chain, is $C_{J,k}
\in \lc 4,5 \rc$. At index 0, it is 4, unless only node 5 is active, in
which case it is 5. At any subsequent index $k\geq 1$, the control step is determined as follows: (i) if both nodes 4 and 5 are active at index $k$, $C_{J,k}$ is the node that has become active first, (ii) If only one is active, $C_{J,k}$ is the active node, (iii) if neither, then $C_{J,k}$  is the last node that was active.

A control swap to control step 5 occurs at step $k$ if $C_{J,k-1}=4$ and $C_{J,k} = 5$, and vice versa for a control swap to control step 4. These events are denoted as
$4 \rightarrow 5$ and $5 \rightarrow 4$. The step at which the $r$-th control swap takes place is a (discrete) stopping time $\sW_r$.
\end{definition}

For the subnetwork of nodes 3, 4, 5, 6 in isolation, it is readily shown that the probability is
1 that an infinite number of control swaps occur, $\sW_r < \infty$,
$  \sW_1 < \sW_2 < \ldots$, with corresponding filtration $\lc {\cal G}_r \rc_{r \in \NN}$.

Next let $S$ be the stopping time until one of nodes 1 and 2 gain the medium (this event is blocked
since we are considering the subnetwork in isolation). Using the properties of our network it
can then be shown that,
\begin{lemma}
\label{lemma_victory}
$\exists \epsilon > 0$ and $N_S \in \NN$ such that $\forall r \in \NN$, $ \prob{ S < r + N_S~|~{\cal G}_r} > \epsilon$.\\
\end{lemma}
The proof relies on showing that once a control swap has taken place say $4 \rightarrow 5$
then within a bounded number of additional control swaps either it will occur that node 3 has a backoff with $Q_4 = 0$ or node 6 has a backoff with $Q_5$ = 0 and with probability at least
$\epsilon > 0$.

The result of Lemma \ref{lemma_victory} implies that $m_S = \expect{S} < \infty$ and actually that
$\prob{ S > rN_S} < \lb 1 - \epsilon \rb^r$, see \cite{Williams91} for example. That is
nodes 1 and 2 will gain the medium within a number of control swaps which has finite expectation.
It then follows from Markov's inequality, that given any $\eta > 0$, there is a number of control swaps
$M_\eta$ such that $\prob{ S > M_\eta} < \eta$.

Given the above, in proving Lemma \ref{lemma_ThereisM1} it will be enough to show
the following lemma.
\begin{lemma}
\label{lemma_Meta}
Suppose that given any $\delta_1 > 0$, the sequence of initial conditions (starting from 0)
for the jump chain  of the subnetwork, determined by nodes $3, 4, 5, 6$, satisfies
$Q_\ell(0) < \delta_1 R$, for $R$ sufficiently large.

Furthermore, given $\epsilon > 0, \eta > 0$ and any fixed number of control swaps $M_\eta \in \NN$, let
$X_{M_\eta}$ be the total number of steps to complete $M_\eta$ control swaps. Then $\exists \delta > 0$
such that for all $R$ sufficiently large
$$
\prob{ X_{M_\eta} > \ceil{R \epsilon} |~Q_\ell(0) \leq \delta R, \ell=3,\cdots,6 } < \eta.\\
$$
\end{lemma}

The result of the lemma clearly relies on the supposed initial conditions. It can
be demonstrated via a construction. The construction works by determining
$M_\eta$ intervals so that the probability of a control swap in each is close to 1
and so the total queue length at the start of each interval is small based on the arrivals
which might have taken place during previous intervals. The proof is then completed by showing that given any control step at the start of an interval, say 5, a swap $5 \rightarrow 4$ will occur with high probability in a number of steps in
proportion to the initial condition for the interval. Since only $M_\eta$
such intervals are required, a $\delta$ can be obtained accordingly.


Other cases are dealt with similarly but there is some positive but
bounded delay before an $M_1$-period occurs. For example if $q_1= q_2 = q_6 = q_5 = 0$ and
$q_3 > 0$, then one shows that a $M_2$-period occurs in negligible fluid time.
Thus in all the cases, we may can show that $M_1$-periods occur within bounded fluid time following a natural state. This concludes Part~D.

\section{Additional proofs}

\subsection*{E.I. Proof of Lemma~\ref{auxi1}}

The proof relies on basic sample path properties of the fluid limit
process $\{Q(t)\}$ as described in Subsection~\ref{desc}.
First of all, the $M_1$-period that initiates the $i$-th cycle ends
at time $t_i + T_{i1}$, with
\[
T_{i1} = \max\left\{\frac{Q_1(t_i)}{1 - \rho_1},
\frac{Q_2(t_i)}{1 - \rho_2}\right\} \leq
\frac{\max\{Q_1(t_i), Q_2(t_i)\}}{1 - \rho_0} \leq
\frac{\Load(t_i)}{1 - \rho_0}.
\]

Define $K(t) = \max\{Q_3(t), Q_4(t)\} + \max\{Q_5(t), Q_6(t)\}$
and recall that $\rho = \rho_0+\rho_3+\rho_6$.
Then
\ben
K(t_i + T_{i1}) & \leq & K(t_i) + (\rho_3 + \rho_6) T_{i1} \\
&\leq & \Load(t_i) - \max\{Q_1(t_i), Q_2(t_i)\} + (\rho_3 + \rho_6)
\frac{\max\{Q_1(t_i), Q_2(t_i)\}}{1 - \rho_0}\\
&=& \Load(t_i) - \frac{(1 - \rho) \max\{Q_1(t_i), Q_2(t_i)\}}{1 - \rho_0} \\
&=& \Load(t_i) - (1 - \rho) T_{i1},
\een
which may also be seen from the fact that $\Load(t)$ decreases at
a rate $1 - \rho$ or larger during the time interval
$[t_i, t_i + T_{i1}]$ and $K(t_i + T_{i1}) = \Load(t_i + T_{i1})$
since $Q_1(t_i + T_{i1}) = Q_2(t_i + T_{i1}) = 0$.

Define $T_0 = \frac{K(t_i + T_{i1})}{1 - \rho_3 - \rho_6}$.
We distinguish between two cases, depending on whether an
$M_4$-period starts before time $t_i + T_{i1} + T_0$ or not.

If no $M_4$-period occurs before time $t_i + T_{i1} + T_0$, then
$K(t)$ decreases at a rate $1 - \rho_3 - \rho_6$ or larger for all
$t \in [t_i + T_{i1}, t_i + T_{i1} + T_0]$ and reaches zero no later
than time $t_i + T_{i1} + T_0$, unless an $M_1$-period intervenes.
This implies that the next $M_1$-period must start no later than
time $t_i + T_{i1} + T_0$.

Using the above results, a simple calculation shows that
\[
\Delta t_i \leq T_{i1} + T_0  \leq
T_{i1} + \frac{\Load(t_i) - (1 - \rho) T_{i1}}{1 - \rho_3 - \rho_6}
 \leq
\frac{\Load(t_i)}{(1 - \rho_0) (1 - \rho_3 - \rho_6)} \leq C_T \Load(t_i).
\]
Also, $\Load(t)$ has continuously decreased during the cycle,
so $\Load(t_{i+1})-\Load(t_i) \leq 0$.

Now suppose that an $M_4$-period does start at some time
$t_0 \in [t_i + T_{i1}, t_i + T_{i1} + T_0]$, and ends at time~$u_0$.

Since $K(t)$ decreases at a rate $1 - \rho_3 - \rho_6$ or larger
during the time interval $[t_i + T_{i1}, t_0]$, it follows that
\[
K(t_0) \leq K(t_i + T_{i1}) - (1 - \rho_3 - \rho_6) (t_0 - t_i - T_{i1}).
\]

Noting that $Q_4(t_0), Q_5(t_0) \leq K(t_0)$, we conclude that the
duration of the $M_4$-period is no longer than
\[
u_0 - t_0 \leq
\min\left\{\frac{Q_4(t_0)}{1 - \rho_4}, \frac{Q_5(t_0)}{1 - \rho_5}\right\} \leq
\frac{K(t_0)}{1 - \max\{\rho_4, \rho_5\}}.
\]

Since $K(t)$ increases at a rate no larger than $\rho_3 + \rho_6$
during the time interval $[t_0, u_0]$, it follows that
\[
K(u_0) \leq K(t_0) + (\rho_3 + \rho_6) (u_0 - t_0).
\]

The $M_4$-period will cause the queue of node~4 to empty at some
point and become smaller than the queue of node~3, and likewise the
queue of node~5 must empty at some point and become smaller than
the queue of node~6.
Because $M_4$-periods can no longer be initiated from $M_2$ and $M_3$,
$K(t)$ decreases at a rate $1 - \rho_3 - \rho_6$ or larger from
time~$u_0$ onward, and reaches zero no later than time
$u_0 + \frac{K(u_0)}{1 - \rho_3 - \rho_6}$, unless an $M_1$-period
intervenes.
This implies that the next $M_1$-period must start no later than
time $u_0 + \frac{K(u_0)}{1 - \rho_3 - \rho_6}$.

Combining the above results, we obtain
\ben
\Delta t_i &\leq &u_0 + \frac{K(u_0)}{1 - \rho_3 - \rho_6} - t_i =
T_{i1} + (t_0 - t_i - T_{i1}) + (u_0 - t_0) +
\frac{K(u_0)}{1 - \rho_3 - \rho_6} \\
&\leq& T_{i1} + (t_0 - t_i - T_{i1}) + \frac{K(t_0) + (u_0 - t_0)}{1 - \rho_3 - \rho_6} \\
&\leq& T_{i1} + (t_0 - t_i - T_{i1}) +
\left(1 + \frac{1}{1 - \max\{\rho_4, \rho_5\}}\right)
\frac{K(t_0)}{1 - \rho_3 - \rho_6} \\
&\leq & T_{i1} - \frac{t_0 - t_i - T_{i1}}{1 - \max\{\rho_4, \rho_5\}} +
\frac{(2 - \max\{\rho_4, \rho_5\}) K(t_i + T_{i1})}
{(1 - \max\{\rho_4, \rho_5\}) (1 - \rho_3 - \rho_6)} \\
&\leq& T_{i1} + \frac{(2 - \max\{\rho_4, \rho_5\}) (\Load(t_i) - (1 - \rho) T_{i1})}
{(1 - \max\{\rho_4, \rho_5\}) (1 - \rho_3 - \rho_6)} \\
&\leq& \frac{(2 - \max\{\rho_4, \rho_5\}) \Load(t_i) +
\rho_0 (1 - \max\{\rho_4, \rho_5\}) T_{i1})}
{(1 - \max\{\rho_4, \rho_5\}) (1 - \rho_3 - \rho_6)} \\
&=& \frac{\Load(t_i)}{(1 - \max\{\rho_4, \rho_5\}) (1 - \rho_3 - \rho_6)} +
\frac{\Load(t_i) + \rho_0 T_{i1}}{1 - \rho_3 - \rho_6} \\
& \leq & \frac{\Load(t_i)}{(1 - \max\{\rho_4, \rho_5\}) (1 - \rho_3 - \rho_6)} +
\frac{\Load(t_i)}{(1 - \rho_0) (1 - \rho_3 - \rho_6)} \\
&=& \frac{\Load(t_i)}{1 - \rho_3 - \rho_6}
\left(\frac{1}{1 - \rho_0} + \frac{1}{1 - \max\{\rho_4, \rho_5\}}\right) =
C_T \Load(t_i).
\een
Also, $\Load(t)$ has only increased during the $M_4$-period at a rate
no larger than $\rho=\rho_0+\rho_3 + \rho_6$, so
\[
\Load(t_{i+1})-\Load(t_i) \leq \rho (u_0 - t_0) \leq
\frac{\rho K(t_0)}{1 - \max\{\rho_4, \rho_5\}} \leq
\frac{\rho \Load(t_i)}{1 - \max\{\rho_4, \rho_5\}} =
C_{\Load} \Load(t_i).
\]

\subsection*{E.II. Proof of Lemma~\ref{auxi2}}

Denote by $t_1$ and $t_2$ the times that the cycles start and by
$u_1$ and $u_2$ the times that the $M_1$-periods end.
First assume $\max\{Q_1(t_1), Q_2(t_1)\} \leq \epsilon \Load(t_1)$.
Then, $\max\{Q_3(t_1), Q_4(t_1)\} + \max\{Q_5(t_1), Q_6(t_1)\} \geq
(1 - 2 \epsilon) \Load(t_1)$, so we must have
$\max\{Q_3(t_1), Q_4(t_1)\} \geq (1 - 2 \epsilon) \Load(t_1) / 2$ or
$\max\{Q_5(t_1), Q_6(t_1)\} \geq (1 - 2 \epsilon) \Load(t_1) / 2$.
In the former scenario, with probability 3/8 the $M_1$-period is
followed by an $M_2$-period, which will last for an amount of time
no less than
$\max\{\frac{Q_3(t_1)}{1 - \rho_3}, \frac{Q_4(t_1)}{1 - \rho_4}\} \geq
\frac{\max\{Q_3(t_1), Q_4(t_1)\}}{1 - \rho_4} \geq
\frac{(1 - 2 \epsilon) \Load(t_1)}{2 (1 - \rho_4)}$.
Likewise, in the latter scenario, with probability 3/8 the
$M_1$-period is followed by an $M_3$-period, which will last for
an amount of time no less than
$\max\{\frac{Q_5(t_1)}{1 - \rho_5}, \frac{Q_6(t_1)}{1 - \rho_6}\} \geq
\frac{\max\{Q_5(t_1), Q_6(t_1)\}}{1 - \rho_5} \geq
\frac{(1 - 2 \epsilon) \Load(t_1)}{2 (1 - \rho_5)}$.
Thus, in either scenario, with probability at least 3/8, the time
until the start of the next cycle is at least
$\frac{(1 - 2 \epsilon) \Load(t_1)}{2 (1 - \min\{\rho_4, \rho_5\})}$,
so that
\[
\max\{Q_1(t_2), Q_2(t_2)\} \geq Q_2(t_2) \geq
\frac{\rho_2 (1 - 2 \epsilon) \Load(t_1)}{2 (1 - \min\{\rho_4, \rho_5\})}.
\]
Invoking the fact that $\Load(t_2) \leq C_{\Load} \Load(t_1)$,
with $C_{\Load}$ as defined in the previous lemma, we find that
\[
\max\{Q_1(t_2), Q_2(t_2)\} \geq \epsilon \Load(t_2),
\]
with $\epsilon$ as specified in the statement of the lemma.

Now consider a cycle with
$\max\{Q_1(t_k), Q_2(t_k)\} \geq \epsilon \Load(t_k)$, $k = 1, 2$.
Then
\[
Q_i(u_k) =
Q_i(t_k) + \rho_i \frac{\max\{Q_1(t_k), Q_2(t_k)\}}{1 - \rho_2},
\mbox{ for } i = 3, 4, 5, 6.
\]
Note that $0 \leq Q_i(t_k) \leq (1-\epsilon) \Load(t_k)$, $i = 3, 4, 5, 6$,
and $\epsilon \Load(t_k) \leq \max\{Q_1(t_k), Q_2(t_k)\} \leq \Load(t_k)$.
Then it is easily verified that the queues are weakly balanced at
time~$u_k$ with $\beta^{\min}$ and $\beta^{\max}$ as given in the
statement of lemma.

\subsection*{E.III. Proof of Theorem~\ref{main}}

Let $(U(n),\Qorig(n))$ denote the jump chain obtained from the
continuous-time Markov process by uniformization according to
a Poisson clock of rate~$\beta$ as described in Appendix~A.I.
In order to prove Theorem~\ref{main} for the original stochastic process,
it suffices to establish a similar result for the jump chain:
\be
\lim _{\|\Qorig(0)\| \to \infty}
\PP_{\Qorig(0)} \{\liminf_n \|\Qorig(n)\| = \infty\} = 1.
\label{jumpunstable1}
\ee

In order to apply Theorem~\ref{meyn}, consider the function
$W(x) = \expect{\WW|\Qorig(0)=x}$, where the random variable $\WW$ is
defined as
\[
\WW :=
\sum_{n=0}^{\|\Qorig(0)\| T} [1 + \|\Qorig(0)\| + a \|\Qorig(n)\|]^{- m}
\]
for some positive constants~$a$ and~$T$ to be determined later and $m >1$.
Note that, with minor abuse of notation, $W(\Qorig(0)=x, U(0)=u)=W(x)$,
i.e., $W$ only depends on the queue and not on the activity vector.
The function $W(x)$ may be interpreted as the following approximation to
a Lyapunov function for the fluid limit process
\be
\|x\|^{m-1}W(x) \approx
\expectx{\int_0^T (1+a \|Q_{\hat{x}}(t/\beta)\|)^{-m} \dd t} = V(Q_{\hat{x}}(t)),
\label{eq:approx}
\ee
with equality when $\|x\| \to \infty$, and $\hat{x}=\frac{x}{\|x\|}$
is the initial state of the fluid limit process.
Then it follows from the instability of the fluid limit process that
we can choose $a$ and $T$ large enough such that
$V(Q_{\hat{x}}(t+r))< V(Q_{\hat{x}}(t))$ for any $r>0$ and any
initial state $\hat{x}$.
This implies that
\ben
\|x\|^m \expect{W(\Qorig(n+1))-W(\Qorig(n))|\FF_n} \leq - \mbox{constant}
\een
when $x=\Qorig(n)$ and $\|x\|$ is sufficiently large.
Thus, we can apply Theorem~\ref{meyn}.

The detailed arguments may be described as follows.
First of all, note that
$$\expect{W(X(1)) - W(X(0))|\Qorig(0)=x, U(0)=u} =\expect{\theta^1 \WW - \WW|\Qorig(0)=x},$$
where $\theta^1$ is the usual
backward shift operator on the sample path space~\cite{Meyn95}.
We write $\theta^1 \WW - \WW = A+B+C$, where
\ben
A = -[1+\|\Qorig(0)\|+a\|\Qorig(0)\|]^{-m},
\een
\ben
B = \sum_{n=1}^{\|\Qorig(0)\| T}
\left\{[1 + \|\Qorig(1)\| + a \|\Qorig(n)\|]^{- m}\right. \\
- \left.[1 + \|\Qorig(0)\| + a \|\Qorig(n)\|]^{- m}\right\},
\een
and
\ben
C = \sum_{n=\|\Qorig(0)\|T+1}^{\|\Qorig(1)\|T}
[1 + \|\Qorig(1)\| + a\|\Qorig(n)\|]^{- m}.
\een
The term~`$A$' provides the negative drift and the other terms can
be bounded as follows.
Using the fact that $\|\Qorig(1)\| \geq \|\Qorig(0)\| - 1$,
and noting that $[\cdot]^{- m}$ is a convex decreasing function, we have
\be
B \leq
\sum_{n=1}^{\|\Qorig(0)\| T} m [\|\Qorig(0)\| + a \|\Qorig(n)\|]^{- m - 1}.
\ee
Multiplying both sides by $\|\Qorig(0)\|^m$, we see that
\be
\|\Qorig(0)\|^m B \leq \frac{m}{\|\Qorig(0)\|} \sum_{n=1}^{\|\Qorig(0)\| T}
\left(1 + a \frac{\|\Qorig(n)\|}{\|\Qorig(0)\|}\right)^{- m - 1}
\label{form1}
\ee
Let $\Qorig(0) = x$ and $\hat{x} := x/\|x\|$.
For any $x$, the random variable in the right-hand-side (RHS)
of~(\ref{form1}) is bounded by $m T$, and hence
\ben
\limsup_{\|x\| \to \infty} \expectx{\|x\|^m B} \leq
\expectx{m \int_0^T [1 + a \|Q(s/\beta)\|]^{- m - 1} \dd s}
\label{form2}
\een
because of the weak limit convergence of $\frac{1}{\|x\|}\Qorig^{(\|x\|)}(\|x\|t) \Rightarrow Q(t/\beta)$ over $[0, T]$ and uniform integrability of
the random variables of the form RHS of~(\ref{form1}).

Next, for `C', it is sufficient to consider the case that
$\|\Qorig(1)\| = \|\Qorig(0)\| + 1$, where
\[
C \leq T [1 + \|x\| + 1 + a (\|\Qorig(\|x\|T)\|-T)]^{- m}.
\]
Similarly to~`B', multiplying both sides with $\|x\|^m$ and taking
the limit gives
\be
\limsup_{\|x\| \to \infty} \expectx{\|x\|^m C} \leq
\expectx{T [1 + a \|Q(T/\beta)\|]^{- m}},
\label{form3}
\ee
again, because $\|x\|^m C < T$ (thus, uniform integrability holds)
and by the weak limit convergence. Putting the bounds together, we obtain
\ben
\limsup_{\|x\| \to \infty} \|x\|^m \expectx{\theta^1 \WW - \WW} &\leq& - (1 + a)^{- m} + m \expectx{\int_0^\infty (1 + a L(s/\beta))^{- m - 1} \dd s} \\
&& +\expectx{T (1 + a L(T/\beta))^{- m}},
\een
because $\|Q(s)\| \geq L(s)$ based on our notation with
some initial state $Q(0) = \hat{x}$ such that $\|\hat{x}\| = 1$.
Consider the cycle pairs $D_k$, $k = 1, 2, \dots$, as defined for
Theorem~\ref{theorem:rho<1}.
Then,
\ben
\expectx{\int_0^\infty (1 + a L(s/\beta))^{- m - 1} \dd s}& \leq & \beta \expectx{\sum_{k=0}^\infty \int_{T_k}^{T_{k+1}}(1 + a L(s))^{- m - 1} \dd s } \\
& \leq &\beta \expectx{\sum_{k=0}^\infty \int_{T_k}^{T_{k+1}}(1 + a \theta L_k)^{- m - 1} \dd s } \\
& \leq & \beta \expectx{\sum_{k=0}^\infty \Delta T_k (a \theta L_k)^{- m - 1}} \\
&  \leq & \beta C_{LT} (a \theta)^{- m - 1} \sum_{k=0}^\infty \expectx{L_k^{- m}}.
\een
Note that the times $T_k$ are random variables in general and we have
used the fact that $L_{k+1} \geq \theta L_k$ with $0 < \theta \leq 1$.
As we saw in the proof of Theorem~\ref{expectedrate},
for $\rho \in (\rho^*, 1]$, $\expect{L_k^{- m}} \leq \alpha^k$.
Therefore,
\be
m \expectx{\int_0^\infty (1 + a L(s/\beta))^{- m - 1} \dd s} \leq
m \beta C_{LT} (a \theta)^{- m - 1} \frac{1}{1 - \alpha}.
\label{form4}
\ee
So, we can choose $a$ large enough to ensure that the RHS
of~(\ref{form4}) is less than $\frac{1}{3} (1 + a)^{-m}$.
Next we show that we can choose $T$ large enough such that
\be
\expectx{T [1 + a L(T/\beta)]^{- m}} \leq \frac{1}{3} (1 + a)^{- m}.
\label{form7}
\ee
Note that
\be
\expectx{T [1 + a L(T/\beta)]^{- m}} \leq a^{- m} \expectx{T L^{-m}(T/\beta)},
\ee
and by~Theorem~\ref{expectedrate},
$\limsup_{T \to \infty} \expectx{T L^{- m}(T)} = 0$,
for $\rho \in (\rho^*, 1]$.
Hence, we can choose $T$ large enough such that (\ref{form7}) holds.

Therefore,
\[
\limsup_{\|x\| \to \infty} \|x\|^m \expect{W(X(1))-W(X(0))|(X(0),U(0))=(x, u)} \leq
- \frac{1}{3}(1 + a)^{- m}.
\]
This means that there exists a a positive constant $\|x_0\|$ such that,
\[
\expect{W(X(1))-W(X(0))|X(0)=(x, u)} \leq
-\frac{1}{6} (1 + a)^{- m} \|x\|^{- m},
\]
whenever $\|x\| > \|x_0\|$.
Let $c_0 = W(x_0) = W(\|x_0\|)$.
On the other hand, it follows from~(\ref{eq:approx}) that
$\limsup_{\|x\| \to \infty} W(x) = 0$, which means that $A_{c_0}$
is well-defined and also $c_0$ can be made arbitrary small by
letting $\|x_0\| \to \infty$.
Therefore, the conditions of Theorem~(\ref{meyn}) are satisfied
with $\Delta(x) = \frac{1}{6}(1 + a)^{-m} \|x\|^{- m}$.
This shows that
\be
\PP_{\Qorig(0)} \left\{\sum_{n=0}^\infty
\frac{\mbox{constant}}{\|\Qorig(n)\|^m} < \infty\right\} \to 1,
\ee
as $\|\Qorig(0)\| \to \infty$, which implies~(\ref{jumpunstable1}).

\end{document}